\documentclass{IEEEtran} 
\usepackage{ifpdf}
\ifpdf
\usepackage[pdftex]{graphicx}
\usepackage[update]{epstopdf}
\else
\usepackage{graphicx}
\fi

\usepackage{tikz}
\usepackage{amsmath,amssymb,amsthm,bm,mathtools}  
\usepackage{soul} 
\usepackage{tabularx}
\usepackage{arydshln} 
\usepackage{enumerate}
\usepackage{balance}   
\usepackage[hidelinks]{hyperref}  
\usepackage{cleveref} 
\usepackage{multirow}
\usepackage{subcaption}
\usepackage{cite}
\usepackage{url}
\allowdisplaybreaks[3]  
\usepackage[normalem]{ulem} 
\usepackage{hhline}
\usepackage{diagbox}  
\DeclareGraphicsExtensions{.pdf,.png,.jpg,.eps}

\begin{document}
	\title{Weakly Secure Symmetric Multilevel Diversity Coding}   
	\author{Tao Guo, {\it Member, IEEE}, Chao Tian, {\it Senior Member, IEEE}, Tie Liu, {\it Senior Member, IEEE}, \\ and Raymond W. Yeung, {\it Fellow, IEEE}
		\thanks{The work of C. Tian was supported in part by the National Science Foundation under Grant CCF-18-32309 and CCF-18-16546. 
			This paper was presented in part at 2019 IEEE Information Theory Workshop (ITW).
		
		T. Guo was with the Department of Electrical and Computer Engineering, Texas A\&M University, College Station, TX, USA. He is now with the Department of Electrical and Computer Engineering, the University of California, Los Angeles, CA, USA. (e-mail: guotao@ucla.edu)
		
		C. Tian and T. Liu are with the Department of Electrical and Computer Engineering, Texas A\&M University, College Station, TX, USA. (e-mail: chao.tian@tamu.edu, tieliu@tamu.edu)
		
		R. W. Yeung is with the Institute of Network Coding and the Department of Information Engineering, The Chinese University of Hong Kong, Hong Kong (e-mail: whyeung@ie.cuhk.edu.hk)
		
		}}
	\maketitle
	\newcommand{\reffig}[1]{Fig. \ref{#1}}
	\newcommand{\cA}{\mathcal{A}}
	\newcommand{\cB}{\mathcal{B}}
	\newcommand{\cD}{\mathcal{D}}
	\newcommand{\cE}{\mathcal{E}}
	\newcommand{\cI}{\mathcal{I}}
	\newcommand{\cK}{\mathcal{K}}
	\newcommand{\cL}{\mathcal{L}}
	\newcommand{\cN}{\mathcal{N}}
	\newcommand{\cO}{\mathcal{O}}
	\newcommand{\cQ}{\mathcal{Q}}
	\newcommand{\cR}{\mathcal{R}}
	\newcommand{\cT}{\mathcal{T}}
	\newcommand{\cU}{\mathcal{U}}
	\newcommand{\cV}{\mathcal{V}}
	\newcommand{\cX}{\mathcal{X}}
	\newcommand{\fD}{\mathfrak{D}}
	\newcommand{\bl}{\bm{\lambda}}
	\newcommand{\bB}{\mathbb{B}}
	\newcommand{\bF}{\mathbb{F}}
	\newcommand{\sm}{\mathsf{m}}
	\newcommand{\sR}{\mathsf{R}}
	
	\captionsetup[subtable]{labelformat=simple, labelsep=space}
	\renewcommand\thesubtable{(\alph{subtable})}
	
	\theoremstyle{plain}
	\newtheorem{theorem}{Theorem}
	\newtheorem{lemma}{Lemma}
	\newtheorem{corollary}{Corollary}[theorem]
	\newtheorem{proposition}{Proposition}
	\newtheorem{conjecture}{Conjecture}
	\newtheorem{claim}{Claim}

\newtheorem*{namedthm*}{\thistheoremname}
\newcommand{\thistheoremname}{} 
\newenvironment{namedthm}[1]
{\renewcommand{\thistheoremname}{#1}\begin{namedthm*}}
	{\end{namedthm*}}
	
	\theoremstyle{remark}
	\newtheorem{remark}{Remark}
	
	\theoremstyle{definition}
	\newtheorem{definition}{Definition}
	\newtheorem{example}{Example}


\begin{abstract}
	Multilevel diversity coding is a classical coding model where multiple mutually independent information messages are encoded, such that different reliability requirements can be afforded to different messages. It is well known that {\em superposition coding}, namely separately encoding the independent messages, is optimal for symmetric multilevel diversity coding (SMDC) (Yeung-Zhang 1999). In the current paper, we consider weakly secure SMDC where security constraints are injected on each individual message, and provide a complete characterization of the conditions under which superposition coding is sum-rate optimal. Two joint coding strategies, which lead to rate savings compared to superposition coding, are proposed, where some coding components for one message can be used as the  encryption key for another. By applying different variants of Han's inequality, we show that the lack of  opportunity to apply these two coding strategies directly implies the optimality of superposition coding. It is further shown that under a set of particular security constraints, one of the proposed joint coding strategies can be used to construct a code that achieves the optimal rate region.
\end{abstract}
\section{Introduction}\label{section-intro}
Symmetric multilevel diversity coding (SMDC) was introduced by Roche \textit{et al.} \cite{yeung97} 
for applications in distributed data storage and robust network communication. 
Albanese {\it et al.} \cite{PEC-96} independently studied the problem of {\it priority encoding transmission} (PET), which shares the same mathematical model as SMDC. 
In a symmetric $L$-level diversity coding system, there are $L$ independent messages $(M_1,M_2,\ldots,M_L)$, 
where the importance of messages decreases with the subscript~$l$. 
The messages are encoded by $L$ encoders. 
There are totally $2^L-1$ decoders, each of which has access to the outputs of a distinct subset of the encoders. 
A decoder which can access any $\alpha$ encoders, called a Level-$\alpha$ decoder, is required to reconstruct the first $\alpha$ most important messages. 
The system is symmetric in the sense that the reconstruction requirement of a decoder depends on the set of encoders it can access only via its cardinality. 

It was shown \cite{yeung97,yeung99} that separately encoding these independent messages, 
referred to as {\it superposition coding}, is optimal in terms of achieving the entire rate region. 
The characterization of the coding rate region therein involves implicit and uncountably many inequalities, 
and an explicit characterization of the coding rate region was recently obtained \cite{guo-yeung-SMDC-IT20}. 
The problem has also been extended and generalized, e.g., to allow node regeneration \cite{tian2016multilevel} 
and to allow asymmetric decoders \cite{mohajer-tian-diggavi10}. 
Li {\it et~al.}~\cite{Congduan-16} studied the multilevel diversity coding problem with at most 3 sources and 4 encoders 
in a systematic way and obtained the exact rate region of each of the over 7,000 instances with the aid of computation. 

The SMDC problem with a {\em strong} security guarantee was considered by Balasubramanian et al. \cite{liutie13-sSMDC} and Jiang et al. \cite{liutie14-s-all}. In this setting, a security threshold $N$ is given, and the first $N$ messages are degenerate. For the remaining $L-N$ messages $M_\alpha$, $\alpha = N+1,N+2,\cdots, L$, in addition to the standard multilevel reconstruction requirement, it is also required that all these messages need to be kept perfectly {\em jointly} secure if no more than $N$ encoders are accessible by an eavesdropper. Despite the additional security constraints, it was shown that superposition coding remains to be optimal in terms of both the sum rate \cite{liutie13-sSMDC} and the entire rate region \cite{liutie14-s-all}.

In this paper we consider a {\em weakly} secure setting of the classical SMDC problem, where the security level of each message is specified by a separate security parameter $N_\alpha$. More specifically, for any $\alpha=1,2,\ldots,L$, we require the message $ M_\alpha$ to be kept perfectly secure if the outputs of no more than $N_\alpha$ encoders are accessible by an eavesdropper. Such a security requirement is ``weak" in the sense that the eavesdropper is only prevented from obtaining any information about the {\em individual} messages. By comparison, the security requirement of \cite{secure-NC,liutie13-sSMDC,liutie14-s-all} is strong in that it prevents the eavesdropper to obtain any information about the {\em entire} set of messages. The notion of weak security has been considered in various network coding settings \cite{guo-tly-sSMDC-itw2019,Weakly-NC-Alex11,Weakly-NC-Alex13,Weakly-NC-Alex14} and also channel coding perspectives \cite{secure-channel1,secure-channel2,secure-channel3,secure-channel4} in the literature and is generally considered to be more practical for protecting individual messages. 
For example, when the messages are video sequences, the user should not obtain information about any individual video segment, but obtaining the binary XOR of two video sequences may not be an issue since it will not lead to a meaningfully decodable video sequence. Moreover, a protocol with a weak security constraint can potentially be implemented more efficiently in practical settings and may not require encryption keys.
Note that the notion of ``weak/strong security" here is different from the asymptotic notion of weak/strong security in \cite{asymptotic-weak-1,asymptotic-weak-2,asymptotic-weak-3,asymptotic-weak-4}, wherein asymptotic weak security requires vanishing of the information leakage rate and the corresponding strong security requires the vanishing of leaked information content. 
Another notion of  ``weak security" is defined in \cite{Weakly-NC-Narayanan05} which requires the eavesdropper to be unable to obtain any meaningful information about the source. 

On the one hand, the notion of weak security has significantly enriched the collection of secure SMDC problems: Unlike the strongly secure setting where a single security parameter is set for all the messages, for the weakly secure setting, a different security parameter can be set for each message. On the other hand, the notion of weak security has also cast the optimality of superposition coding in much greater doubt, as requiring the messages to be protected only {\em marginally} (instead of jointly) significantly opens up the set of feasible coding strategies. The main goals of this paper are: 1) to understand under what configurations of the security parameters $(N_1,N_2,\ldots,N_L)$ superposition coding remains to be optimal; and 2) to identify optimal coding strategies when superposition coding is suboptimal.

The main message of this paper is that the optimality of superposition coding depends {\em critically} on the security parameters $(N_1,N_2,\ldots,N_L)$. 
More specifically, we consider a natural joint coding strategy that encodes a pair of messages together by using one of the messages as part of the secret key for securing the other. 
We term this coding strategy {\em pairwise encoding}, and \Cref{section-scheme-alpha-key-beta,section-scheme-beta-key-alpha} discuss two scenarios for which pairwise encoding is possible. 
The main results of the paper are:
\begin{itemize}
	\item[1)] We show that superposition coding can achieve the minimum sum rate whenever pairwise encoding is not possible between any two messages. 
	This immediately leads to a necessary and sufficient condition on the security parameters $(N_1,N_2,\cdots,N_L)$ for superposition coding to be optimal in terms of minimizing the sum rate.
	\item[2)] We consider a special class, referred to as {\it differential-constant secure SMDC} (DS-SMDC), for which the more important messages are maximally protected ($N_\alpha=\alpha-1$) and the less important messages are not protected at all ($N_\alpha=0$), and show that a simple extension of the pairwise encoding strategy (from a pair of messages to a pair of groups of messages and hence termed as {\em group pairwise encoding}) can achieve the entire rate region.
\end{itemize}

Note that the min-cut capacity for multicasting a single source is achievable using linear network codes \cite{Network-Coding-IT2000,Linear-NC,secure-NC-IT11}. It was shown in~\cite{Weakly-NC-Narayanan05} that the min-cut bound can also be achieved for a single-source secure network coding model, where the security measure is similar to the weak security notion we used in this work. 
However, the min-cut bound may not be achievable for general multi-source network coding problems (even without any security measure), e.g., the example illustrated by Fig. 21.3 in~\cite{raymondbook}. In particular, the min-cut bound is not achievable for the secure SMDC problem here. 

The rest of the paper is organized as follows. We first formulate the problem and state some preliminary results in \Cref{section-formulation}. 
In \Cref{section-main}, we state the main results, i) a precise classification of the cases where superposition is sum-rate optimal;
ii)  the optimal rate region for DS-SMDC. 
In \Cref{section-joint-strategy,section-conditions-converse}, we describe the pairwise encoding strategies that reduce coding rates 
and prove the optimality of superposition under the conditions in i). 
\Cref{section-DS-proof} is devoted to the proof of the optimal rate region for DS-SMDC. 
We conclude the paper in \Cref{section-conclusion}. Some technical proofs can be found in the appendices.

\section{Problem Formulation and Preliminaries}\label{section-formulation}
\subsection{Problem Formulation}
Let $\cL\triangleq \{1,2,\cdots,L\}$, where $L\geq2$. 
Let $M_1,M_2,\cdots,M_L$ be a collection of $L$ mutually independent messages uniformly distributed over the direct product of certain finite sets.  For simplicity, we assume the message set to be $\bF_{p^{m_1}}\times\bF_{p^{m_2}}\times\cdots\times\bF_{p^{m_L}}$, where $\bF_{p^{m_1}}$ is a finite field of order $p^{m_1}$ and $p$ itself can be an integer power of some prime number. 
We may also regard $M_\alpha~(\alpha\in\cL)$ as $M_\alpha=(M_\alpha^1,M_\alpha^2,\cdots,M_\alpha^{m_\alpha})$ where $M_\alpha^i\in\bF_p$ for $i=1,2,\cdots,m_\alpha$. 

The weakly secure SMDC problem is depicted in \reffig{fig_WS-SMDC}. 
\begin{figure}[t]
	\centering
	\begin{tikzpicture}[thick,scale=0.85, every node/.style={transform shape},font=\footnotesize]
	\draw [->,>=stealth](-1.0,0)--(-0.3,0);
	\node (message) at (-1.2,0.6) {$M_1$};
	\node (message) at (-1.2,0.3) {$\vdots$};
	\node (message) at (-1.2,-0.1) {$M_L$};
	\node (message) at (-1.2,-0.5) {$K$};
	\draw [thick] (-0.3,-1.0) rectangle (1.0,1.0) node(encoder)[midway] {Encoder};	
	
	\draw[thick] (3.5,-1.0) rectangle (4.8,1.0) node (decoder)[midway]{Decoder};
	\draw [->,>=stealth] (4.8,0)--(5.3,0);
	\draw (5.3,0) node [right]{$M_1,M_2,\cdots,M_{\alpha}$};
	
	\draw (1.0,0.6)--(2.6,0.6)--(3.0,0.8); \draw (3.0,0.6)--(3.5,0.6); \node at (1.25,0.75) {$W_1$};
	\draw (1.0,0.2)--(2.6,0.2)--(3.0,0.4); \draw (3.0,0.2)--(3.5,0.2); \node at (1.25,0.35) {$W_2$};
	\draw (1.25,0) node {$\vdots$};
	\draw  (1.0,-0.7)--(2.6,-0.7)--(3.0,-0.5); \draw (3.0,-0.7)--(3.5,-0.7); \node at (1.25,-0.52) {$W_L$};
	
	\draw[thick] (3.5,-1.3) rectangle (4.8,-3.3); \node at (4.1,-2.1){Eaves-};\node at (4.1,-2.4) {dropper};
	\draw [->,>=stealth] (4.8,-2.3)--(5.3,-2.3);
	\draw (5.3,-2.3) node [right]{no information of $M_{\alpha}$};
	
	\draw (2.2,0.6)--(2.2,-1.7)--(2.6,-1.7)--(3.0,-1.5); \draw (3.0,-1.7)--(3.5,-1.7);
	\draw (1.8,0.2)--(1.8,-2.1)--(2.6,-2.1)--(3.0,-1.9); \draw (3.0,-2.1)--(3.5,-2.1);
	\draw (2.5,-2.5) node {$\vdots$};
	\draw  (1.3,-0.7)--(1.3,-3.0)--(2.6,-3.0)--(3.0,-2.8); \draw (3.0,-3.0)--(3.5,-3.0);
	\fill [black] (2.2,0.6) circle (1.3pt); \fill [black] (1.8,0.2) circle (1.3pt); \fill [black] (1.3,-0.7) circle (1.3pt);
	
	\draw[dashed] (2.8,-2.35) ellipse (0.5 and 1.1);
	\node at (2.8,-3.7) {$|\cA|\leq N_\alpha$};
	\draw[dashed] (2.8,0.1) ellipse (0.5 and 1.0);
	\node at (2.8,1.25) {$|\cU|=\alpha$};
	\end{tikzpicture}
	\caption{The Weakly Secure SMDC Model}
	\label{fig_WS-SMDC}
\end{figure}
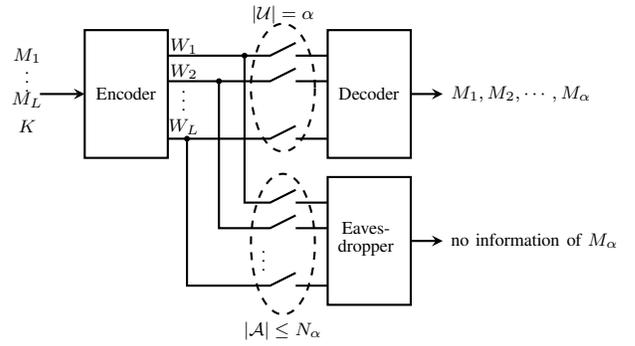
There are $L$ encoders, indexed by $\cL$, each of which can access all the $L$ information messages. 
There are also $2^L-1$ decoders. For each $\cU\subseteq\cL$ such that $\cU\neq \emptyset$, 
Decoder-$\cU$ can access the outputs of the subset of encoders indexed by $\cU$. 
For $\alpha\in\cL$ and any $\cU$ such that $|\cU|=\alpha$, Decoder-$\cU$ can completely recover the first $\alpha$ messages $M_1,M_2,\cdots,M_{\alpha}$. 
In addition, there is an eavesdropper who has access to the outputs of a subsets $\cA$ of encoders. 
Let $\bm{N}=(N_1,N_2,\cdots,N_L)$ be $L$ non-negative integers, where $N_\alpha<\alpha$ for $\alpha\in \cL$. 
Weak security requires that each individual message $M_\alpha$ should be kept perfectly secure from the eavesdropper if $|\cA|\leq N_\alpha$.

Let $\cK$ be the key space.
An $(m_1,m_2,\cdots,m_L,R_1,R_2,\cdots,R_L)$ code is formally defined by the encoding functions
\begin{equation}
E_l:\prod_{i=1}^L\bF_{p^{m_i}}\times\cK\rightarrow\bF_{p^{R_l}}, \text{ for }l\in\cL \label{encoding fn}
\end{equation}
and decoding functions
\begin{equation}
D_{\cU}:\prod_{l\in\cU}\bF_{p^{R_l}}\rightarrow\prod_{i=1}^{|\cU|}\bF_{p^{m_i}},\text{ for }\cU\subseteq\cL\text{ and }\cU\neq\emptyset.  \label{decoding fn}
\end{equation}
Denote the shared key as $K$ (accessible to all the encoders), which is uniformly distributed in the key space $\cK$. 
Let $W_l=E_l(M_1,M_2,\cdots,M_L,K)$ be the output of $\text{Encoder-}l$ and $W_{\cU}=(W_l:l\in\cU)$ for $\cU\subseteq\cL$. Define the normalized message rates $\sm_l\triangleq m_l/\sum_{l=1}^L m_l$, from which it follows that $\sum_l \sm_l=1$.
A normalized non-negative rate tuple $\bm{\sR}\triangleq(\sR_1,\sR_2,\cdots,\sR_L)$ is \textit{achievable} for the normalized message rates $(\sm_1,\ldots,\sm_L)$, if for any $\epsilon>0$, there exist an integer $a$ and an $(a\sm_1,a\sm_2,\cdots,a\sm_L,R_1,R_2,\cdots,R_L)$ code such that
\begin{align}
&\textbf{perfect reconstruction:   }D_{\cU}(W_{\cU})=(M_1,M_2,\cdots,M_{|\cU|}),\nonumber\\
&\qquad\qquad\qquad\qquad\qquad\qquad ~\forall~\cU\subseteq\cL\text{ s.t. }\cU\neq\emptyset,   \label{recover-constraint}\\
&\textbf{perfect secure:   }H(M_\alpha|W_\cA)=H(M_\alpha),\nonumber\\
&\qquad\qquad\qquad~\forall~\alpha\in\cL\text{ and }\cA\subseteq\cL\text{ s.t. }|\cA|\leq N_\alpha,   \label{security-constraint}
\end{align}
and
\begin{align}
\textbf{coding rate:   }\sR_l+\epsilon\geq a^{-1}R_l, \qquad l\in \cL.
\end{align}
The optimal coding rate region $\cR$ is defined as the collection of all achievable rate tuples.

\begin{remark}
	Here each message $M_\alpha$ can be essentially represented in $m_\alpha\log_2 p$ bits, and each codeword $W_l$ can be represented in $R_l\log_2 p$ bits. Thus $R_l$ can be viewed as the coding rate of encoder $E_l$, when the definition of the entropy function uses logarithm of base $p$, which will be adopted from here on. The quantity $\sR_l$ is then essentially the normalized $R_l$.
\end{remark}


The minimum achievable normalized sum rate is defined as $\sR_{\text{sum}}^*\triangleq \min\sum_{l=1}^L \sR_l$, and one of our main results is a necessary and sufficient condition for superposition coding to be sum-rate optimal. We also study an important case where $\bm{N}$ is given by
\begin{equation}
N_\alpha=
\begin{cases}
\alpha-1,&\text{ for }1\leq \alpha\leq r\\
0,&\text{ for }r+1\leq \alpha\leq L,
\end{cases}\label{def-DS-SMDC}
\end{equation}
for certain parameters $(L,r)$, where $r \ge 1$.
We refer to this system as the $(L,r)$ differential-constant secure SMDC (DS-SMDC), 
where the more important messages (i.e., small $\alpha$ values) are maximally secure ($N_\alpha=\alpha-1$) and the less important messages do not have any security guarantee at all ($N_\alpha=0$). 
For the protected messages ($1\leq \alpha\leq r$), while the security constraint $N_\alpha$ grows with the reconstruction requirement $\alpha$, 
the difference between $N_\alpha$ and $\alpha$ remains to be a constant equal to 1. 
We refer to this feature as ``differential-constant secure", in contrast to the ``level-constant secure" guarantee \cite{liutie13-sSMDC,liutie14-s-all} which requires $N_{\alpha}=N$ for all $\alpha>N$.  
Denote the optimal coding rate region of the $(L,r)$ DS-SMDC problem by $\cR_{L,r}$, which is the collection of all achievable normalized rate tuples. 
For $r=1$, the problem reduces to the classical SMDC. 

\subsection{An Achievable Rate Region via Superposition Coding}\label{section-Rsup}
Let $M$ be a message encoded by $n$ encoders. 
For any $0\leq c<k\leq n$, the $(c,k,n)$ ramp secret sharing problem \cite{yamamoto85}, also known as the secure symmetrical single-level diversity coding (S-SSDC) problem in \cite{liutie13-sSMDC}, requires that 
the outputs from any subset of no more than $c$ encoders provide no information about the message,  
and the outputs from any subset of $k$ encoders can completely recover the message. The optimal rate region for this problem can be found in \cite{ramp-secret-sharing96,liutie13-sSMDC}, as stated in the following lemma. 

\begin{lemma}\label{lemma-region-ramp}
	The optimal rate region of the $(c,k,n)$ ramp secret sharing problem is the collection of rate tuples $(R_1,R_2,\cdots,R_n)$ such that 
	\begin{equation}
	\sum_{l\in\cB}R_l\geq H(M),~\forall\cB\subseteq \{1,2,\cdots,n\},|\cB|=k-c.  \label{region-rss}
	\end{equation}
\end{lemma}
\begin{remark}
	If $k=c+1$, the $(c,k,n)$ ramp secret sharing problem reduces to the $(k,n)$ threshold secret sharing problem and the rate region reduces accordingly. 
\end{remark}

\begin{figure*}[t]
	\centering
	\begin{tikzpicture}[scale=1.0]
	\node at (0.5,4.3) {Messages}; 
	\node at (3.3,4.3) {Encrytion Keys}; 
	\node at (8,4.3) {Codewords}; 
	
	\node at (1.5,4.3) {$+$}; 
	\node at (1.5,1.6) {$+$}; 
	\node at (1.5,0.4) {$+$}; 
	
	\draw [->] (5.0,4.3)--(5.8,4.3); 
	\draw [->] (5.0,2.8)--(5.5,2.8); 
	\draw [->] (5.0,1.6)--(5.5,1.6); 
	\draw [->] (5.0,0.4)--(5.5,0.4); 
	
	\draw[blue](-0.3,2.4) rectangle (1.2,3.2); \node at (0.5,2.8) {$M_1$}; 
	\draw[red](-0.3,1.2) rectangle (1.2,2); \node at (0.5,1.6) {$M_2$}; 
	\draw[cyan](-0.3,0) rectangle (1.2,0.8); \node at (0.5,0.4) {$M_3$}; 
	
	\draw[red](2.3,1.2) rectangle (3.8,2.0); \node at (3,1.6) {$Z_2$}; 
	\draw[cyan](1.7,0) rectangle (3.2,0.8); \node at (2.5,0.4) {$Z_3^1$}; 
	\draw[cyan](3.3,0) rectangle (4.8,0.8); \node at (4.05,0.4) {$Z_3^2$}; 
	
	\draw[blue](5.7,2.4) rectangle (7.2,3.2); \node at (6.45,2.8) {$M_1$}; 
	\draw[blue](7.5,2.4) rectangle (9.0,3.2); \node at (8.25,2.8) {$M_1$}; 
	\draw[blue](9.3,2.4) rectangle (10.8,3.2); \node at (10.05,2.8) {$M_1$}; 
	
	\draw[red](5.7,1.2) rectangle (7.2,2); \node at (6.45,1.6) {$M_2+Z_2$}; 
	\draw[red](7.5,1.2) rectangle (9.0,2); \node at (8.25,1.6) {$M_2\!+\!2Z_2$}; 
	\draw[red](9.3,1.2) rectangle (10.8,2); \node at (10.05,1.6) {$Z_2$}; 
	
	\draw[cyan](5.7,0) rectangle (7.2,0.8); \node at (6.45,0.4) {$M_3+Z_3^1$}; 
	\draw[cyan](7.5,0) rectangle (9.0,0.8); \node at (8.25,0.4) {$M_3+Z_3^2$}; 
	\draw[cyan](9.3,0) rectangle (10.8,0.8); \node at (10.05,0.4) {$Z_3^1+Z_3^2$}; 
	
	\draw[dashed,blue](-0.4,2.3) rectangle (11.1,3.3); 
	\draw[dashed,red](-0.4,1.1) rectangle (11.1,2.1); 
	\draw[dashed,cyan](-0.4,-0.1) rectangle (11.1,0.9); 
	
	\draw(5.6,-0.2) rectangle (7.3,3.4); \node at (6.45,3.7) {$W_1$}; 
	\draw(7.4,-0.2) rectangle (9.1,3.4); \node at (8.25,3.7) {$W_2$}; 
	\draw(9.2,-0.2) rectangle (10.9,3.4); \node at (10.05,3.7) {$W_3$}; 
	\end{tikzpicture}
	\caption{The superposition coding scheme for $(3,3)$ DS-SMDC with $(m_1,m_2,m_3)=(1,1,1)$, $(N_1,N_2,N_3)=(0,1,2)$, and $p=3$.}
	\label{fig_superposition}
\end{figure*}
In light of this result, a natural coding scheme (i.e., superposition coding) for the weakly secure SMDC problem formulated above is to 
separately encode each message $M_\alpha$ using an $(N_\alpha,\alpha,L)$ ramp secret sharing code as shown in Fig.~\ref{fig_superposition}. 
The rate region induced by superposition coding provides an inner bound $\cR_{\sup}$ for $\cR$, and by \Cref{lemma-region-ramp}, it can be written as the set of non-negative rate tuples $\bm{\sR}=(\sR_1,\sR_2,\cdots,\sR_L)$ such that 
\begin{equation}
\sR_l=\sum_{\alpha=1}^L r_l^\alpha, \text{ for }l\in\cL  \label{sup-region-R}
\end{equation}
for some $r_l^\alpha\geq 0,~l,\alpha\in\cL$, satisfying 
\begin{equation}
\sum_{l\in\cB}r_l^{\alpha}\geq \sm_{\alpha}, \text{ for }\cB\subseteq\cL \text{ s.t. } |\cB|=\alpha-N_\alpha.  \label{sup-region-r}
\end{equation}
The induced sum rate provides an upper bound $\bar{\sR}_{\text{sum}}$ for $\sR_{\text{sum}}^*$, and can be written simply as,
\begin{equation}
\bar{\sR}_{\text{sum}}\triangleq\sum_{\alpha=1}^L\frac{L\sm_\alpha}{\alpha-N_\alpha}.   \label{sup-sumrate}
\end{equation}

\subsection{Properties of MDS Code for Secret Sharing}\label{section-MDS}
In this section, we describe in some details two $(n,k)$ maximum distance separable (MDS) codes for ramp secret sharing 
that achieve the minimum sum rate in \Cref{lemma-region-ramp}, and provide important properties that are instrumental to the joint coding strategy we later propose. 

Let $M=(U_1,U_2,\cdots,U_{k-c})$ be a length-$(k-c)$ message where each symbol is chosen uniformly and independently from the finite field $\mathbb{F}_p$. 
Let $Z_1,Z_2,\cdots, Z_c$ be independent random keys chosen uniformly from the same finite field $\mathbb{F}_p$. 
For $i=1,2,\cdots,k$, define the following length-$k$ vectors:
\begin{equation}
f_i=[\underbrace{0~ \cdots~ 0}_{i-1}~ 1~ 0~ \cdots~ 0]^T.   \label{MDS-def-f}
\end{equation}
Let $g_{1}, g_{2}, \cdots, g_n$ be length-$k$ vectors with entries from $\mathbb{F}_p$ such that
any $k$ vectors $\{h_{j_1}, h_{j_2}, \cdots, h_{j_k}\}$ chosen from the set $\{f_1, f_2, \cdots, f_{k},$ $g_1, g_2, \cdots, g_n\}$ 
satisfy the full rank condition over $\mathbb{F}_p$, i.e., 
\begin{equation}
\text{rank}\left[h_{j_1}~ h_{j_2}~ \cdots~ h_{j_k}\right]=k.   \label{MDS-full-rank-condition}
\end{equation}
It can be shown that as long as $p\geq n+k$, there exist such vectors $g_{1}, g_{2}, \cdots, g_n$, e.g., it can be chosen as the columns from a Cauchy matrix. 
The generator matrices of the two MDS codes of interest are given, respectively, as
\begin{align}
&G^{(1)}=\left[f_{k-c+1}~\cdots~f_{k}~g_1~ g_2~ \cdots~ g_{n-c}\right], \\
&G^{(2)}=\left[g_1~ g_2~ \cdots~ g_n\right].
\end{align}
Then the codewords of two MDS codes are, respectively, 
\begin{align}
&[Y_1,Y_2,\cdots,Y_n]=\big[U_1~\cdots~U_{k-c}~~Z_1~\cdots~ Z_c\big]G^{(1)},  \label{MDS-codeword-1}\\
&[Y_1,Y_2,\cdots,Y_n]=\big[U_1~\cdots~U_{k-c}~~Z_1~\cdots~ Z_c\big]G^{(2)}.  \label{MDS-codeword-2}
\end{align}
We shall refer these two codes  as MDS-A and MDS-B, respectively. By the definition of $f_{k-c+1},\cdots,f_{k}$ in \eqref{MDS-def-f}, MDS-A has the random keys explicitly as part of the coded message,
\begin{equation}
[Y_1,Y_2,\cdots,Y_c]=\big[Z_1~Z_2~\cdots~ Z_c\big].
\end{equation}
It is obvious that for both codes, $M$ and $Z_1, Z_2,\cdots, Z_c$ can be perfectly recovered from any $k$ coded symbols. 

Since all the coded symbols are linear combinations of the messages and the random keys that are uniformly distributed, we have the following lemma. 
\begin{lemma}\label{lemma-property-0}
	Any $k$ coded symbols of MDS-A and MDS-B are uniformly distributed over $\bF_{p^k}$ . 
\end{lemma}

The main difference between the two codes, which is the most relevant to this work, is given in the following two lemmas. 

\begin{lemma}\label{lemma-property-MDS-1}
	For any integer $t$ such that $c\leq t\leq k$, let $\cE\subseteq\{1,2,\cdots,n\}$ where $|\cE|=t$, and $\cA\subseteq\{1,2,\cdots,k-c\}$ where $|\cA|=k-t$. 
	The codewords of MDS-A has the following property:
	\begin{equation}
	I(Y_{\cE};U_{\cA})=0,
	\end{equation}
	where $Y_{\cE}\triangleq\{Y_i:i\in\cE\}$ and $U_{\cA}\triangleq\{U_i:i\in\cA\}$.
\end{lemma}
\begin{proof}
	We consider the following chain of equality
	\begin{align}
	&I(Y_{\cE};U_{\cA})   \nonumber \\
	=&H(Y_{\cE})-H(Y_{\cE}|U_{\cA}) \\ 
	=&H(Y_{\cE})-H(Y_{\cE}|U_{\cA})+H(Y_{\cE}|U_{\cA}U_{\bar{\cA}}Z_1Z_2\cdots Z_c)   \label{MDS-indep-2}\\ 
	=&H(Y_{\cE})-H(U_{\bar{\cA}}Z_1Z_2\cdots Z_c|U_{\cA})\\
	&+H(U_{\bar{\cA}}Z_1Z_2\cdots Z_c| Y_{\cE}U_{\cA}) \\ 
	=&H(Y_{\cE})-H(U_{\bar{\cA}}Z_1Z_2\cdots Z_c)   \label{MDS-indep-4} \\ 
	=&\frac{t}{k-c}H(M)-\frac{t-c+c}{k-c}H(M)    \label{MDS-indep-5} \\
	=&0, 
	\end{align}
	where \eqref{MDS-indep-2} follows from \eqref{MDS-codeword-1}, 
	and both \eqref{MDS-indep-4} and \eqref{MDS-indep-5} follow from the full rank condition in \eqref{MDS-full-rank-condition} and the uniform and mutually independent distribution of the messages and the encryption key. 
\end{proof}
\begin{remark}
	For $t=c$, \Cref{lemma-property-MDS-1} reduces to the stated security constraint of parameter $c$; on the other hand, for $t>c$ (but $t\leq k$), any $t$ coded symbols reveal no information about any subset of $k-t$ message symbols. 
\end{remark}

\begin{lemma}\label{lemma-property-MDS-2}
	For any integer $t$ such that $0\leq t\leq k$, let $\cE\subseteq\{1,2,\cdots,n\}$ where $|\cE|=t$, and $\cA_1\subseteq\{1,2,\cdots,k-c\}$, and $\cA_2\subseteq\{1,2,\cdots,c\}$ where $|\cA_1|+|\cA_2|=k-t$.
	The codewords of MDS-B has the following property:
	\begin{equation}
	I(Y_{\cE};U_{\cA_1},Z_{\cA_2})=0,
	\end{equation}
	where $Y_{\cE}\triangleq\{Y_i:i\in\cE\}$, $U_{\cA_1}\triangleq\{U_i:i\in\cA_1\}$, and $Z_{\cA_2}\triangleq\{Z_i:i\in\cA_2\}$.	
\end{lemma}
\begin{proof}
	This is direct from the full-rank condition in \eqref{MDS-full-rank-condition} and the uniform and mutually independent distribution of the messages and the encryption key.
\end{proof}

From the above two lemmas, in contrast to MDS-A, MDS-B has the additional advantage that part of the keys can also be made secure against some $t$ eavesdroppers.  
This property becomes important to us in the sequel. 


\section{Main Results}\label{section-main}
\subsection{Sum-rate Optimality Conditions of Superposition}\label{section-conditions}
The main question we seek to answer here is under what condition the equality $\sR_{\text{sum}}^*=\bar{\sR}_{\text{sum}}$ will hold, and the following theorem provides the exact answer to this question. 
\begin{theorem}\label{thm-sup-opt-condition}
	$\sR_{\text{sum}}^*=\bar{\sR}_{\text{sum}}$, if and only if 
	for any $\alpha<\beta\in\cL$ where $\sm_\alpha, \sm_\beta>0$, we have
	\begin{equation}
	\text{ either } N_\alpha<\alpha\leq N_\beta<\beta, \text{ or }N_\alpha=N_\beta=0. \label{sup-opt-condition}
	\end{equation}
\end{theorem}

\begin{remark}
	If all $L$ messages are non-degenerate, i.e., all the message entropies are non-zero, 
	the condition in \eqref{sup-opt-condition} is equivalent to that there exists a $T_s\in\{1,2,\cdots,L\}$ such that for all $\alpha\in\cL$, 
	\begin{equation}
	N_{\alpha}=
	\begin{cases}
	0,&\text{for }\alpha\leq T_s \\
	\alpha-1,&\text{for }\alpha>T_s. 
	\end{cases}
	\label{sup-opt-condition-2}
	\end{equation}
	If we do not assume non-degeneration, then the following necessary condition for optimality can be induced from \eqref{sup-opt-condition}: There exists a $T_s\in\{1,2,\cdots,L\}$ such that for any $\alpha\in\cL$ satisfying $m_\alpha>0$,
	\begin{equation}
	\begin{cases}
	N_{\alpha}=0,&\text{ for }\alpha\leq T_s \\
	N_{\alpha}>0,&\text{ for }\alpha>T_s. 
	\end{cases}
	\label{def-T_s}
	\end{equation}
\end{remark}
\begin{remark}
	The following are two examples that superposition coding is optimal in terms of achieving the entire rate region and thus \Cref{thm-sup-opt-condition} reduces correctly. 
	\begin{itemize}
		\item If the threshold in \eqref{sup-opt-condition-2} is $T_s=L$, the security constraints are given as  
		\begin{equation}
		N_\alpha=0, \text{ for all } \alpha\in\cL,  \label{opt-condition-1}
		\end{equation}
		then the problem reduces to the classical SMDC problem without security constraints, where superposition is known to be optimal \cite{yeung99}. 
		
		\item If the threshold in \eqref{sup-opt-condition-2} is $T_s=1$, the security constraint becomes 
		\begin{equation}
		N_{\alpha}=\alpha-1, \text{ for all } \alpha\in\cL, 
		\end{equation}
		and the problem reduces to the special case of DS-SMDC for $r=L$ in \Cref{section-DS}. 
	\end{itemize}
\end{remark}


The following definition will be used in the sequel. 
\begin{definition}
	For any $\alpha<\beta\in\cL$, we define two conditions.
	\begin{align}
	\text{Condition 1}:~ &N_\alpha<N_\beta<\alpha;  \label{not-opt-condition-1}  \\
	\text{Condition 2}:~ &N_\beta\leq N_\alpha~\&~N_\alpha>0.  \label{not-opt-condition-2}
	\end{align} 
\end{definition}
Theorem \ref{thm-sup-opt-condition} can be alternatively written in the following form, by taking the complement of the conditions in~\eqref{sup-opt-condition}.
\begin{namedthm}{Theorem \ref{thm-sup-opt-condition}'}\label{thm-not-opt-condition}
	$\sR_{\text{sum}}^*<\bar{\sR}_{\text{sum}}$, if and only if there exist $\alpha<\beta\in\cL$ where $\sm_\alpha, \sm_\beta>0$ such that either Condition 1 in \eqref{not-opt-condition-1} or Condition 2 in \eqref{not-opt-condition-2} holds.
\end{namedthm}

We prove \Cref{thm-sup-opt-condition} in two parts. 
In \Cref{section-joint-strategy}, we show that superposition is suboptimal under the security constraints in \eqref{not-opt-condition-1} or \eqref{not-opt-condition-2}, 
by providing joint coding strategies that can reduce coding rates. 
In \Cref{section-conditions-converse},  the optimality of superposition coding is established by proving 
that the sum rate is lower bounded by $\bar{\sR}_{\text{sum}}$ in (\ref{sup-sumrate}).

\begin{remark}
	Superposition coding is optimal for classical SMDC where there is no security constraints, 
	i.e., suboptimality only happens when there is a security constraint. 
	In view of the suboptimality in \Cref{section-scheme-alpha-key-beta,section-scheme-beta-key-alpha}, 
	we see intuitively that joint encoding helps only when some message can perform as the secret key of another message. 
\end{remark}

\subsection{Rate Region of DS-SMDC}\label{section-DS}
When superposition is not optimal, it is generally hard to characterize the coding rate region or even the minimum sum rate, since it is difficult to find the optimal code structures. 
In this section, we study the $(L,r)$ DS-SMDC problem for which we fully characterize the optimal rate region. 
The pairwise coding strategy in \Cref{section-scheme-beta-key-alpha} can be generalized to a multi-message regime, and we obtain a {\it group pairwise} coding scheme that achieves the entire rate region of the DS-SMDC problem. 

We first present an example that motivates the general group pairwise coding scheme. 
\begin{example}\label{example-DSSMDC}
	Let $L=4, (m_1,m_2,m_3,m_4)=(1,1,1,4)$, and $p=11$. The security constraint for the $(4,3)$ DS-SMDC problem should be $(N_1,N_2,N_3,N_4)=(0,1,2,0)$.
	We can follow a naive strategy as illustrated in \eqref{code-example-DS-sup}: use generator matrices $G_2$ and $G_3$ generated from MDS-B to encode $M_2$ and $M_3$ separately with encryption keys $Z_2$ and $Z_3^1,Z_3^2$; equally partition $M_4$ into four pieces $M_4^1,M_4^2,M_4^3,M_4^4$. 
	\begin{align}
	W_1&=(M_2+{\color{blue}Z_2},\quad~M_3+2{\color{red}Z_3^1}+9{\color{orange}Z_3^2},W_4^1),  \nonumber \\
	W_2&=(M_2+2{\color{blue}Z_2},~9M_3+8{\color{red}Z_3^1}+6{\color{orange}Z_3^2},W_4^2),  \nonumber \\
	W_3&=(M_2+3{\color{blue}Z_2},6M_3+10{\color{red}Z_3^1}+7{\color{orange}Z_3^2},W_4^3),  \nonumber \\
	W_4&=(M_2+4{\color{blue}Z_2},~~7M_3+9{\color{red}Z_3^1}+7{\color{orange}Z_3^2},W_4^4);  \label{code-example-DS-sup}
	\end{align}
	The first part of the group pairwise coding scheme is simply to use $M_4$, specifically $M_4^1,M_4^2,M_4^3$, to replace $Z_2$ and $Z_3^1,Z_3^2$ as secret keys to encrypt $M_2$ and $M_3$, as given in \eqref{code-example-DS-joint}. 
      \begin{align}
	W_1'&=(M_2+{\color{blue}M_4^1},\quad~M_3+2{\color{red}M_4^2}+9{\color{orange}M_4^3}\qquad ),  \nonumber \\
	W_2'&=(M_2+2{\color{blue}M_4^1},~9M_3+8{\color{red}M_4^2}+6{\color{orange}M_4^3}\qquad ), \nonumber \\
	W_3'&=(M_2+3{\color{blue}M_4^1},6M_3+10{\color{red}M_4^2}+7{\color{orange}M_4^3}\qquad ), \nonumber \\
	W_4'&=(M_2+4{\color{blue}M_4^1},~~7M_3+9{\color{red}M_4^2}+7{\color{orange}M_4^3},W_4^4).  \label{code-example-DS-joint}
	\end{align}	
	The second part of  the group pairwise coding scheme simply encodes $M_4^4$ as part of the fourth coded message. 
	Since $M_4^1,M_4^2,M_4^3$ does not need to be separately encoded, rate saving is obtained compared to the naive version. 
	The reconstruction and security requirements of $M_2$ and $M_3$ are immediate from the MDS-B code. 
	The reconstruction requirement of $M_4$ is straightforward since $M_4^1,M_4^2,M_4^3$ is recovered with any three coded symbols. 
\end{example}

\noindent{\textbf{Coding scheme for general parameters:}}\\
The group pairwise coding scheme is illustrated in \reffig{fig_group-pairwise-scheme}. 
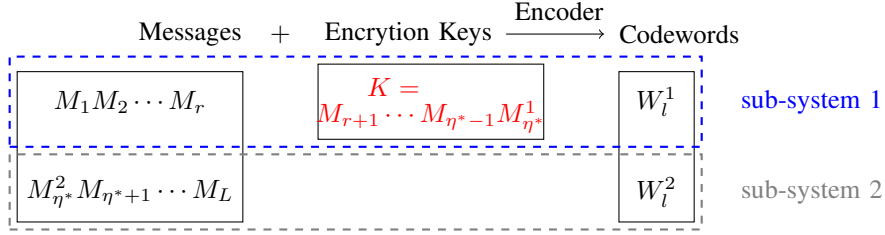
\begin{figure*}[!h]
	\centering
	\begin{tikzpicture}
	\node at (1.3,3.5) {Messages}; \node at (2.5,3.5) {$+$}; 
	\node at (4.2,3.5) {Encrytion Keys}; \draw [->] (5.5,3.5)--(6.8,3.5); 
	\node at (7.8,3.5) {Codewords}; \node at (6.2,3.8) {Encoder}; 
	
	\draw (-1,1) rectangle (2,3); 
	\node at (0.5,2.6) {$M_1M_2\cdots M_r$};
	\node at (0.5,1.4) {$M_{\eta^*}^2M_{\eta^*+1}\cdots M_L$};
	
	\draw (3,2.1) rectangle (6,3.1); \node at (4,2.8) {\color{red}$K=$}; 
	\node at (4.5,2.4) {\color{red}$M_{r+1}\cdots M_{\eta^*-1}M_{\eta^*}^1$};
	
	\draw (7,1) rectangle (8,3);
	\node at (7.5,2.6) {$W_l^1$}; \node at (7.5,1.4) {$W_l^2$};
	
	\draw [dashed,thick,blue] (-1.1,2.0) rectangle (8.1,3.2); \node [right] at (8.5,2.6) {\color{blue}sub-system 1}; 
	\draw [dashed,thick,gray] (-1.1,0.9) rectangle (8.1,1.9); \node [right] at (8.5,1.4) {\color{gray}sub-system 2}; 
	\end{tikzpicture}
	\caption{The group pairwise coding scheme.}
	\label{fig_group-pairwise-scheme}
\end{figure*}
For each $\alpha\in\{1,2,\cdots,r\}$, we will use an $(\alpha,L)$-threshold secret sharing scheme to encode $M_{\alpha}$ 
and use the last $L-r$ messages $M_{r+1},\cdots,M_L$ as keys. 
It is proved in \cite{yamamoto85} that the minimum key size for $M_{\alpha}$ is $(\alpha-1)m_\alpha$. 
Thus the total size of keys needed is 
\begin{equation}
|\cK|=\sum_{\alpha=1}^{r}(\alpha-1)m_\alpha.
\end{equation}

For notational simplicity, we define an auxiliary message\footnote{We use the auxiliary message $M_{L+1}$ to perform as encryption keys for the first $r$ messages if the messages $M_{r+1},\cdots,M_L$ are not enough. Thus, $M_{L+1}$ is non-vanishing (i.e., $m_{L=1}>0$), only when the total key size needed is strictly larger than the total size of messages $M_{r+1},\cdots,M_L$. }
$M_{L+1}$, which is independent with other messages and uniformly distributed over $\bF_{p^{m_{L+1}}}$ with 
\begin{equation}
m_{L+1}=\left[\sum_{\alpha=1}^{r}(\alpha-1)m_\alpha-\sum_{\alpha=r+1}^{L}m_\alpha\right]^+,  \label{def-nonnegative-function}
\end{equation}
where for any $x\in\mathbb{R}$, $[x]^+\triangleq \max\{0,x\}$. 
It is easy to check that 
\begin{equation}
\sum_{\alpha=r+1}^{L+1}m_{\alpha}\geq \sum_{\alpha=1}^{r}(\alpha-1)m_{\alpha}.
\end{equation}
Thus, there exists a unique $\eta^*\in\{r+1,r+2,\cdots, L+1\}$ such that
\begin{equation}
\sum_{\alpha=r+1}^{\eta^*-1}m_\alpha< \sum_{\alpha=1}^{r}(\alpha-1)m_\alpha\leq \sum_{\alpha=r+1}^{\eta^*}m_\alpha.  \label{opt-eta}
\end{equation}
The parameter $\eta^*$ determines which messages of $M_{r+1},M_{r+2},\cdots,M_{L+1}$ will be used as the encryption keys. 
In light of the definition of $\eta^*$ in \eqref{opt-eta}, denote the first $\frac{\sum_{\alpha=1}^{r}(\alpha-1)m_\alpha- \sum_{\alpha=r+1}^{\eta^*-1}m_\alpha} {m_{\eta^*}}$ fraction 
of $M_{\eta^*}$ by $M_{\eta^*}^1$, and the rest by $M_{\eta^*}^2$. 
Then we use the messages $(M_{r+1},M_{r+2},\cdots,M_{\eta^*-1},M_{\eta^*}^1)$ to replace the keys of $M_1,\cdots,M_r$. 
The messages $M_{\eta^*}^2, M_{\eta^*+1}, M_{\eta^*+2}, \cdots,M_{L}$ are separately encoded in the same way as in classical SMDC. 

Next, we verify the reconstruction and security constraints.

\vspace{0.2cm}
\noindent{\textit{Reconstruction: }} By the code construction in \Cref{section-MDS}, the reconstruction requirements of all messages $(M_1,M_2,\cdots,M_r)$, $(M_{r+1},M_{r+2},\cdots,M_{\eta^*-1},M_{\eta^*}^1)$, and $(M_{\eta^*}^2,M_{\eta+1},\cdots,M_L)$ are satisfied immediately. 

\vspace{0.2cm}
\noindent{\textit{Security: }} The security constraints of $(M_1,M_2,\cdots,M_r)$ is straightforward, and there is no security constraint for $(M_{r+1},M_{r+2},\cdots,M_L)$.

\begin{remark}
	The first $r$ messages $M_1,M_2,\cdots,M_{r}$ are encoded separately, and 
	the last $r$ messages $M_{r+1}, M_{r+2}, \cdots,M_L$ are also encoded separately. 
	The reason why we call the coding scheme ``group pairwise" is that joint encoding are only performed between the two groups of messages
	\begin{equation}
	\big\{M_1,M_2,\cdots,M_{r}\big\} \text{ and }\big\{M_{r+1}, M_{r+2}, \cdots,M_{\eta^*}\big\}.
	\end{equation}
\end{remark}

The group pairwise coding scheme can also be interpreted as superposition coding of the messages $M_1,M_2,\cdots,M_r$, $M_{r+1}^*,\cdots,M_L^*$, 
where the independent pseudo-messages $M_\alpha^* (r+1\leq \alpha\leq L)$ are defined by the message size $\sm_{\alpha}^*$ as 
\begin{equation}
\sm_{\alpha}^*=
\begin{cases}
0,&\text{ for }r+1\leq \alpha\leq \eta^*-1 \\
\sum\limits_{j=r+1}^{\eta^*}\sm_j-\sum\limits_{j=1}^{r}(j-1)\sm_j,&\text{ for }\alpha=\eta^*\\
\sm_{\alpha},&\text{ for }\eta^*+1\leq \alpha\leq L.
\end{cases}  \label{def-sm*}
\end{equation}
Then the coding rate region $\cR_{\text{gp}}^{L,r}$ induced by group pairwise coding is the set of $\bm{\sR}\geq \bm{0}$ such that 
\begin{equation}
\sR_l=\sum_{\alpha=1}^L r_l^{\alpha}, \text{ for }l\in\cL,  \label{gpRegion-R}
\end{equation}
where $r_l^{\alpha}\geq 0$ and 
\begin{align}
r_l^{\alpha}&\geq \sm_{\alpha},\text{ for }1\leq \alpha\leq r,  \label{gpRegion-r-1}\\
\sum_{l\in\cB}r_l^{\alpha}&\geq \sm_{\alpha}^*,\text{ for all }\cB\subseteq\cL\text{ s.t. }|\cB|=\alpha,~r+1\leq \alpha\leq L.  \label{gpRegion-r-2}
\end{align}

Our main result on DS-SMDC is the following theorem.
\begin{theorem}\label{thm-group-pairwise-region}
	$\cR_{L,r}=\cR_{\text{gp}}^{L,r}$.
\end{theorem}
\begin{proof}
	The achievability is immediate from the group pairwise coding scheme. 
	The converse is proved through a sophisticated iteration of information inequalities, which can be found in \Cref{section-DS-proof}. 
\end{proof}
\begin{remark}
	From the group pairwise code design and the converse proof in \Cref{section-DS-proof}, 
	we see that both the group pairwise coding scheme and the converse are compatible with $r=1$ and $r=L$. 
	Nevertheless, in order to emphasize the specificity of the case $r=L$ and to distinguish superposition and group pairwise joint coding, 
	we discuss the optimality for $r=L$ separately in the following.
\end{remark}

\subsubsection{Optimality of Superposition Coding for $(L,L)$-DS-SMDC}
For $r=L$, all the messages are protected. 
We separately encode the $L$ independent messages, where each $M_\alpha$ is encoded using an $(\alpha,L)$ threshold secret sharing scheme. 
The induced superposition rate region $\cR_{\text{sup}}^L$ can be obtained from \eqref{sup-region-R} and \eqref{sup-region-r} by letting $\alpha-N_\alpha=1$ for all $\alpha\in\cL$. To be specific, $\cR_{\text{sup}}^L$ is the set of nonnegative rate tuples $\bm{\sR}$ such that 
\begin{equation}
\sR_l=\sum_{\alpha=1}^{L}r_l^{\alpha}, \text{ for }l\in\cL  \label{R_sup-R}
\end{equation}
where $r_l^{\alpha}\geq 0$, and 
\begin{equation}
r_l^{\alpha}\geq \sm_{\alpha}, \text{ for }1\leq \alpha\leq L.  \label{R_sup-r}
\end{equation}
It is easy to eliminate $r_l^{\alpha}~(l,\alpha\in\cL)$ and obtain the following equivalent characterization of the superposition region,
\begin{equation}
\cR_{\text{sup}}^L=\{\bm{\sR}: \sR_l\geq \sum_{\alpha=1}^{L}\sm_{\alpha}, \text{ for all }l\in\cL\}.  \label{R_sup}
\end{equation}
The following corollary of \Cref{thm-group-pairwise-region} states that superposition coding is optimal for the $(L,L)$ DS-SMDC problem.
\begin{corollary}\label{thm-LL}
	$\cR_{L,L}=\cR_{\text{sup}}^L$. 
\end{corollary}
\begin{proof}
	The proof of the converse part is straightforward, so we omit the details and derive the conclusion directly from \Cref{thm-group-pairwise-region}. It is easily seen by comparing \eqref{gpRegion-R}-\eqref{gpRegion-r-2} and \eqref{R_sup-R}-\eqref{R_sup-r} that $\cR_{\text{gp}}^{L,r}$ reduces to $\cR_{\text{sup}}^L$ for $r=L$. Thus, by \Cref{thm-group-pairwise-region}, we have $\cR_{L,L}=\cR_{\text{gp}}^{L,r}=\cR_{\text{sup}}^L$. 
\end{proof}

\section{Achievability of \Cref{thm-sup-opt-condition}: Joint Coding Strategies}\label{section-joint-strategy}
In order to prove the necessity part of \Cref{thm-sup-opt-condition}, we instead prove the sufficiency part of Theorem~1', in the two separate cases given in \eqref{not-opt-condition-1} and \eqref{not-opt-condition-2}.

\subsection{Low Security Level at Higher Diversity Level}\label{section-scheme-alpha-key-beta}
In this section, we provide a joint coding strategy for the case that Condition 1 in \eqref{not-opt-condition-1} holds which provides rate saving, compared to superposition coding. 
We first discuss a motivating example to illustrate the key insight on how such rate saving is obtained. 

\begin{example}\label{example-2-3}
	Let $L=3, (\alpha, \beta)=(2,3), (m_2,m_3)=(2,2), (N_2,N_3)=(0,1)$, and $p=5$. Let $Z_3$ be an independent random key uniformly chosen from $\bF_p$. 
	Let the two messages be encoded with generator matrices constructed using MDS-A, which induce the coded symbols as shown in \Cref{table-example-2-3=1} through superposition. The important insight is that the coded message of $M_2$ can be used as the secret key to encode $M_3$, which reduces the coding rate.  More precisely, we replace $Z_3$ by $Y_2^1=Z_2^1+Z_2^2$ to serve as the key for $M_3$. 
	The coded symbols for this joint coding strategy are shown in \Cref{table-example-2-3=2}. By comparing the two tables, it is seen that the sum rate is reduced since the coded symbol $Z_3$ is eliminated. 
	
	\begin{table}[tb]
		\centering
		\def\arraystretch{1.2}
		\caption{Coding strategy for \Cref{example-2-3}}
		\begin{subtable}{1.0\linewidth}
			\centering
			\begin{tabular}{|c|c|c|c|}\hline
				\diagbox[dir=SE]{~}{}&$W_1$&$W_2$&$W_3$ \\ \hline
				\kern-0.5em$\alpha=2$\kern-0.5em& \kern-0.5em $Y_2^1= M_2^1+M_2^2$ \kern-0.5em& $2M_2^1+M_2^2$& $M_2^1+2M_2^2$  \\ \hline
				\kern-0.5em$\beta=3$\kern-0.5em& {\color{red}$Z_3$}& \kern-0.5em $M_3^1+2M_3^2+{\color{red}Z_3}$ \kern-0.5em& \kern-0.5em $2M_3^1+M_3^2+{\color{red}Z_3}$ \kern-0.5em  \\ \hline
			\end{tabular}
			\caption{Superposition coding strategy}
			\label{table-example-2-3=1}
		\end{subtable}%
		
		\begin{subtable}{1.0\linewidth}
			\centering
			\begin{tabular}{|c|c|c|c|}\hline
				\diagbox[dir=SE]{~}{}&$W_1$&$W_2$&$W_3$  \\ \hline
				\kern-0.5em$\alpha=2$\kern-0.5em& \multirow{2}*{\kern-0.5em $Y_2^1= M_2^1+M_2^2$ \kern-0.5em}&  $2M_2^1+M_2^2$& $M_2^1+2M_2^2$ \\ \cline{1-1}\cline{3-4}
				\kern-0.5em$\beta=3$\kern-0.5em& & \kern-0.5em $M_3^1+2M_3^2+{\color{red}Y_2^1}$ \kern-0.5em& \kern-0.5em $2M_3^1+M_3^2+{\color{red}Y_2^1}$ \kern-0.5em  \\ \hline
			\end{tabular}
			\caption{Joint coding strategy}
			\label{table-example-2-3=2}
		\end{subtable}%
		\label{table-example-2-3}
	\end{table}

	The reconstruction requirements of both $M_2$ and $M_3$ are straightforward. 
	There is no security requirement on $M_2$. For $M_3$, it is seen that any one coded symbol $W_l$ reveals no information about $M_3$. 
	For instance, eavesdropping $W_2$ gives
	\begin{align}
	&H(M_3|W_2)=H(M_3|M_3^1+2M_3^2+Y_2^1,M_2^1+Y_2^1)   \\
	&=H(M_3,M_3^1+2M_3^2+Y_2^1|M_2^1+Y_2^1)  \nonumber \\
	&\quad -H(M_3^1+2M_3^2+Y_2^1|M_2^1+Y_2^1)   \\
	&=H(M_3,M_3^1+2M_3^2+Y_2^1) -H(M_3^1+2M_3^2+Y_2^1)  \label{secure-M3} \\
	&=H(M_3|M_3^1+2M_3^2+Y_2^1) \\
	&=H(M_3),
	\end{align}
	where \eqref{secure-M3} follows from that $M_2^1$ is independent of $M_3,M_3^1+2M_3^2+Y_2^1,Y_2^1$ and $M_2^1$ is independent of $M_3^1+2M_3^2+Y_2^1,Y_2^1$. 
\end{example}

\noindent{\textbf{Coding strategy for general parameters:}}\\
First encode separately $M_{\alpha}$ and $M_{\beta}$ with generator matrices $G_\alpha$ and $G_\beta$ using MDS-A in \Cref{section-MDS}. 
The coded symbols for superposition coding strategy are as given in \Cref{table-strategy-a-key-b=1}. 
The joint coding strategy we propose is then to replace the first $\theta=\min\{N_\beta,\alpha-N_\beta\}$ encryption key symbols $(Z_\beta^1~Z_\beta^2~\cdots~ Z_\beta^{\theta})$ by the coded symbols
$(Y_\alpha^1,Y_\alpha^2,\cdots,Y_\alpha^{\theta})$. The parameter $\theta$ is strictly positive, which is implied by Condition 1 in \eqref{not-opt-condition-1}. 
Denote the corresponding codewords for $M_\beta$ thus obtained as $(Y_\beta^{1*},Y_\beta^{2*},\cdots,Y_\beta^{L*})$. 
The joint coding strategy of $M_\alpha$ and $M_\beta$ is illustrated in \Cref{table-strategy-a-key-b=2} and can be described as follows: 
\begin{equation}
W_i=
\begin{cases}
Y_\alpha^i, &\text{ for }1\leq i\leq \theta\\
[Y_\alpha^i,Y_\beta^{i*}], &\text{ for } \theta< i\leq L.
\end{cases} \label{joint-strategy-1}
\end{equation}

\begin{table}[tb]
	\centering
	\def\arraystretch{1.2}
	\caption{Coding strategy to replace encryption keys for $M_\beta$}
	\begin{subtable}{1.0\linewidth}
		\centering
		\begin{tabular}{|c|c|c|c|c|c|c|c|}\hline
			\diagbox[dir=SE]{~}{}&$W_1$&$W_2$&$\cdots$&$W_\theta$&$W_{\theta+1}$&$\cdots$&$W_L$  \\ \hline
			$\alpha$&$Y_\alpha^1$&$Y_\alpha^2$&$\cdots$&$Y_\alpha^\theta$&$Y_\alpha^{\theta+1}$&$\cdots$&$Y_\alpha^L$  \\ \hline
			$\beta$&{\color{red}$Y_\beta^1$}&{\color{red}$Y_\beta^2$}&{\color{red}$\cdots$}&{\color{red}$Y_\beta^\theta$}&{\color{red}$Y_\beta^{\theta+1}$}&{\color{red}$\cdots$} &{\color{red}$Y_\beta^L$}  \\ \hline
		\end{tabular}
		\caption{Superposition coding strategy}
		\label{table-strategy-a-key-b=1}
	\end{subtable}%
	
	\begin{subtable}{1.0\linewidth}
		\centering
		\begin{tabular}{|c|c|c|c|c|c|c|c|}\hline
			\diagbox[dir=SE]{~}{}&$W_1$&$W_2$&$\cdots$&$W_\theta$&$W_{\theta+1}$&$\cdots$&$W_L$  \\ \hline
			$\alpha$&\multirow{2}*{$Y_\alpha^1$}& \multirow{2}*{$Y_\alpha^2$}& \multirow{2}*{$\cdots$}& \multirow{2}*{$Y_\alpha^\theta$}& $Y_\alpha^{\theta+1}$& $\cdots$& $Y_\alpha^L$  \\ \cline{1-1}\cline{6-8}
			$\beta$&&&&&{\color{red}$Y_\beta^{(\theta+1)*}$}&{\color{red}$\cdots$} &{\color{red}$Y_\beta^{L*}$}  \\ \hline
		\end{tabular}
		\caption{Joint coding strategy}
		\label{table-strategy-a-key-b=2}
	\end{subtable}%
\end{table}

By comparing \Cref{table-strategy-a-key-b=1} and \Cref{table-strategy-a-key-b=2}, 
it can be seen that the coding rate is reduced compared to superposition coding 
because $(Y_\beta^1,Y_\beta^2,\cdots,Y_\beta^{\theta})$ are removed from the codewords, while the rates for all the others are unchanged. 
Next, we verify the reconstruction and security constraints for the two messages. 

\vspace{0.2cm}
\noindent{\textit{Reconstruction: }} The verification of the reconstruction requirements of both $M_\alpha$ and $M_\beta$ is straightforward. 

\vspace{0.2cm}
\noindent{\textit{Security: }} We consider the security requirements for the two levels separately.
\begin{enumerate}
	
	\item Assume we can access $N_\alpha$ coded symbols $W_\cB, |\cB|=N_\alpha$. 
	Partition $\cB$ into $\cB_1$ and $\cB_2$ such that $\cB_1\subseteq\{1,2,\cdots,\theta\}$ and $\cB_2\subseteq\{\theta+1,\cdots,L\}$. 
	Notice that
	\begin{align}
	&H(Y_\beta^{*\cB_2}|M_\alpha,Y_\alpha^{\cB_1}Y_\alpha^{\cB_2})   \nonumber \\
	&\geq H(Y_\beta^{*\cB_2}|M_\alpha,Y_\alpha^{1:\theta},Y_\alpha^{\cB_2})    \\
	&=H(Y_\beta^{*\cB_2}|Y_\alpha^{1:\theta})    \\
	&=H(Y_\beta^{*\cB_2}),  \label{check-secure-indep-3}
	\end{align} 
	where the second equality follows from the fact that conditioning does not increase entropy, and the last equality follows from \Cref{lemma-property-0} because 
	\begin{align}
	|\cB_2|+\theta&\leq N_\alpha+\theta  \\
	&= N_\alpha+\min\{\alpha-N_\beta,N_\beta\}  \\
	&\leq \alpha  \\
	&<\beta,
	\end{align}
	where the second inequality follows from $N_\alpha<N_\beta$ which is part of Condition 1 in \eqref{not-opt-condition-1}. 
	Since conditioning does not increase entropy, in light of \eqref{check-secure-indep-3}, we obtain  
	\begin{align}
	&H(Y_\beta^{*\cB_2}|M_\alpha,Y_\alpha^{\cB_1}Y_\alpha^{\cB_2})=H(Y_\beta^{*\cB_2}).  \label{check-secure-indep}
	\end{align} 
	It follows that
	\begin{align}
	&I(W_\cB;M_\alpha)  \nonumber \\
	&=I(W_{\cB_1}W_{\cB_2};M_\alpha) \notag\\
	&=I(Y_\alpha^{\cB_1}~Y_\alpha^{\cB_2}Y_\beta^{*\cB_2};M_\alpha)  \\
	&=I(Y_\alpha^{\cB_1}Y_\alpha^{\cB_2};M_\alpha)+I(Y_\beta^{*\cB_2};M_\alpha|Y_\alpha^{\cB_1}Y_\alpha^{\cB_2})  \\
	&=I(Y_\beta^{*\cB_2};M_\alpha|Y_\alpha^{\cB_1}Y_\alpha^{\cB_2}) \\
	&=0,   \label{check-secure-alpha5}
	\end{align}
	where the last but one equality follows from \Cref{lemma-property-MDS-1} and the fact that $|\cB_1|+|\cB_2|=N_\alpha$, 
	and \eqref{check-secure-alpha5} follows from \eqref{check-secure-indep}. 
	Thus indeed $W_\cB$ reveals nothing about $M_\alpha$. 
	
	\item Assume we can access $N_\beta$ coded symbols $W_\cB, |\cB|=N_\beta$. 
	Partition $\cB$ into $\cB_1$ and $\cB_2$ such that $\cB_1\subseteq\{1,2,\cdots,\theta\}$ and $\cB_2\subseteq\{\theta+1,\cdots,L\}$. 
	We first consider
	\begin{align}
	&H(Y_\alpha^{\cB_2}|Y_\alpha^{\cB_1} Y_\beta^{*\cB_2})  \nonumber \\
	&\geq H(Y_\alpha^{\cB_2}|Y_\alpha^{\cB_1} Y_\beta^{*\cB_2}M_\beta)  \label{check-secure-entropy1-1} \\
	&\geq H(Y_\alpha^{\cB_2}|Y_\alpha^1\cdots Y_\alpha^{\theta}, Z_\beta^{\theta+1}\cdots Z_\beta^{N_\beta}, M_{\beta} Y_\beta^{*\cB_2})  \label{check-secure-entropy1-2} \\
	&=H(Y_\alpha^{\cB_2}|Y_\alpha^1\cdots Y_\alpha^{\theta}, Z_\beta^{\theta+1}\cdots Z_\beta^{N_\beta}, M_{\beta})  \label{check-secure-entropy1-3} \\ 
	&=H(Y_\alpha^{\cB_2},Y_\alpha^1\cdots Y_\alpha^{\theta}|Z_\beta^{\theta+1}\cdots Z_\beta^{N_\beta}, M_{\beta})  \nonumber \\
	&  -H(Y_\alpha^1\cdots Y_\alpha^{\theta}|Z_\beta^{\theta+1}\cdots Z_\beta^{N_\beta}, M_{\beta})   \\ 
	&=H(Y_\alpha^{\cB_2},Y_\alpha^1\cdots Y_\alpha^{\theta})-H(Y_\alpha^1\cdots Y_\alpha^{\theta})   \label{check-secure-entropy1-5} \\
	&=H(Y_\alpha^{\cB_2}), \label{check-secure-entropy1-6}
	\end{align}
	where both \eqref{check-secure-entropy1-1} and \eqref{check-secure-entropy1-2} follow from the fact that conditioning does not increase entropy, 
	\eqref{check-secure-entropy1-3} follows from that $Y_\beta^{*\cB_2}$ is a function of 
	$(Y_\alpha^1\cdots Y_\alpha^{\theta}, Z_\beta^{\theta+1} \cdots Z_\beta^{N_\beta},M_{\beta})$, 
	\eqref{check-secure-entropy1-5} follows from that $(Z_\beta^{\theta+1} \cdots Z_\beta^{N_\beta}, M_{\beta})$ 
	are independent of $(Y_\alpha^{\cB_2},Y_\alpha^1\cdots Y_\alpha^{\theta})$, 
	and the last equality follows from \Cref{lemma-property-0}, 
	since $|\cB_2|+\theta\leq \alpha$ which is induced by $\theta\leq \alpha-N_\beta$. 	
	Since conditioning does not increase entropy, in light of \eqref{check-secure-entropy1-6}, we obtain  
	\begin{align}
	&H(Y_\alpha^{\cB_2}|Y_\alpha^{\cB_1} Y_\beta^{*\cB_2}M_\beta)\notag\\
	&=H(Y_\alpha^{\cB_2}|Y_\alpha^{\cB_1} Y_\beta^{*\cB_2})=H(Y_\alpha^{\cB_2}).   \label{check-secure-entropy1}
	\end{align}
	Then we have 
	\begin{align}
	&I(W_\cB;M_\beta)  =I(W_{\cB_1}W_{\cB_2};M_\beta)   \\
	&=I(Y_\alpha^{\cB_1}~Y_\alpha^{\cB_2}Y_\beta^{*\cB_2};M_\beta)  \\
	&=I(Y_\alpha^{\cB_1}Y_\beta^{*\cB_2};M_\beta)+I(Y_\alpha^{\cB_2};M_\beta|Y_\alpha^{\cB_1}Y_\beta^{*\cB_2})  \\
	&=I(Y_\alpha^{\cB_2};M_\beta|Y_\alpha^{\cB_1}Y_\beta^{*\cB_2})  \label{check-secure-beta4} \\
	&=H(Y_\alpha^{\cB_2}|Y_\alpha^{\cB_1} Y_\beta^{*\cB_2})-H(Y_\alpha^{\cB_2}|Y_\alpha^{\cB_1}Y_\beta^{*\cB_2}M_\beta)   \\
	&=H(Y_\alpha^{\cB_2})-H(Y_\alpha^{\cB_2})\label{check-secure-beta6} \\
	&=0,
	\end{align}
	where \eqref{check-secure-beta4} follows from \Cref{lemma-property-MDS-1} and the fact that $|\cB_1|+|\cB_2|=N_\beta$, 
	and \eqref{check-secure-beta6} follows from \eqref{check-secure-entropy1}. 
	Thus we obtain that $W_\cB$ reveals nothing about $M_\beta$. 
\end{enumerate}

\subsection{Reversed Security Level}\label{section-scheme-beta-key-alpha}
We next provide a joint coding strategy for the case that Condition 2 in \eqref{not-opt-condition-2} holds. 
\begin{example}\label{example-4-3}
	Let $L=4, (\alpha, \beta)=(3,4), (m_3,m_4)=(1,1), (N_3,N_4)=(2,1)$, and $p=11$. 
	We use generator matrix $G_3$ generated using MDS-B to encode $M_3$ separately with encryption keys $Z_1,Z_2$, as given in \eqref{code-example-4-3-sup}. 
	The joint coding strategy is simply to use $M_4$ to replace $Z_1$ as secret keys to encrypt $M_3$, as given in \eqref{code-example-4-3-joint}. 
	\begin{align}
	&M_3+2{\color{red}Z_1}+9Z_2, ~9M_3+8{\color{red}Z_1}+6Z_2,  \nonumber \\
	&6M_3+10{\color{red}Z_1}+7Z_2, ~7M_3+9{\color{red}Z_1}+7Z_2;  \label{code-example-4-3-sup} \\
	\longrightarrow &M_3+2{\color{red}M_4}+9Z_2, ~9M_3+8{\color{red}M_4}+6Z_2, \nonumber \\
	&6M_3+10{\color{red}M_4}+7Z_2, ~7M_3+9{\color{red}M_4}+7Z_2.  \label{code-example-4-3-joint}
	\end{align}
	
	Since $M_4$ does not need to be separately encoded, rate saving is obtained. 
	The reconstruction and security requirements of $M_3$ are immediate. 
	The reconstruction requirement of $M_4$ is straightforward since everything is recovered with any three coded symbols. 
	The security requirement of $M_4$ can be easily seen that any one coded symbol reveals nothing about $M_4$.

\end{example}

\noindent{\textbf{Coding strategy for general parameters:}}

Next, we present the general coding strategy that $M_{\beta}$ performs as secret keys for $M_{\alpha}$ so that we can reduce the coding rates. Let $G_{\alpha}$ be a generator matrix generated using MDS-B in \Cref{section-MDS}, which can be used to encode $M_\alpha$ separately with encryption keys $(Z_1,Z_2,\ldots,Z_{N_\alpha})$. The joint coding strategy is simply to use $\eta=\min\{N_\alpha,\alpha-N_\beta\}$ symbols of the message $M_\beta$ (i.e., $M_\beta^1,M_\beta^2,\cdots,M_\beta^{\eta}$) to replace the encryption keys $(Z_1~Z_2~\cdots~ Z_{\eta})$ for encrypting $M_{\alpha}$. 
The parameter $\eta$ is strictly positive, which is implied by Condition 2 in \eqref{not-opt-condition-2} as well as $\alpha>N_\alpha$. 
Denote the corresponding coded symbols for $M_\alpha$ after this replacement as $(Y_\alpha^{1*},Y_\alpha^{2*},\cdots,Y_\alpha^{L*})$. Since the $\eta$ message symbols of $M_\beta$ do not need to be separately encoded, rate saving is thus obtained. 
Next, we verify the reconstruction and security constraints.

\vspace{0.2cm}
\noindent{\textit{Reconstruction: }} By the code construction in \Cref{section-MDS}, both the message $M_\alpha$ and the keys $M_\beta$ can be losslessly recovered from any $\alpha$ coded symbols. 
Since $\alpha<\beta$, the reconstruction requirements of both $M_\alpha$ and $M_\beta$ are satisfied immediately. 

\vspace{0.2cm}
\noindent{\textit{Security: }} The security constraint of $M_\alpha$ is straightforward, and thus let us consider $M_\beta$. For any $\cB\subseteq\cL$ such that $|\cB|=N_\beta$, let $Y_\alpha^{*\cB}=(Y_\alpha^{i*}:i\in \cB)$. 
By \Cref{lemma-property-MDS-2}, we have 
\begin{equation}
I(Y_\alpha^{*\cB};M_\beta^{1},M_\beta^2,\cdots,M_\beta^{\eta})=0,
\end{equation}
since $\eta\leq \alpha-N_\beta$.

\section{Converse of \Cref{thm-sup-opt-condition}}\label{section-conditions-converse}
To show the optimality of \Cref{thm-sup-opt-condition}, we only need to prove that under the condition in \eqref{sup-opt-condition}, the sum rate is lower bounded by \eqref{sup-sumrate}, i.e.,  
\begin{equation}
\sum_{l=1}^L\sR_l\geq \sum_{\alpha=1}^{L}\frac{L\sm_\alpha}{\alpha-N_\alpha}.   \label{sum-rate-bound}
\end{equation}
For any $\alpha\in\cL$, let $\bB_\alpha$ be the set of \textit{disjoint subset} pairs $(\cB_\alpha^1,\cB_\alpha^2)$ such that $\cB_\alpha^1,\cB_\alpha^2\subseteq\cL$, 
\begin{equation}
|\cB_{\alpha}^1|=\alpha-N_\alpha \text{ and } |\cB_{\alpha}^2|=N_\alpha. 
\end{equation}
For $\alpha\in\cL$, let $M_{1:\alpha}\triangleq(M_1,M_2,\cdots,M_\alpha)$. 
Define $\mu_\alpha$ by 
\begin{align}
\mu_\alpha&=\frac{L}{\alpha-N_\alpha}\frac{1}{{L\choose N_\alpha} {L-N_\alpha\choose \alpha-N_\alpha}}\sum_{(\cB_\alpha^1,\cB_\alpha^2)\in\bB_\alpha} \hspace{-1.5em}H(W_{\cB_{\alpha}^1}| W_{\cB_{\alpha}^2}M_{1:\alpha}). \label{def-mu_alpha}
\end{align}
We need the following lemma to proceed. 
\begin{lemma}\label{lemma-converse_condition-iteration}
	Under the condition in \eqref{sup-opt-condition}, for any $\alpha\in\cL$, we have 
	\begin{equation}
	\sum_{l=1}^LH(W_l)\geq \sum_{j=1}^{\alpha}\frac{Lm_j}{j-N_j}+\mu_\alpha.  \label{claim}
	\end{equation}
\end{lemma}
\begin{proof}
	For $\alpha\leq T_s$, \eqref{claim} is simply the inequality (27) in \cite{yeung97}. 
	For $\alpha\geq T_s$, we prove the lemma by induction on $\alpha$. 
	Similar to the proof of Theorem 2 in \cite{yeung97} where Han's inequality plays a key role, we apply Han's inequality and its complementary conditioning version. The details of the proof can be found in Appendix~\ref{proof_lemma-converse_condition-iteration}.
\end{proof}
For $\alpha=L$, in light of \eqref{claim}, we have 
\begin{equation}
\sum_{l=1}^LR_l=\sum_{l=1}^LH(W_l)\geq \sum_{\alpha=1}^{L}\frac{L m_\alpha}{\alpha-N_\alpha}+\mu_L\geq \sum_{\alpha=1}^{L}\frac{L m_\alpha}{\alpha-N_\alpha},
\end{equation}
from which we can obtain, by normalization, the sum rate bound \eqref{sum-rate-bound}. 
\begin{remark}
	It is clear that superposition coding must induce $\mu_L=0$ under the condition in \eqref{sup-opt-condition}.
	Since the messages are encoded separately, we can indeed verify that for any $\alpha\in\cL$, 
	\begin{equation}
	H(Y_{\alpha}^{\cB_{L}^1}|Y_{\alpha}^{\cB_{L}^2}M_{\alpha})=0, \label{vanish-check}
	\end{equation}
	where $Y_{\alpha}^1,Y_{\alpha}^2,\cdots,Y_{\alpha}^L$ are coded symbols of $M_{\alpha}$ and $Y_{\alpha}^{\cB}\triangleq (Y_{\alpha}^i:i\in\cB)$ for any $\cB\subseteq\cL$. 
	To see this, observe that if the weakly secure SMDC problem reduces to classical SMDC, \eqref{vanish-check} is true immediately.
	Otherwise, by \eqref{sup-opt-condition}, we have $N_L\geq N_\alpha$ for any $\alpha\in\cL$. 
	Since we use an $(N_\alpha,\alpha,L)$ ramp secret sharing code to encode $M_\alpha$, any $\alpha$ symbols from the set $\{M_\alpha^1,M_\alpha^2,\cdots,M_\alpha^{\alpha-N_\alpha}, Y_{\alpha}^1,Y_{\alpha}^2,\cdots,Y_{\alpha}^L\}$ can completely recover the whole set. 
	Thus, $(Y_{\alpha}^{\cB_{L}^2},M_{\alpha})$ provide complete information about $Y_{\alpha}^{\cB_{L}^1}$, which verifies \eqref{vanish-check}. 
\end{remark}

\section{Converse Proof of \Cref{thm-group-pairwise-region}}\label{section-DS-proof}
Before proving \Cref{thm-group-pairwise-region}, we introduce some terminologies and notations in \cite{yeung99}. 
Let $\bm{\lambda}=(\lambda_1,\lambda_2,\cdots,\lambda_L)$ and
\begin{equation}
\mathbb{R}_+^L=\{\bm{\lambda}:~\bm{\lambda}\neq\bm{0} \text{ and } \lambda_i\in\mathbb{R},\lambda_i\geq 0\text{ for }i\in\cL\}.   \label{def-R-set}
\end{equation}
Let $\Omega_L^{\alpha}=\left\{\bm{v}\in\{0,1\}^L:|\bm{v}|=\alpha\right\}$, 
where $|\bm{v}|$ is the Hamming weight of a vector $\bm{v}=(v_1,v_2,\cdots,v_L)$. 
For any $\bm{v}\in\Omega_L^{\alpha}$, let $c_{\alpha}(\bm{v})$ be any nonnegative real number. 
For any $\bm{\lambda}\in\mathbb{R}_+^L$ and $\alpha\in\cL$, let $f_{\alpha}(\bm{\lambda})$ be the optimal solution to the following optimization problem: 
\begin{eqnarray}
f_{\alpha}(\bm{\lambda})\triangleq&\max&\sum_{\bm{v}\in\Omega_L^{\alpha}}c_{\alpha}(\bm{v})  \label{optimization-1} \\
&\text{s.t.}&\sum_{\bm{v}\in\Omega_L^{\alpha}}c_{\alpha}(\bm{v})\cdot \bm{v}\leq \bm{\lambda}   \label{optimization-2}\\
&&c_{\alpha}(\bm{v})\geq 0,\forall \bm{v}\in\Omega_L^{\alpha}.  \label{optimization-3}
\end{eqnarray}
A set $\left\{c_{\alpha}(\bm{v}): \bm{v}\in\Omega_L^{\alpha}\right\}$ is called 
an $\alpha$-\textit{resolution} for $\bm{\lambda}$ if \eqref{optimization-2} and \eqref{optimization-3} are satisfied 
and it will be abbreviated as $\{c_{\alpha}(\bm{v})\}$ if there is no ambiguity. 
Furthermore, an $\alpha$-resolution is called \textit{optimal} if it achieves the optimal value $f_{\alpha}(\bm{\lambda})$. 
In the following proof, we will take advantage of some lemmas and theorems from \cite{yeung99} and \cite{guo-yeung-SMDC-IT20}, which are enclosed in Appendix~\ref{referenced-lemmas/theorems} for convenience. 

To prove the converse of \Cref{thm-group-pairwise-region}, we follow the idea of Theorem 2 in \cite{yeung99}, 
i.e., we provide an alternative characterization of the group pairwise region $\cR_{\text{gp}}^{L,r}$. 
For simplicity, let $f_{L+1}(\bm{\lambda})=0$ for all $\bm{\lambda}\in\mathbb{R}_+^L$. For $\eta\in\{r+1,r+2,\cdots,L+1\}$, let
\begin{align}
g_{\eta}(\bm{\lambda})&=\sum_{\alpha=1}^{r}f_1(\bm{\lambda})\sm_\alpha+ \sum_{\alpha=\eta+1}^{L}f_{\alpha}(\bm{\lambda})\sm_\alpha \nonumber \\
&\quad +f_{\eta}(\bm{\lambda}) \left[\sum_{\alpha=r+1}^{\eta}\sm_\alpha-\sum_{\alpha=1}^{r}(\alpha-1)\sm_\alpha\right].   \label{def-g}
\end{align}
In particular, for $\eta=\eta^*$ which is defined by \eqref{opt-eta}, we have
\begin{equation}
g_{\eta^*}(\bm{\lambda})=\sum_{\alpha=1}^{r}f_1(\bm{\lambda}) \sm_\alpha+\sum_{\alpha=r+1}^{L}f_{\alpha}(\bm{\lambda}) \sm_\alpha^*,
\end{equation}
where $\sm_\alpha^*$ is defined in \eqref{def-sm*}. 
From the group pairwise coding scheme in \reffig{fig_group-pairwise-scheme}, we have the following intuitions on the coding rates. 
\begin{enumerate}[i.]
	\item Superposition of $M_1,M_2,\cdots,M_{r}$ induces the rate 
	\begin{equation}
	\sum_{l=1}^L\lambda_l\sR_l=\sum_{\alpha=1}^{r}f_1(\bm{\lambda}) \sm_\alpha.  \label{subrate-sup-1}
	\end{equation}
	
	\item The messages $M_{r+1}, M_{r+2}, \cdots, M_{\eta^*-1}, M_{\eta^*}^1$ perform as keys for $M_1,M_2,\cdots,M_{r}$. 
	Thus, we do not need extra rates to encode them beyond the rate given in \eqref{subrate-sup-1}.
	
	\item The other messages $M_{\eta^*}^2, M_{\eta^*+1},\cdots, M_{L}$ will be encoded in the same way as in classical SMDC, i.e., superposition coding. 
	The coding rate is characterized in~\cite{yeung99} using the technique of $\alpha$-resolution, which is 
	\begin{align}
	\sum_{l=1}^L\lambda_l\sR_l&=f_{\eta^*}(\bm{\lambda})\left[\sum_{\alpha=r+1}^{\eta^*}\sm_\alpha-\sum_{\alpha=1}^{r}(\alpha-1)\sm_\alpha\right]  \nonumber \\
	&\quad +\sum_{\alpha=\eta^*+1}^{L}f_{\alpha}(\bm{\lambda})\sm_\alpha. \label{subrate-sup-2}
	\end{align}
\end{enumerate}
Summing up the rates in \eqref{subrate-sup-1} and \eqref{subrate-sup-2}, 
we obtain $g_{\eta^*}(\bm{\lambda})$ which is the total rate of group pairwise coding. 
Let $\cR_{L,r}^*$ be the set of all $\bm{\sR}\geq \bm{0}$ such that 
\begin{equation}
\bm{\lambda}\cdot\bm{\sR}\geq g_{\eta^*}(\bm{\lambda}). \label{equal-region}
\end{equation}
In particular, for $\bm{\lambda}=(100\cdots)$ and $\eta^*=L+1$, 
the constraint in \eqref{equal-region} becomes the single rate bound 
\begin{equation}
\sR_l\geq \sum_{\alpha=1}^{r}\sm_\alpha.   \label{equal-region-L+1}
\end{equation}
For $\bm{\lambda}=\bm{1}$, the constraint in \eqref{equal-region} becomes the sum rate bound 
\begin{equation}
\sR_{\text{sum}}^*=\sum_{\alpha=1}^r(L-\alpha+1)\sm_\alpha+\sum_{\alpha=r+1}^{\eta^*}\frac{L\sm_\alpha}{\eta^*}+\sum_{\alpha=\eta^*+1}^L\frac{L\sm_\alpha}{\alpha}.   \label{equal-region-L}
\end{equation}
For $\eta^*=r+1$, the constraint becomes
\begin{equation}
\bm{\lambda}\cdot\bm{\sR}\geq \sum_{\alpha=1}^{r}\left[f_1(\bm{\lambda})-(\alpha-1) f_{r+1}(\bm{\lambda})\right] \sm_{\alpha}+ \sum_{\alpha=r+1}^{L}f_{\alpha}(\bm{\lambda}) \sm_{\alpha}.  \label{equal-region-1} 
\end{equation}

Inspired by the above intuitions on the group pairwise coding rates, 
we can alternatively characterize $\cR_{\text{gp}}^{L,r}$ in another equivalent form, given in the following theorem. 
\begin{theorem}\label{thm-group-pairwise-region-alternative}
	$\cR_{\text{gp}}^{L,r}=\cR_{L,r}^*$.  
\end{theorem}
\begin{proof}
	See Appendix \ref{proof_thm-group-pairwise-region-alternative}. 
\end{proof}

To complete the converse proof of \Cref{thm-group-pairwise-region}, in light of the fact $\cR_{\text{gp}}^{L,r}\subseteq\cR_{\text{L,r}}$ as well as \Cref{thm-group-pairwise-region-alternative}, we now only need to show $\cR_{\text{L,r}}\subseteq\cR_{L,r}^*$, i.e., for any $\bm{\sR}\in\cR_{\text{L,r}}$, the following inequality holds
\begin{equation}
\bm{\lambda}\cdot\bm{\sR}\geq g_{\eta^*}(\bm{\lambda}). 
\end{equation}
The following lemma provides an alternative representation of $g_{\eta^*}(\bm{\lambda})$. 
\begin{lemma}\label{lemma-eta*}
	$\max_{\eta=r+1,\cdots,L+1}\big\{g_{\eta}(\bm{\lambda})\big\}=g_{\eta^*}(\bm{\lambda}).$
\end{lemma}
\begin{proof}
	See Appendix \ref{proof_lemma-eta*}.
\end{proof}
By \Cref{lemma-eta*}, it only remains to show that for any $\bm{\sR}\in\cR_{\text{L,r}}$ and $\eta=r+1,\cdots,L+1$, the inequality $\bm{\lambda}\cdot\bm{\sR}\geq g_{\eta}(\bm{\lambda})$ holds. 
The converse for SMDC in~\cite{yeung99} is proved using iterations to extract the entropies $H(M_1),H(M_2),\cdots,H(M_L)$ successively with coefficient $f_1(\bm{\lambda}),f_2(\bm{\lambda}),\cdots,f_L(\bm{\lambda})$ which have the same form of expression. In the secure setting here, the desired inequality $\bm{\lambda}\cdot\bm{\sR}\geq g_{\eta}(\bm{\lambda})$ will have two forms of coefficients, i.e., coefficients related to the secure messages and those related to the non-secure messages. The latter is the same as that in~\cite{yeung99}, but the former is different. For this reason, the iterations in the converse proof in~\cite{yeung99} do not apply to the former, i.e., the secure messages. Therefore, we need to derive new iterations to extract the entropies of the secure messages, such that the $r$-th iteration can be connected with the iterations in~\cite{yeung99}. 
Specifically, the main idea of proving $\bm{\lambda}\cdot\bm{\sR}\geq g_{\eta}(\bm{\lambda})$ is as follows: 
\begin{enumerate}[i)]
	\item we extract the entropies $H(M_1),H(M_2),\cdots,H(M_L)$ with proper coefficients in \eqref{def-g} from $\sum_{l=1}^L\lambda_lH(W_l)$ successively and iteratively;
	
	\item when extracting $H(M_\alpha)$ for $\alpha\in\{1,2,\cdots,r\}$, we explicitly design the coefficients of each intermediate term in closed-form so that we can finally connect to the $r$-th iteration of the converse proof in~\cite{yeung99};
	
	\item for $\alpha\geq r+1$, since there is no security constraints, we simply use the iterations in~\cite{yeung99}. 
\end{enumerate}

One of the main contributions of the converse proof compared with that in~\cite{yeung99} is the new technique of explicitly designing the coefficients in closed-form in each iteration for the secure messages. In contrast, in each iteration of the non-secure messages which is simply the iteration in~\cite{yeung99}, the coefficients in the iteration do not have a closed-form.

Instead of formally proving this inequality here, we provide an example for $(L,r)=(4,2)$ and $\eta=3$ to illustrate the main idea, and relegate the formal proof to Appendix \ref{proof_DS-converse}. 
The connection between this example and the formal proof will be discussed in \Cref{remark-cof-gamma-lambda}, \Cref{remark-cof-c}, and \Cref{remark-connection} in Appendix \ref{proof_DS-converse}.  
For different $i,j,k\in\{1,2,3,4\}$, we first present two equalities that will be used in the example: 
\begin{align}
H(W_i|W_kM_1)&=H(M_2|W_kM_1)+H(W_i|W_kM_{1:2}) \nonumber \\
&\qquad -H(M_2|W_iW_kM_1)  \nonumber \\
&=H(M_2)+H(W_i|W_kM_{1:2}),   \label{example-pre1}  \\
H(W_iW_jW_k|M_{1:2})&=H(W_iW_jW_kM_3|M_{1:2})  \nonumber \\
&=H(M_3)+H(W_iW_jW_kM_{1:3}).   \label{example-pre2}
\end{align}
Now we can write the following chain of inequalities without much difficulty: 
\begin{align}
&R_1+R_2+R_3+R_4  \nonumber \\
&=H(W_1)+H(W_2)+H(W_3)+H(W_4)  \\
&=4H(M_1)+H(W_1|M_1)+H(W_2|M_1)  \nonumber \\
&\quad +H(W_3|M_1)+H(W_4|M_1)   \label{converse-example-1}\\
&=4H(M_1)+\bigg\{0H(W_1|M_1)+\frac{2}{3}H(W_2|M_1)  \nonumber \\
&\quad +H(W_3|M_1)+H(W_4|M_1)\bigg\}_{\triangleq S_1}+\bigg\{H(W_1|M_1)  \nonumber \\
&\quad +\frac{1}{3}H(W_2|M_1)+0H(W_3|M_1)+0H(W_4|M_1) \bigg\}_{\triangleq S_2}  \label{converse-example-2} \\
&\geq 4H(M_1)+\frac{8}{3}H(M_2)  \nonumber \\
&\quad +\bigg\{\frac{1}{3}H(W_2W_3|W_1M_{1:2})+\frac{1}{3}H(W_2W_4|W_1M_{1:2})  \nonumber \\
&\quad +\frac{1}{3}H(W_3W_4|W_1M_{1:2})+\frac{1}{3}H(W_3W_4|W_2M_{1:2})\bigg\}  \nonumber \\
&\quad +\left\{H(W_1|M_{1:2})+\frac{1}{3}H(W_2|M_{1:2})\right\}  \label{converse-example-5} \\
&= 4H(M_1)+\frac{8}{3}H(M_2)  \nonumber \\
&\quad +\frac{1}{3}H(W_1W_2W_3|M_{1:2})+\frac{1}{3}H(W_1W_2W_4|M_{1:2})  \nonumber \\
&\quad +\frac{1}{3}H(W_1W_3W_4|M_{1:2})+\frac{1}{3}H(W_2W_3W_4|M_{1:2})  \label{converse-example-6}\\
&\stackrel{\eqref{example-pre2}}{=} 4H(M_1)+\frac{8}{3}H(M_2)+\frac{4}{3}H(M_3)  \nonumber \\
&\quad +\frac{1}{3}H(W_1W_2W_3|M_{1:3})+\frac{1}{3}H(W_1W_2W_4|M_{1:3})  \nonumber \\
&\quad +\frac{1}{3}H(W_1W_3W_4|M_{1:3})+\frac{1}{3}H(W_2W_3W_4|M_{1:3})  \label{converse-example-7}\\
&\geq 4H(M_1)+\frac{8}{3}H(M_2)+\frac{4}{3}H(M_3)+H(M_4)  \label{converse-example-8}\\
&=4m_1+\frac{8}{3}m_2+\frac{4}{3}m_3+m_4,
\end{align}
where \eqref{converse-example-8} follows from the fact that $(\frac{1}{3},\frac{1}{3},\frac{1}{3},\frac{1}{3})$ is an optimal $3$-resolution for $\bm{\lambda}=(1,1,1,1)$ (\textit{cf.} (\ref{optimization-1})-(\ref{optimization-3})),  
and the nontrivial step from \eqref{converse-example-2} to \eqref{converse-example-5} can be derived as follows 
\begin{align}
S_1&=\frac{2}{3}H(W_2|M_1)+H(W_3|M_1)+H(W_4|M_1)  \\
&=\left[\frac{1}{3}H(W_2|M_1)+\frac{1}{3}H(W_3|M_1)\right]  \nonumber \\
&\quad +\left[\frac{1}{3}H(W_2|M_1)+\frac{1}{3}H(W_4|M_1)\right]  \nonumber \\
&\quad +\left[\frac{1}{3}H(W_3|M_1)+\frac{1}{3}H(W_4|M_1)\right]\nonumber \\
&\quad +\left[\frac{1}{3}H(W_3|M_1)+\frac{1}{3}H(W_4|M_1)\right]   \label{converse-example-S1-1} \\
&\geq \left[\frac{1}{3}H(W_2|W_1M_1)+\frac{1}{3}H(W_3|W_1M_1)\right]\nonumber \\
&\quad +\left[\frac{1}{3}H(W_2|W_1M_1)+\frac{1}{3}H(W_4|W_1M_1)\right]  \nonumber \\
&\quad+\left[\frac{1}{3}H(W_3|W_1M_1)+\frac{1}{3}H(W_4|W_1M_1)\right]\nonumber \\
&\quad +\left[\frac{1}{3}H(W_3|W_2M_1)+\frac{1}{3}H(W_4|W_2M_1)\right]   \label{converse-example-S1-2} \\
&\stackrel{\eqref{example-pre1}}{=}\frac{8}{3}H(M_2)+\left[\frac{1}{3}H(W_2|W_1M_{1:2})+\frac{1}{3}H(W_3|W_1M_{1:2})\right]\nonumber \\
&\quad +\left[\frac{1}{3}H(W_2|W_1M_{1:2})+\frac{1}{3}H(W_4|W_1M_{1:2})\right]  \nonumber \\
&\quad+\left[\frac{1}{3}H(W_3|W_1M_{1:2})+\frac{1}{3}H(W_4|W_1M_{1:2})\right]\nonumber \\
&\quad +\left[\frac{1}{3}H(W_3|W_2M_{1:2})+\frac{1}{3}H(W_4|W_2M_{1:2})\right]    \label{converse-example-S1-3} \\
&\geq \frac{8}{3}H(M_2)\nonumber \\
&\quad +\bigg\{\frac{1}{3}H(W_2W_3|W_1M_{1:2})+\frac{1}{3}H(W_2W_4|W_1M_{1:2})  \nonumber \\
&\quad +\frac{1}{3}H(W_3W_4|W_1M_{1:2})+\frac{1}{3}H(W_3W_4|W_2M_{1:2})\bigg\}_{\triangleq S_1'},   \label{converse-example-S1-4}
\end{align}
and 
\begin{align}
S_2&=H(W_1|M_1)+\frac{1}{3}H(W_2|M_1) \\
&\geq \left\{H(W_1|M_{1:2})+\frac{1}{3}H(W_2|M_{1:2})\right\}_{\triangleq S_2'}. 
\end{align}

The main ideas of the example are as follows:
\begin{enumerate}
	\item The two terms $S_1$ and $S_2$ have a similar form in \eqref{converse-example-2}, but with different coefficient vectors $(0,\frac{2}{3},1,1)$ and $(1,\frac{1}{3},0,0)$, respectively, which are chosen strategically for this bound. The two terms are bounded in rather different manners. We extract $\frac{8}{3}H(M_2)$ from $S_1$ (with $S_1'$ left) and use $S_2$ to convert $S_1'$ from the form $H(W_iW_j|W_kM_{1:2})$ to the form $H(W_iW_jW_k|M_{1:2})$. This further generates the terms $H(W_iW_jW_k|M_{1:3})$ in \eqref{converse-example-7}, which ensures that the $\alpha$-resolution technique can be applied subsequently. 
	\item When bounding $S_1$, we reorganize its coefficient vector $(0,\frac{2}{3},1,1)$ as given  in \eqref{converse-example-S1-1} for two purposes: firstly, Han's inequality can be applied as in \eqref{converse-example-S1-4};  secondly, $H(W_iW_j|W_kM_{1:2})$ in $S_1'$ can be converted to $H(W_1W_jW_k|M_{1:2})$ using $S_2'$; 
	\item The coefficient $(\frac{1}{3},\frac{1}{3},\frac{1}{3},\frac{1}{3})$ in \eqref{converse-example-6}-\eqref{converse-example-7} is an optimal $3$-resolution for $(1,1,1,1)$. In the general proof, the $\alpha$-resolution technique used in the converse proof for SMDC~\cite{yeung99} will be invoked in a more systematic manner. 
\end{enumerate}

\section{Conclusion}\label{section-conclusion}
We studied the weakly secure SMDC problem and characterized the condition that superposition coding is optimal in terms of achieving the minimum sum rate. 
It is generally difficult to design the optimal coding schemes and characterize the rate regions for those cases that superposition is suboptimal.
In this paper, we consider a special case called differential-constant secure SMDC, for which the optimal rate region is characterized. 
A group pairwise coding scheme is shown to be optimal in terms of achieving the entire rate region. 

The optimality condition is proved only for the minimum sum rate, 
we conjecture that it is also the optimality condition that superposition coding can achieve the entire rate region. 
This is currently under our investigation.

\begin{appendices}
	\section{Some lemmas/theorems from \cite{yeung99,guo-yeung-SMDC-IT20}}\label{referenced-lemmas/theorems}
	In the following lemmas and theorem, we assume $\bm{\lambda}\in\mathbb{R}_+^L$ (c.f. \eqref{def-R-set}) is ordered, i.e., $\lambda_1\geq \lambda_2\geq \cdots\geq \lambda_L$. Let $\{c(\bm{v})\}$ be an $\alpha$-resolution for $\bm{\lambda}$ (c.f. \eqref{optimization-2},\eqref{optimization-3}) and $\tilde{\bm{\lambda}}=\sum_{\bm{v}\in\Omega_L^{\alpha}}c(\bm{v})\cdot \bm{v}$.  
	An $\alpha$-resolution is called \textit{perfect} if the equality in \eqref{optimization-2} holds, i.e., $\sum_{\bm{v}\in\Omega_L^{\alpha}}c(\bm{v})\cdot \bm{v}=\bm{\lambda}$. 
	
	\noindent\textbf{Lemma 2} in \cite{yeung99}: \textit{Let $\{c(\bm{v})\}$ be an optimal $\alpha$-resolution for~$\bm{\lambda}$. Then there exists $0\leq l\leq \alpha-1$ such that $\lambda_i-\tilde{\lambda}_i>0$ if and only if $1\leq i\leq l$. }
	
	\noindent\textbf{Lemma 4} in \cite{yeung99}: \textit{
	\begin{enumerate}[(i)]
		\item $f_{\alpha}(\bm{\lambda})\leq \alpha^{-1}\sum_{i=1}^L\lambda_i$; 
		\item $\sum_{\bm{v}\in\Omega_L^{\alpha}}c(\bm{v})=\alpha^{-1}\sum_{i=1}^L\lambda_i$ if and only if $\{c(\bm{v})\}$ is a perfect $\alpha$-resolution for $\bm{\lambda}$. In this case, $f_{\alpha}(\bm{\lambda})= \alpha^{-1}\sum_{i=1}^L\lambda_i$. 
	\end{enumerate}}
	
	\noindent\textbf{Lemma 7} in \cite{yeung99}: \textit{For $\alpha\geq 2$, $\bm{\lambda}$ has a perfect $\alpha$-resolution if and only if $\lambda_1\leq \frac{\lambda_2+\cdots+\lambda_L}{\alpha-1}$. }
	
	\noindent\textbf{Theorem 1} in \cite{guo-yeung-SMDC-IT20}: \textit{$f_{\alpha}(\bm{\lambda})=\min\limits_{\beta\in\{0,1,\cdots,\alpha-1\}}\frac{1}{\alpha-\beta}\sum_{i=\beta+1}^L\lambda_i$. }
	
	\noindent\textbf{Lemma 1} in \cite{guo-yeung-SMDC-IT20}: \textit{For $\alpha\geq 2$, if $\lambda_1\leq \frac{\lambda_2+\lambda_3+\cdots\lambda_L}{\alpha-1}$, then $f_{\alpha}(\bm{\lambda})=\frac{1}{\alpha}\sum_{i=1}^L\lambda_i$. }
	
	For any permutation $\omega$ on $\{1, 2,\cdots , L\}$, denote $\left(\lambda_{\omega(1)},\lambda_{\omega(2)},\cdots, \lambda_{\omega(L)}\right)$ by $\omega(\bm{\lambda})$. 
	
	\noindent\textbf{Lemma 2} in \cite{guo-yeung-SMDC-IT20}: \textit{$f_{\alpha}\big(\omega(\bm{\lambda})\big) = f_{\alpha}(\bm{\lambda})$ for any $\alpha\in\cL$.}
	
	\noindent\textbf{Lemma 5} in \cite{guo-yeung-SMDC-IT20}: \textit{Let $\bm{\lambda}_1 = (\lambda_{1,1},\lambda_{1,2},\cdots ,\lambda_{1,L})$ and $\bm{\lambda}_2 = (\lambda_{2,1}, \lambda_{2,2},\cdots , \lambda_{2,L})$ be two ordered vectors such that $\lambda_{1,1} > \lambda_{2,1}$ and $\lambda_{1,i} =\lambda_{2,i}$ for all $2\leq i\leq L$. 
	For any $\alpha_0\in\cL$, if $f_{\alpha_0}(\bm{\lambda}_1) = f_{\alpha_0}(\bm{\lambda}_2)$, then $f_{\alpha}(\bm{\lambda}_1) = f_{\alpha}(\bm{\lambda}_2)$ for all $\alpha\geq \alpha_0$. }
	
	Let $\bm{\lambda}^{[1]}$ be the length-$L$ vector with the first component being 1 and the rest being 0, i.e., $\bm{\lambda}^{[1]}=(1,0,0,\cdots,0)$. 
	
	\noindent\textbf{Lemma 6} in \cite{guo-yeung-SMDC-IT20}: \textit{If $\lambda_1>\sum_{i=2}^L\lambda_i$, let $\bm{\lambda}'=\left(\sum_{i=2}^L\lambda_i,\lambda_2,\lambda_3,\cdots,\lambda_L\right)$. Then for all $\alpha\in\cL$, $$f_{\alpha}(\bm{\lambda})=\left(\lambda_1-\sum_{i=2}^L\lambda_i\right)f_{\alpha}\left(\bm{\lambda}^{[1]}\right)+f_{\alpha}(\bm{\lambda}').$$}
	
	\noindent\textbf{Lemma 7} in \cite{guo-yeung-SMDC-IT20}: \textit{For any $\eta\in\{1,2,\cdots,L-1\}$, 
	\begin{enumerate}[(i)]
		\item if $\lambda_1\leq \frac{1}{\eta}\sum_{i=2}^L\lambda_i$, then $f_{\alpha}(\bm{\lambda})=\frac{1}{\alpha}\sum_{i=1}^L\lambda_i$ for $\alpha=1,2,\cdots,\eta+1$; 
		\item if $\lambda_1\geq \frac{1}{\eta}\sum_{i=2}^L\lambda_i$, then $f_{\alpha}(\bm{\lambda})=f_{\alpha-1}(\lambda_2,\lambda_3, \cdots,\lambda_L)$ for $\alpha=\eta+1,\eta+2,\cdots,L$. 
	\end{enumerate}}

	\section{Proof of \Cref{lemma-converse_condition-iteration}}\label{proof_lemma-converse_condition-iteration}
	When condition \eqref{sup-opt-condition} is satisfied, there must exist a $T_s$ as defined in \eqref{def-T_s}. 
	For $\alpha\leq T_s$, since $N_\alpha=0$, we have 
	\begin{equation}
	\mu_\alpha=\frac{L}{\alpha}\frac{1}{{L\choose \alpha}} \sum_{\cB_\alpha^1\subseteq\cL:|\cB_\alpha^1|=\alpha} H(W_{\cB_{\alpha}^1}|M_{1:\alpha}).
	\end{equation}
	The claim in \eqref{claim} is exactly inequality (27) in \cite{yeung97} which was proved by applying Han's inequality. 
	
	For any $\alpha\geq T_s$, we prove the claim by induction. 
	Firstly, the claim is true for $\alpha=T_s$.
	Then we assume the claim is true for $\alpha=\zeta$ for some $\zeta\geq T_s$. 
	We now show that it is true for $\alpha=\varphi$, which is the index of the next non-vanishing message following $M_\zeta$. 
	In light of \eqref{claim}, we only need to show that 
	\begin{equation}
	\mu_\zeta\geq\frac{L}{\varphi-N_\varphi}H(M_\varphi)+\mu_\varphi.   \label{iteration}
	\end{equation}
	Since now $\varphi>\zeta>0$, the first condition in \eqref{sup-opt-condition} must hold, i.e., 
	\begin{equation}
	N_\zeta<\zeta\leq N_\varphi<\varphi. 
	\end{equation}
	For any $(\cB_{\varphi}^1,\cB_{\varphi}^2)\in\bB_\varphi$ and $(\cB_{\zeta}^1,\cB_{\zeta}^2)\in\bB_\zeta$ such that $\cB_{\zeta}^2\subseteq \cB_{\varphi}^2$, 
	from the reconstruction and security constraints of $M_\varphi$, we obtain  
	\begin{align}
	H(M_\varphi)&=H(M_\varphi|W_{\cB_{\varphi}^2})-H(M_\varphi|W_{\cB_{\varphi}^1}W_{\cB_{\varphi}^2})  \\
	&=I(M_\varphi;W_{\cB_{\varphi}^1}|W_{\cB_{\varphi}^2})  \\
	&=H(W_{\cB_{\varphi}^1}|W_{\cB_{\varphi}^2})-H(W_{\cB_{\varphi}^1}|W_{\cB_{\varphi}^2}M_\varphi)  \\
	&=H(W_{\cB_{\varphi}^1}|W_{\cB_{\varphi}^2}M_{1:\zeta})-H(W_{\cB_{\varphi}^1}|W_{\cB_{\varphi}^2}M_{1:\varphi})   \label{ir-entropy-varphi-4} 
	\end{align}
	where \eqref{ir-entropy-varphi-4} follows from the fact that $N_\varphi\geq \zeta$ and the reconstruction constraints of $M_1,M_2,\cdots,M_\zeta$. 
	In the following, we prove the iteration of \eqref{iteration} in two different situations:
	\begin{enumerate}[i.]
		\item $\zeta-N_\zeta\leq \varphi-N_\varphi$;
		\item $\zeta-N_\zeta>\varphi-N_\varphi$.
	\end{enumerate}
	\begin{remark}
		It is easy to see by checking the following proof that the case of $\zeta-N_\zeta=\varphi-N_\varphi$ is compatible with both (i) and (ii).
	\end{remark}
	\textit{Case i. }\uline{$\zeta-N_\zeta\leq \varphi-N_\varphi$}:
	Consider the following, 
	\begin{align}
	\mu_\zeta&=\frac{L}{\zeta-N_\zeta}\frac{1}{{L \choose N_\zeta}{L-N_\zeta \choose \zeta-N_\zeta}} \cdot \sum_{(\cB_\zeta^1,\cB_\zeta^2)\in\bB_\zeta} H(W_{\cB_{\zeta}^1}| W_{\cB_{\zeta}^2}M_{1:\zeta})   \nonumber \\
	&\geq \frac{L}{\zeta-N_\zeta}\frac{1}{{L \choose N_\zeta}{L-N_\zeta \choose \zeta-N_\zeta}} \cdot \nonumber \\
	&\quad \sum_{(\cB_\zeta^1,\cB_\zeta^2)\in\bB_\zeta} \sum_{\substack{\cV\subseteq\cL\backslash(\cB_\zeta^1\cup\cB_\zeta^2):\\ |\cV|=N_\varphi-N_\zeta}} \frac{1}{{L-\zeta \choose N_\varphi-N_\zeta}} H(W_{\cB_{\zeta}^1}| W_{\cB_{\zeta}^2\cup\cV}M_{1:\zeta})  \label{ir-alpha>=varphi-2} \\
	&=\frac{L}{\zeta-N_\zeta}\frac{1}{{L \choose N_\zeta}{L-N_\zeta \choose \zeta-N_\zeta}}  \cdot \nonumber \\
	&\quad \sum_{\substack{\cB_\zeta^1\subseteq\cL: \\ |\cB_\zeta^1|=\zeta-N_\zeta}} \sum_{\substack{\cB_\varphi^2\subseteq\cL\backslash\cB_{\zeta}^1: \\|\cB_\varphi^2|=N_\varphi}} \frac{{N_\varphi \choose N_\zeta}} {{L-\zeta \choose N_\varphi-N_\zeta}} H(W_{\cB_{\zeta}^1}| W_{\cB_\varphi^2}M_{1:\zeta})   \nonumber \\
	&=\frac{L}{\zeta-N_\zeta}\frac{{L-N_\varphi-1 \choose \zeta-N_\zeta-1}}{{L \choose N_\zeta}{L-N_\zeta \choose \zeta-N_\zeta}} \frac{{N_\varphi \choose N_\zeta}} {{L-\zeta \choose N_\varphi-N_\zeta}} \cdot \nonumber \\
	&\quad  \sum_{\substack{\cB_\varphi^2 \subseteq\cL:\\ |\cB_\varphi^2|=N_\varphi}} \sum_{\substack{\cB_\zeta^1\subseteq\cL\backslash\cB_{\varphi}^2: \\|\cB_\zeta^1|=\zeta-N_\zeta}} \frac{1}{{L-N_\varphi-1 \choose \zeta-N_\zeta-1}} H(W_{\cB_{\zeta}^1}| W_{\cB_{\varphi}^2}M_{1:\zeta})    \nonumber \\
	&\geq \frac{L}{\zeta-N_\zeta}\frac{{L-N_\varphi-1 \choose \zeta-N_\zeta-1}}{{L \choose N_\zeta}{L-N_\zeta \choose \zeta-N_\zeta}} \frac{{N_\varphi \choose N_\zeta}}{{L-\zeta \choose N_\varphi-N_\zeta}}  \cdot \nonumber \\
	&\quad \sum_{(\cB_\varphi^1,\cB_\varphi^2)\in\bB_\varphi} \frac{1}{{L-N_\varphi-1 \choose \varphi-N_\varphi-1}} H(W_{\cB_{\varphi}^1}| W_{\cB_{\varphi}^2}M_{1:\zeta})   \label{ir-alpha>=varphi-5} \\
	&=\frac{L}{\varphi-N_\varphi}\frac{1}{{L\choose N_\varphi} {L-N_\varphi\choose \varphi-N_\varphi}} \cdot \sum_{(\cB_\varphi^1,\cB_\varphi^2)\in\bB_\varphi} H(W_{\cB_{\varphi}^1}| W_{\cB_{\varphi}^2}M_{1:\zeta})    \nonumber \\
	&=\frac{L}{\varphi-N_\varphi}\frac{1}{{L\choose N_\varphi} {L-N_\varphi\choose \varphi-N_\varphi}}  \cdot \nonumber \\
	&\quad \sum_{(\cB_\varphi^1,\cB_\varphi^2)\in\bB_\varphi} \Big[H(M_\varphi)+ H(W_{\cB_{\varphi}^1}| W_{\cB_{\varphi}^2}M_{1:\varphi})\Big]  \label{ir-alpha>=varphi-7} \\
	&=\frac{L}{\varphi-N_\varphi}H(M_\varphi)+\mu_\varphi,
	\end{align}
	where \eqref{ir-alpha>=varphi-2} follows from the fact that conditioning does not increase entropy,
	\eqref{ir-alpha>=varphi-5} follows from Han's inequality and the assumption that $\zeta-N_\zeta\leq \varphi-N_\varphi$, 
	and \eqref{ir-alpha>=varphi-7} follows from \eqref{ir-entropy-varphi-4}. 
	
	\textit{Case ii. }\uline{$\zeta-N_\zeta>\varphi-N_\varphi$}:
	We derive the iteration as follows, 
	\begin{align}
	\mu_\zeta&=\frac{L}{\zeta-N_\zeta}\frac{1}{{L \choose N_\zeta}{L-N_\zeta \choose \zeta-N_\zeta}} \cdot \sum_{(\cB_\zeta^1,\cB_\zeta^2)\in\bB_\zeta} H(W_{\cB_{\zeta}^1}| W_{\cB_{\zeta}^2}M_{1:\zeta})   \nonumber \\
	&\geq \frac{L}{\zeta-N_\zeta}\frac{1}{{L \choose N_\zeta}{L-N_\zeta \choose \zeta-N_\zeta}} \cdot \nonumber \\
	&\quad \sum_{(\cB_\zeta^1,\cB_\zeta^2)\in\bB_\zeta} \sum_{\substack{\cV\subseteq\cL\backslash(\cB_\zeta^1\cup\cB_\zeta^2):\\ |\cV|=\varphi-\zeta}} \frac{1}{{L-\zeta \choose \varphi-\zeta}} H(W_{\cB_{\zeta}^1}| W_{\cB_{\zeta}^2\cup\cV}M_{1:\zeta})    \nonumber \\
	&=\frac{L}{\zeta-N_\zeta}\frac{1}{{L \choose N_\zeta}{L-N_\zeta \choose \zeta-N_\zeta}} \cdot \nonumber \\
	&\quad  \sum_{\substack{\cB_\zeta^1 \subseteq\cL:\\ |\cB_\zeta^1|=\zeta-N_\zeta}} \sum_{\substack{\cV'\subseteq\cL\backslash\cB_{\zeta}^1: \\|\cV'|=\varphi-(\zeta-N_\zeta)}} \frac{{\varphi-(\zeta-N_\zeta) \choose N_\zeta}} {{L-\zeta \choose \varphi-\zeta}} H(W_{\cB_{\zeta}^1}| W_{\cV'}M_{1:\zeta})    \nonumber \\
	&=\frac{L}{\zeta-N_\zeta}\frac{1}{{L \choose N_\zeta}{L-N_\zeta \choose \zeta-N_\zeta}} \frac{{\varphi-(\zeta-N_\zeta) \choose N_\zeta}}{{L-\zeta \choose \varphi-\zeta}}  \cdot \nonumber \\
	&\quad \sum_{\substack{\cD\subseteq\cL:\\|\cD|=\varphi}} \sum_{\substack{\cB_\zeta^1\subseteq\cD:\\|\cB_{\zeta}^1|=\zeta-N_\zeta}} \sum_{\cV'=\cD\backslash\cB_{\zeta}^1} H(W_{\cB_{\zeta}^1}| W_{\cV'}M_{1:\zeta})   \nonumber \\
	&=\frac{L}{\zeta-N_\zeta}\frac{1}{{L \choose N_\zeta}{L-N_\zeta \choose \zeta-N_\zeta}} \frac{{\varphi-(\zeta-N_\zeta) \choose N_\zeta}} {{L-\zeta \choose \varphi-\zeta}} {\varphi \choose \zeta-N_\zeta}(\zeta-N_\zeta) \cdot  \nonumber \\
	&\sum_{\substack{\cD\subseteq\cL:\\|\cD|=\varphi}} \sum_{\substack{\cB_\zeta^1\subseteq\cD:\\|\cB_{\zeta}^1|=\zeta-N_\zeta}}\!\! \sum_{\cB_{\varphi}^2=\cD\backslash\cB_{\zeta}^1} \frac{1}{{\varphi \choose \zeta-N_\zeta}(\zeta-N_\zeta)} H(W_{\cB_{\zeta}^1}| W_{\cV'}M_{1:\zeta})   \nonumber  \\
	&\geq \frac{L}{\zeta-N_\zeta}\frac{1}{{L \choose N_\zeta}{L-N_\zeta \choose \zeta-N_\zeta}} \frac{{\varphi-(\zeta-N_\zeta) \choose N_\zeta}} {{L-\zeta \choose \varphi-\zeta}} {\varphi \choose \zeta-N_\zeta}(\zeta-N_\zeta) \cdot  \nonumber \\
	&\sum_{\substack{\cD\subseteq\cL:\\|\cD|=\varphi}} \sum_{\substack{\cB_\varphi^1\subseteq\cD:\\|\cB_{\varphi}^1|=\varphi-N_\varphi}}\!\!\!\! \sum_{\cV''=\cD\backslash\cB_{\varphi}^1} \frac{1/(\varphi-N_\varphi)}{{\varphi\choose \varphi-N_\varphi} } H(W_{\cB_{\varphi}^1}| W_{\cV''}M_{1:\zeta})   \label{ir-alpha-<varphi-6} \\
	&=\frac{L}{\varphi-N_\varphi}\frac{1}{{L \choose N_\varphi} {L-N_\varphi \choose \varphi-N_\varphi}} \cdot \nonumber \\
	&\quad \sum_{\substack{\cD\subseteq\cL:\\|\cD|=\varphi}} \sum_{\substack{\cB_\varphi^1\subseteq\cD:\\|\cB_{\varphi}^1|=\varphi-N_\varphi}} \sum_{\cB_{\varphi}^2=\cD\backslash\cB_{\varphi}^1} H(W_{\cB_{\varphi}^1}| W_{\cB_{\varphi}^2}M_{1:\zeta})    \nonumber \\
	&=\frac{L}{\varphi-N_\varphi}\frac{1}{{L \choose N_\varphi} {L-N_\varphi \choose \varphi-N_\varphi}} \cdot \sum_{(\cB_\varphi^1, \cB_\varphi^2)\in \bB_\varphi} \!\!\!\! H(W_{\cB_{\varphi}^1}| W_{\cB_{\varphi}^2}M_{1:\zeta})    \nonumber \\
	&=\frac{L}{\varphi-N_\varphi}\frac{1}{{L \choose N_\varphi} {L-N_\varphi \choose \varphi-N_\varphi}} \cdot \nonumber \\
	&\quad \sum_{(\cB_\varphi^1, \cB_\varphi^2)\in \bB_\varphi} \Big[H(M_\varphi)+H(W_{\cB_{\varphi}^1}| W_{\cB_{\varphi}^2}M_{1:\varphi})\Big]    \label{ir-alpha-<varphi-9}\\
	&=\frac{L}{\varphi-N_\varphi}H(M_\varphi)   \nonumber \\
	&\quad +\frac{L}{\varphi-N_\varphi}\frac{1}{{L \choose N_\varphi} {L-N_\varphi \choose \varphi-N_\varphi}} \sum_{(\cB_\varphi^1, \cB_\varphi^2)\in \bB_\varphi} \!\!\!\! H(W_{\cB_{\varphi}^1}| W_{\cB_{\varphi}^2}M_{1:\varphi})    \nonumber \\
	&=\frac{L}{\varphi-N_\varphi}H(M_\varphi)+\mu_\varphi,
	\end{align}
	where \eqref{ir-alpha-<varphi-6} follows from Han's inequality (complementary conditioning version),
	and \eqref{ir-alpha-<varphi-9} follows from \eqref{ir-entropy-varphi-4}. 
	This proves \Cref{lemma-converse_condition-iteration}. 
	
	\section{Proof of \Cref{thm-group-pairwise-region-alternative}}\label{proof_thm-group-pairwise-region-alternative}
	Similar to Lemma 11 in \cite{yeung99}, the theorem can be obtained by proving i) $\cR_{\text{gp}}^{L,r}\subseteq\cR_{L,r}^*$; ii) for any $\bm{\lambda}\in\mathbb{R}_+^L$, 
	there exists $\bm{\sR}\in\cR_{\text{gp}}^{L,r}$ such that $\bm{\lambda}\cdot\bm{\sR}=g_{\eta^*}(\bm{\lambda})$.
	\begin{enumerate}[i)]
		\item We first show that $\cR_{\text{gp}}^{L,r}\subseteq\cR_{L,r}^*$. 
		For any $1\leq \alpha\leq r$, let $\bm{r}^{\alpha}=(r^{\alpha}_1,r^{\alpha}_2,\cdots,r^{\alpha}_L)$. 
		For any $\bm{\lambda}\in \mathbb{R}_{+}^L$ and $\bm{\sR}\in\cR_{\text{gp}}^{L,r}$, we have from \eqref{gpRegion-r-1} that
		\begin{equation}
		\bm{\lambda}\cdot\bm{r}^{\alpha}\geq f_1(\bm{\lambda}) \sm_\alpha.  \label{pf-group-region-subrate1}
		\end{equation}
		For $r+1\leq \alpha\leq L$, let $\{c_{\alpha}(\bm{v})\}$ be an optimal $\alpha$-resolution for $\bm{\lambda}$, which implies that 
		\begin{equation}
		\bm{\lambda}\geq \sum_{\bm{v}\in\Omega_L^{\alpha}}c_{\alpha}(\bm{v})\bm{v}.  \label{pf-group-region-resolution-def}
		\end{equation}
		Then we have 
		\begin{align}
		\bm{\lambda}\cdot\bm{r}^{\alpha}&\geq \left(\sum_{\bm{v}\in\Omega_L^{\alpha}}c_{\alpha}(\bm{v})\bm{v}\right)\cdot\bm{r}^{\alpha} \label{pf-group-region-subrate2-1} \\
		&=\sum_{\bm{v}\in\Omega_L^{\alpha}}\big(c_{\alpha}(\bm{v})(\bm{v}\cdot\bm{r}^{\alpha})\big)   \\
		&\geq \sum_{\bm{v}\in\Omega_L^{\alpha}}\left(c_{\alpha}(\bm{v}) \sm_\alpha^*\right)  \label{pf-group-region-subrate2-2} \\
		&=\left(\sum_{\bm{v}\in\Omega_L^{\alpha}}c_{\alpha}(\bm{v})\right)\sm_\alpha^*   \\
		&=f_{\alpha}(\bm{\lambda})\sm_\alpha^*   \label{pf-group-region-subrate2-3}
		\end{align}
		where \eqref{pf-group-region-subrate2-1} follows from \eqref{pf-group-region-resolution-def}, 
		\eqref{pf-group-region-subrate2-2} follows from \eqref{gpRegion-r-2}, 
		and \eqref{pf-group-region-subrate2-3} follows from the optimality of $\{c_{\alpha}(\bm{v})\}$. 
		Summing up \eqref{pf-group-region-subrate1} and \eqref{pf-group-region-subrate2-3} over $\alpha$, we have 
		\begin{align}
		\bm{\lambda}\cdot\bm{\sR}&\geq \sum_{\alpha=1}^{r}f_1(\bm{\lambda}) \sm_\alpha+\sum_{\alpha=r+1}^{L}f_{\alpha}(\bm{\lambda}) \sm_\alpha^*   \\
		&=g_{\eta^*}(\bm{\lambda}).
		\end{align}
		This implies $\bm{\sR}\in\cR_{L,r}^*$ and thus $\cR_{\text{gp}}^{L,r}\subseteq\cR_{L,r}^*$.
		
		\item We now construct a rate tuple $\bm{\sR}$ for each $\bm{\lambda}\in\mathbb{R}_{+}^L$ such that 
		$\bm{\sR}\in\cR_{\text{gp}}^{L,r}$ and $\bm{\lambda}\cdot\bm{\sR}=g_{\eta^*}(\bm{\lambda})$. 
		For $r+1\leq \alpha\leq L$, let $\{c_{\alpha}(\bm{v})\}$ be an optimal $\alpha$-resolution for $\bm{\lambda}$ and let 
		\begin{equation}
		\tilde{\bm{\lambda}}=\sum_{\bm{v}\in\Omega_{L}^{\alpha}}c_{\alpha}(\bm{v})\cdot\bm{v}.
		\end{equation}
		By Lemma 2 in \cite{yeung99}, there exists $1\leq l_{\alpha}\leq \alpha-1$ such that $\lambda_i>\tilde{\lambda}_i$ if and only if $1\leq i\leq l_{\alpha}$. 
		Let $\sR_l=\sum_{\alpha=1}^{L}r_l^{\alpha}$ for $l\in\cL$. 
		We construct $\bm{\sR}$ by designing the sub-rates $r_l^{\alpha}$ as follows.
		\begin{enumerate}
			\item For $1\leq \alpha\leq r$, let 
			\begin{equation}
			r_l^{\alpha}=\sm_\alpha, \text{ for all }1\leq l\leq L.
			\end{equation}
			
			\item For $r+1\leq \alpha\leq \eta^*-1$, let 
			\begin{equation}
			r_l^{\alpha}=0, \text{ for all }1\leq l\leq L.
			\end{equation}
			
			\item For $\eta^*\leq \alpha\leq L$, let 
			\begin{equation}
			r_l^{\alpha}=
			\begin{cases}
			0,&\text{ for }1\leq l\leq l_{\alpha} \\
			\frac{\sm_\alpha^*}{\alpha-l_{\alpha}},&\text{ for }l_{\alpha}+1\leq l\leq L.
			\end{cases}
			\end{equation}			
		\end{enumerate}
		We first verify that such a construction implies $\bm{\sR}\in\cR_{\text{gp}}^{L,r}$.
		\begin{enumerate}
			\item For $1\leq \alpha\leq r$, it is obvious that \eqref{gpRegion-r-1} is satisfied.
			
			\item For $r+1\leq \alpha\leq \eta^*-1$, since $\sm_\alpha^*=0$, \eqref{gpRegion-r-2} is satisfied.
			
			\item For $\eta^*\leq \alpha\leq L$, consider any $\cB\subseteq\cL$ such that $|\cB|=\alpha$. 
			Let $\bm{e}_{\alpha}$ be an $L$-vector with the first $l_{\alpha}$ components being 0 and the last $L-l_{\alpha}$ components being 1. 
			Let $\bm{v}_{\cB}=(v_1,v_2,\cdots,v_L)$ be such that $v_i=1$ if and only if $i\in\cB$. 
			Since $\sum_{i=1}^{l_{\alpha}}v_i\leq l_{\alpha}$, we have $\bm{e}_{\alpha}\cdot\bm{v}_{\cB}\geq \alpha-l_{\alpha}$. 
			Thus,
			\begin{eqnarray}
			\sum_{l\in\cB}r_l^{\alpha}&=&\left(\frac{\sm_\alpha^*}{\alpha-l_{\alpha}}\bm{e}_{\alpha}\right)\cdot\bm{v}_{\cB}  \\
			&=&\frac{\sm_\alpha^*}{\alpha-l_{\alpha}}\left(\bm{e}_{\alpha}\cdot\bm{v}_{\cB}\right)  \\
			&\geq &\frac{\sm_\alpha^*}{\alpha-l_{\alpha}}\left(\alpha-l_{\alpha}\right)  \\
			&=&\sm_\alpha^*.  \label{pf-group-region-subrate-check-1}
			\end{eqnarray}
			
		\end{enumerate}
		Thus, $\bm{\sR}\in\cR_{\text{gp}}^{L,r}$. 
		Now it remains to show that $\bm{\lambda}\cdot\bm{\sR}=g_{\eta^*}(\bm{\lambda})$.  
		We consider the following cases.
		\begin{enumerate}
			\item For $1\leq \alpha\leq r$, it is easy to check that 
			\begin{equation}
			\bm{\lambda}\cdot\bm{r}^{\alpha}=f_1(\bm{\lambda}) \sm_\alpha.  \label{subrate-sum-1}
			\end{equation}
			
			\item For $r+1\leq \alpha\leq \eta^*-1$, it is obvious that
			\begin{equation}
			\bm{\lambda}\cdot\bm{r}^{\alpha}=0.   \label{subrate-sum-2}
			\end{equation}
			
			\item For $\eta^*\leq \alpha\leq L$, only the first $l_{\alpha}$ components of $\bm{\lambda}-\tilde{\bm{\lambda}}$ are nonzero. Thus, we have
			\begin{equation}
			\left(\bm{\lambda}-\sum_{\bm{v}\in\Omega_L^{\alpha}}c_{\alpha}(\bm{v}) \bm{v}\right)\cdot\bm{r}^{\alpha}=0,
			\end{equation}
			which implies that
			\begin{equation}
			\bm{\lambda}\cdot\bm{r}^{\alpha}=\left(\sum_{\bm{v}\in\Omega_L^{\alpha}}c_{\alpha}(\bm{v}) \bm{v}\right)\cdot\bm{r}^{\alpha}=\sum_{\bm{v}\in\Omega_L^{\alpha}}\big(c_{\alpha}(\bm{v}) (\bm{v}\cdot\bm{r}^{\alpha})\big).
			\end{equation}
			By Lemma 2 in \cite{yeung99}, for any $\bm{v}\in\Omega_L^{\alpha}$ such that $c_{\alpha}(\bm{v})>0$, 
			the first $l_{\alpha}$ components are equal to 1, $(\alpha-l_{\alpha})$ of the other $L-l_{\alpha}$ components are equal to 1, 
			and the rest are equal to 0. On the other hand, the first $l_{\alpha}$ components of $\bm{r}^{\alpha}$ are equal to zero. 
			Thus, for any $\bm{v}\in\Omega_L^{\alpha}$ such that $c_{\alpha}(\bm{v})>0$, we have 
			\begin{equation}
			\bm{v}\cdot\bm{r}^{\alpha}=(\alpha-l_{\alpha}) \frac{\sm_\alpha^*}{\alpha-l_{\alpha}}=\sm_\alpha^*.
			\end{equation}
			Then
			\begin{eqnarray}
			\bm{\lambda}\cdot\bm{r}^{\alpha}&=&\sum_{\bm{v}\in\Omega_L^{\alpha}}\big(c_{\alpha}(\bm{v}) (\bm{v}\cdot\bm{r}^{\alpha})\big)  \\
			&=&\sum_{\bm{v}\in\Omega_L^{\alpha}}c_{\alpha}(\bm{v}) \sm_\alpha^*  \\
			&=&\left(\sum_{\bm{v}\in\Omega_L^{\alpha}}c_{\alpha}(\bm{v})\right)\sm_\alpha^*  \\
			&=&f_{\alpha}(\bm{\lambda})\sm_\alpha^*. \label{subrate-sum-4}
			\end{eqnarray}
		\end{enumerate}
		Summing up \eqref{subrate-sum-1}, \eqref{subrate-sum-2}, and \eqref{subrate-sum-4} over all $1\leq \alpha\leq L$, 
		we obtain $\bm{\lambda}\cdot\bm{\sR}=g_{\eta^*}(\bm{\lambda})$. 
		Therefore, \Cref{thm-group-pairwise-region-alternative} is proved.
	\end{enumerate}
	
	\section{Proof of \Cref{lemma-eta*}}\label{proof_lemma-eta*}
	We prove the lemma by proving (i) for $r+1\leq \eta^*\leq L$, $\sum_{\alpha=1}^{r}(\alpha-1)\sm_\alpha\leq \sum_{\alpha=r+1}^{\eta^*}\sm_\alpha$ is equivalent to $g_{\eta^*}(\bm{\lambda})\geq g_{\eta^*+1}(\bm{\lambda})\geq \cdots\geq g_{L+1}(\bm{\lambda})$; (ii) for $r+2\leq \eta^*\leq L+1$, $\sum_{\alpha=r+1}^{\eta^*-1}\sm_\alpha<\sum_{\alpha=1}^{r}(\alpha-1)\sm_\alpha$ is equivalent to $g_{\eta^*}(\bm{\lambda})> g_{\eta^*-1}(\bm{\lambda})> \cdots>g_{r+1}(\bm{\lambda})$.
	\begin{enumerate}[(i)]
		\item For $\eta^*\leq \eta\leq L$, we have
		\begin{align}
		&g_{\eta}(\bm{\lambda})\geq g_{\eta+1}(\bm{\lambda})   \label{pf_lemma-eta*-right-1} \\
		&\Updownarrow  \nonumber \\
		&\sum_{\alpha=\eta+1}^{L}f_{\alpha}(\bm{\lambda})\sm_\alpha+ f_{\eta}(\bm{\lambda}) \left[\sum_{\alpha=r+1}^{\eta}\sm_\alpha- \sum_{\alpha=1}^{r}(\alpha-1)\sm_\alpha\right]  \nonumber  \\
		&\geq \sum_{\alpha=\eta+2}^{L}f_{\alpha}(\bm{\lambda})\sm_\alpha  \nonumber \\
		&\quad + f_{\eta+1}(\bm{\lambda}) \Bigg[\sum_{\alpha=r+1}^{\eta+1}\sm_\alpha-\sum_{\alpha=1}^{r}(\alpha-1)\sm_\alpha\Bigg] \label{pf_lemma-eta*-right-2} \\
		&\Updownarrow  \nonumber \\
		&f_{\eta+1}(\bm{\lambda})\sm_{\eta+1}+ f_{\eta}(\bm{\lambda}) \left[\sum_{\alpha=r+1}^{\eta}\sm_\alpha- \sum_{\alpha=1}^{r}(\alpha-1)\sm_\alpha\right]  \nonumber  \\
		&\geq  f_{\eta+1}(\bm{\lambda}) \left[\sum_{\alpha=r+1}^{\eta+1}\sm_\alpha- \sum_{\alpha=1}^{r}(\alpha-1)\sm_\alpha\right]   \label{pf_lemma-eta*-right-3} \\
		&\Updownarrow  \nonumber \\
		&\left(f_{\eta}(\bm{\lambda})-f_{\eta+1}(\bm{\lambda})\right) \left[\sum_{\alpha=r+1}^{\eta}\sm_\alpha- \sum_{\alpha=1}^{r}(\alpha-1)\sm_\alpha\right] \geq 0  \label{pf_lemma-eta*-right-4}\\
		&\Updownarrow  \nonumber \\
		&\sum_{\alpha=r+1}^{\eta}\sm_\alpha\geq \sum_{\alpha=1}^{r}(\alpha-1)\sm_\alpha.  \label{pf_lemma-eta*-right-5}
		\end{align}
		Thus, we conclude that 
		\begin{equation}
		g_{\eta^*}(\bm{\lambda})\geq g_{\eta^*+1}(\bm{\lambda})\geq \cdots\geq g_{L+1}(\bm{\lambda})
		\end{equation}
		is equivalent to 
		\begin{equation}
		\sum_{\alpha=1}^{r}(\alpha-1)\sm_\alpha\leq \sum_{\alpha=r+1}^{\eta}\sm_\alpha \text{ for all }\eta^*\leq \eta\leq L+1,
		\end{equation}
		which is also equivalent to
		\begin{equation}
		\sum_{\alpha=1}^{r}(\alpha-1)\sm_\alpha\leq \sum_{\alpha=r+1}^{\eta^*}\sm_\alpha.
		\end{equation}
		
		\item For $r+1\leq \eta\leq\eta^*$, we have
		\begin{align}
		&g_{\eta}(\bm{\lambda})>g_{\eta-1}(\bm{\lambda})   \label{pf_lemma-eta*-left-1} \\
		&\Updownarrow  \nonumber \\
		&\sum_{\alpha=\eta+1}^{L}f_{\alpha}(\bm{\lambda})\sm_\alpha+ f_{\eta}(\bm{\lambda}) \left[\sum_{\alpha=r+1}^{\eta}\sm_\alpha-\sum_{\alpha=1}^{r}(\alpha-1)\sm_\alpha\right]  \nonumber  \\
		&> \sum_{\alpha=\eta}^{L}f_{\alpha}(\bm{\lambda})\sm_\alpha  \nonumber \\
		&\quad + f_{\eta-1}(\bm{\lambda}) \left[\sum_{\alpha=r+1}^{\eta-1}\sm_\alpha- \sum_{\alpha=1}^{r}(\alpha-1)\sm_\alpha\right] \label{pf_lemma-eta*-left-2} \\
		&\Updownarrow  \nonumber \\
		&f_{\eta}(\bm{\lambda}) \left[\sum_{\alpha=r+1}^{\eta}\sm_\alpha- \sum_{\alpha=1}^{r}(\alpha-1)\sm_\alpha\right]  \nonumber  \\
		&> f_{\eta}(\bm{\lambda})\sm_{\eta}+ f_{\eta-1}(\bm{\lambda}) \left[\sum_{\alpha=r+1}^{\eta-1}\sm_\alpha- \sum_{\alpha=1}^{r}(\alpha-1)\sm_\alpha\right]  \label{pf_lemma-eta*-left-3} \\
		&\Updownarrow  \nonumber \\
		&\left(f_{\eta-1}(\bm{\lambda})-f_{\eta}(\bm{\lambda})\right) \left[\sum_{\alpha=r+1}^{\eta-1}\sm_\alpha- \sum_{\alpha=1}^{r}(\alpha-1)\sm_\alpha\right]<0   \label{pf_lemma-eta*-left-4} \\
		&\Updownarrow  \nonumber \\
		&\sum_{\alpha=r+1}^{\eta-1}\sm_\alpha<\sum_{\alpha=1}^{r}(\alpha-1)\sm_\alpha.   \label{pf_lemma-eta*-left-5}
		\end{align}
		Thus, we conclude that 
		\begin{equation}
		g_{\eta^*}(\bm{\lambda})>g_{\eta^*-1}(\bm{\lambda})>\cdots>g_{r+1}(\bm{\lambda})
		\end{equation}
		is equivalent to 
		\begin{equation}
		\sum_{\alpha=r+1}^{\eta-1}\sm_\alpha<\sum_{\alpha=1}^{r}(\alpha-1)\sm_\alpha \text{ for all }r+1\leq \eta\leq \eta^*,
		\end{equation}
		which is also equivalent to
		\begin{equation}
		\sum_{\alpha=r+1}^{\eta^*-1}\sm_\alpha<\sum_{\alpha=1}^{r}(\alpha-1)\sm_\alpha.
		\end{equation}
	\end{enumerate}
	
	\section{Converse Proof of \Cref{thm-group-pairwise-region} (continuing)}\label{proof_DS-converse}
	In order to prove the inequality in \eqref{equal-region}, i.e., $\bm{\lambda}\cdot\bm{\sR}\geq g_{\eta^*}(\bm{\lambda})$, 
	we first introduce some lemmas and important parameters that will be used. 
	The connection between the example at the end of \Cref{section-DS-proof} and the general converse proof here will be provided when the corresponding parameters are defined. 
	
	Similar to Lemma 6 in \cite{guo-yeung-SMDC-IT20}, the following lemma gives a sufficient condition of redundancy in the characterization of the rate region. 
	\begin{lemma}\label{lemma-considerable-lambda}
		For any $\eta=r+1,r+2,\cdots,L+1$, the rate constraint $\bm{\lambda}\cdot\bm{\sR}\geq g_{\eta}(\bm{\lambda})$ is redundant in the characterization of $\cR_{L,r}^*$ if
		\begin{equation}
		\lambda_{1}>\frac{\lambda_{2}+\lambda_{3}+\cdots+\lambda_{L}}{\eta-1}. \label{lambda-general-1}
		\end{equation}
	\end{lemma}
	\begin{proof}
		See Appendix \ref{proof_lemma-considerable-lambda}.
	\end{proof}
	
	For any $\eta\in\{r+1,r+2,\cdots,L+1\}$, $\bm{\lambda}$ is called an \textit{$\eta$-considerable coefficient vector} if $\lambda_1\geq \lambda_2\geq \cdots\geq \lambda_L$ and 
	\begin{equation}
	\lambda_1\leq \frac{\lambda_2+\lambda_3+\cdots+\lambda_L}{\eta-1}.
	\end{equation}
	Denote the set of all $\eta$-considerable coefficient vectors by $\mathbb{R}_{\eta}^L$. Then let 
	\begin{equation}
	\mathbb{R}_{\text{con}}^L=\bigcup_{\eta=r+1}^{L+1}\mathbb{R}_{\eta}^L.
	\end{equation}
	We have the following property on vectors in $\mathbb{R}_{\text{con}}^L$, 
	for which a simple proof is given in Appendix \ref{proof_lemma-f>lambda1}. 
	\begin{lemma}\label{lemma-f>lambda1}
		For $\eta=r+1,r+2,\cdots,L+1$ and $\bm{\lambda}\in\mathbb{R}_{\eta}^L$, we have $f_{\eta}(\bm{\lambda})\geq \lambda_1$.
	\end{lemma}
	
	By \Cref{lemma-considerable-lambda}, we only need to prove $\bm{\lambda}\cdot\bm{\sR}\geq g_{\eta^*}(\bm{\lambda})$ for $\bm{\lambda}\in\mathbb{R}_{\text{con}}^L$.
	Thus, we assume $\bm{\lambda}\in\mathbb{R}_{\text{con}}^L$ in the sequel. 
	From Theorem~1 in~\cite{guo-yeung-SMDC-IT20}, we can verify that 
	\begin{equation}
	\sum_{i=1}^{L}\lambda_i=f_1(\bm{\lambda})\geq \eta f_{\eta}(\bm{\lambda}),  \label{f(1)-vs-f(eta)}
	\end{equation}
	which implies that 
	\begin{equation}
	\sum_{i=1}^{L}\lambda_i-(r-1) f_{\eta}(\bm{\lambda})\geq \left[\eta-(r-1)\right] f_{\eta}(\bm{\lambda})>0.   \label{subtraction-ability}
	\end{equation}
	Let $\xi_{\alpha}\in\cL$ be the index of $\bm{\lambda}$ such that 
	\begin{equation}
	\sum_{i=1}^{\xi_{\alpha}-1}\lambda_i<\alpha f_{\eta}(\bm{\lambda})\leq \sum_{i=1}^{\xi_{\alpha}}\lambda_i.
	\end{equation}
	For simplicity, let $\xi_0=1$. From \Cref{lemma-f>lambda1}, we can see that 
	\begin{equation}
	f_{\eta}(\bm{\lambda})\geq \lambda_1\geq \lambda_2\geq \cdots\geq \lambda_L,  \label{f(eta)-vs-lambda}
	\end{equation}
	which implies
	\begin{equation}
	\xi_0\leq \xi_1<\xi_2<\cdots<\xi_{r}.
	\end{equation}
	and 
	\begin{equation}
	\xi_i\geq i.
	\end{equation}
	
	\begin{figure}[!t]
		\centering
		\begin{tikzpicture}[font=\scriptsize, scale=1.0]
		\draw[thick] (0,0) rectangle (6.0,-0.7);
		\draw[thick] (0,-0.8) rectangle (6.0,-1.5);
		\draw[thick] (0,-1.6) rectangle (6.0,-2.3);
		\draw[thick] (0,-2.4) rectangle (6.0,-3.1);
		\draw[thick] (0,-3.2) rectangle (6.0,-3.9);
		\draw[thick] (0,-4.0) rectangle (6.0,-4.5);
		\draw[thick] (0,-4.6) rectangle (6.0,-5.3);
		
		\draw (4.0,0)--(4.0,-0.7);
		\draw (1.8,-0.8)--(1.8,-1.5) (5.0,-0.8)--(5.0,-1.5);
		\draw (2,-1.6)--(2,-2.3) (4.7,-1.6)--(4.7,-2.3);
		\draw (1.2,-2.4)--(1.2,-3.1) (3.6,-2.4)--(3.6,-3.1) (5.5,-2.4)--(5.5,-3.1);
		\draw (1.5,-3.2)--(1.5,-3.9) (3.2,-3.2)--(3.2,-3.9) (4.9,-3.2)--(4.9,-3.9);
		\draw (1.5,-4.6)--(1.5,-5.3) (3.0,-4.6)--(3.0,-5.3) (4.0,-4.6)--(4.0,-5.3) (5.2,-4.6)--(5.2,-5.3);
		
		\node (g1,1) at (2.0,-0.35) {$\gamma_1^{(1)}$}; \node (g2,1) at (5.0,-0.35) {$\gamma_2^{(1)}$};
		\node (g2,2) at (0.9,-1.15) {$\gamma_2^{(2)}$}; \node (g3,2) at (3.4,-1.15) {$\gamma_3^{(2)}$}; \node (g4,2) at (5.5,-1.15) {$\gamma_4^{(2)}$};
		\node (g4,3) at (0.9,-1.95) {$\gamma_4^{(3)}$}; \node (g5,3) at (3.4,-1.95) {$\gamma_5^{(3)}$}; \node (g6,3) at (5.4,-1.95) {$\gamma_6^{(3)}$};
		\node (g6,4) at (0.6,-2.75) {$\gamma_6^{(4)}$}; \node (g7,4) at (2.3,-2.75) {$\gamma_7^{(4)}$}; \node (g8,4) at (4.5,-2.75) {$\gamma_8^{(4)}$}; \node (g9,4) at (5.75,-2.75) {$\gamma_9^{(4)}$};
		\node (g9,5) at (0.7,-3.55) {$\gamma_9^{(5)}$}; \node (g10,5) at (2.4,-3.55) {$\gamma_{10}^{(5)}$}; \node (g11,5) at (4.2,-3.55) {$\gamma_{11}^{(5)}$}; \node (g12,5) at (5.5,-3.55) {$\gamma_{12}^{(5)}$};
		\node (vdots) at (3,-4.15) {$\bm{\vdots}$};
		\node (g9,5) at (0.7,-4.95) {$\gamma_{\xi_{r-2}}^{(r-1)}$}; \node (g10,5) at (2.2,-4.95) {$\gamma_{\xi_{r-2}+1}^{(r-1)}$};  \node (g11,5) at (3.5,-4.95) {$\bm{\cdots}$}; \node (g12,5) at (4.6,-4.95) {$\gamma_{\xi_{r-1}-1}^{(r-1)}$}; \node (g12,5) at (5.6,-4.95) {$\gamma_{\xi_{{r-1}}}^{(r-1)}$};
		\end{tikzpicture}
		\caption{Illustration of $\gamma^{(\alpha)}_{i}$} \label{fig-define-gamma-partition}
	\end{figure}
	Due to \eqref{subtraction-ability}, we can subtract $r-1$ of $f_{\eta}(\bm{\lambda})$ one by one from the sequence $\lambda_1,\lambda_2,\cdots,\lambda_L$. 
	The subtraction process is illustrated in Fig.~\ref{fig-define-gamma-partition}.
	For $\alpha=1,2,\cdots,r-1$, let $\bm{\gamma}^{(\alpha)}=\left(\gamma^{(\alpha)}_1, \gamma^{(\alpha)}_2, \cdots, \gamma^{(\alpha)}_L\right)$ be the $\alpha$-th \textit{subtraction} 
	and $\bm{\lambda}^{(\alpha)}=\left(\lambda^{(\alpha)}_1, \lambda^{(\alpha)}_2,\cdots, \lambda^{(\alpha)}_L\right)$ be the $\alpha$-th \textit{residue} after the first $\alpha$ subtractions such that
	\begin{equation}
	\gamma_i^{(\alpha)}=
	\begin{cases}
	\sum\limits_{i=1}^{\xi_{\alpha-1}}\lambda_i-(\alpha-1) f_{\eta}(\bm{\lambda}),&\text{ if }i=\xi_{\alpha-1}\\
	\alpha f_{\eta}(\bm{\lambda})-\sum\limits_{i=1}^{\xi_{\alpha}-1}\lambda_i,&\text{ if }i=\xi_{\alpha}\\
	\lambda_i,&\text{ if }\xi_{\alpha-1}<i<\xi_{\alpha}\\
	0,&\text{ if }i<\xi_{\alpha-1} \text{ or }i>\xi_{\alpha}.
	\end{cases}   \label{def-gamma(i,alpha)}
	\end{equation}
	and $\lambda_i^{(\alpha)}=\lambda_i-\sum_{j=1}^{\alpha}\gamma_i^{(j)}$. 
	Thus,
	\begin{equation}
	\lambda_i^{(\alpha)}=
	\begin{cases}
	0,& \text{ if }i<\xi_{\alpha}\\
	\sum_{i=1}^{\xi_{\alpha}}\lambda_i-\alpha f_{\eta}(\bm{\lambda}),&\text{ if }i=\xi_{\alpha}\\
	\lambda_i,&\text{ if }i>\xi_{\alpha}
	\end{cases}  \label{def-lambda(i,alpha)}
	\end{equation}
	It is easy to check that
	\begin{equation}
	\sum_{i=\xi_{\alpha-1}}^{\xi_{\alpha}}\gamma_i^{(\alpha)}=f_{\eta}(\bm{\lambda})   \label{property-gamma-1}
	\end{equation}
	and
	\begin{equation}
	\gamma_{\xi_{\alpha}}^{(\alpha)}+\gamma_{\xi_{\alpha}}^{(\alpha+1)}=\lambda_{\xi_{\alpha}}.  \label{property-gamma-2}
	\end{equation}
	
	\begin{remark}\label{remark-cof-gamma-lambda}
		In the example at the end of \Cref{section-DS-proof}, the subtraction and residue parameters are the coefficients in \eqref{converse-example-2}, which is $\bm{\lambda}^{(1)}=(0,\frac{2}{3},1,1)$ and $\gamma^{(1)}=(1,\frac{1}{3},0,0)$.
	\end{remark}
	
	Let $\bm{\lambda}^{(r-1)}=\left(\lambda_{\xi_{r-1}}^{(r-1)}, \lambda_{\xi_{r-1}+1}^{(r-1)}, \cdots, \lambda_{L}^{(r-1)}\right)$. 
	The following lemma will be used in the converse. 
	The detailed proof of the lemma is given in Appendix \ref{proof_lemma-last-choice}.
	\begin{lemma}\label{lemma-last-choice}
		$f_{\eta-(r-1)}\left(\bm{\lambda}^{(r-1)}\right)\geq f_{\eta}(\bm{\lambda})$.
	\end{lemma}
	By the definition of $f_{\eta-(r-1)}\left(\bm{\lambda}^{(r-1)}\right)$ in \eqref{optimization-1}, the value of the objective function $\sum_{\bm{v}\in\Omega_{L-\xi_{r-1}+1}^{\eta-(r-1)}}c_{\eta-(r-1)}(\bm{v})$ lies in the range $\Big[0,f_{\eta-(r-1)}\left(\bm{\lambda}^{(r-1)}\right)\Big]$. The inequality in \Cref{lemma-last-choice} implies $f_{\eta}(\bm{\lambda})\in\Big[0,f_{\eta-(r-1)}\left(\bm{\lambda}^{(r-1)}\right)\Big]$. Thus, there exists an $[\eta-(r-1)]$-resolution $\left\{c_{\eta-(r-1)}(\bm{v}):\bm{v}\in\Omega_{L-\xi_{r-1}+1}^{\eta-(r-1)}\right\}$ for $\bm{\lambda}^{(r-1)}$ such that
	\begin{equation}
	\sum_{\bm{v}\in\Omega_{L-\xi_{r-1}+1}^{\eta-(r-1)}}c_{\eta-(r-1)}(\bm{v})=f_{\eta}(\bm{\lambda}).
	\end{equation}
	For $\bm{v}\in\Omega_{L-\xi_{r-1}+1}^{\eta-(r-1)}$ such that $c_{\eta-(r-1)}(\bm{v})>0$, let $\bm{v}=(v_1, v_2, \cdots, v_{L-\xi_{r-1}+1})$ and
	\begin{equation}
	D_{\bm{v}}=\big\{i\in\{\xi_{r-1}, \xi_{r-1}+1,\cdots, L\}:v_{i-\xi_{r-1}+1}=1\big\}.
	\end{equation}
	Let $\cD=\left\{D_{\bm{v}}: \bm{v}\in\Omega_{L-\xi_{r-1}+1}^{\eta-(r-1)}, ~c_{\eta-(r-1)}(\bm{v})>0\right\}$ and $|\cD|=b_1$. 
	For simplicity, let $\cD=\{D_1,D_2,\cdots,D_{b_1}\}$. 
	For $k=\{1,2,\cdots,b_1\}$, if $D_k=D_{\bm{v}}$ for some $\bm{v}\in\Omega_{L-\xi_{r-1}+1}^{\eta-(r-1)}$, let $c(D_k)=c_{\eta-(r-1)}(\bm{v})$. 
	Then
	\begin{equation}
	\sum_{k=1}^{b_1}c(D_k)=f_{\eta}(\bm{\lambda}).  \label{c(D_k)=f(eta)}
	\end{equation}
	
	For $\alpha\in\{1,2,\cdots,r-1\}$, let $A_{\alpha}=\{i_1,i_2,\cdots,i_{\alpha-1}\}$, where $i_j\in\{\xi_{j-1},\xi_{j-1}+1,\cdots,\xi_j\}$ for $j\in\{1,2,\cdots,\alpha-1\}$. 
	Let $\cA^{(\alpha)}$ be the collection of all $A_{\alpha}$. 
	For $A_{\alpha}\in\cA^{(\alpha)}$, for notational simplicity, let
	\begin{equation}
	H_{A_{\alpha}}=\min_{j=1,2,\cdots,\alpha-1}\left\{\sum_{k=\xi_{j-1}}^{i_j}\gamma^{(j)}_k\right\}
	\end{equation}
	and 
	\begin{equation}
	Q_{A_{\alpha}}=\max_{j=1,2,\cdots,\alpha-1}\left\{\sum_{k=\xi_{j-1}}^{i_j-1}\gamma^{(j)}_k\right\}.
	\end{equation}
	For each $i_{\alpha}\in\{\xi_{\alpha-1},\xi_{\alpha-1}+1,\cdots,\xi_{\alpha}\}$, let 
	\begin{equation}
	h_{\alpha}=\sum_{k=\xi_{j-1}}^{i_{\alpha}}\gamma^{(\alpha)}_k
	\end{equation}
	and
	\begin{equation}
	q_{\alpha}=\sum_{k=\xi_{j-1}}^{i_{\alpha}-1}\gamma^{(\alpha)}_k.
	\end{equation}
	Then for $\alpha\in\{1,2,\cdots,r-1\}$, $i_{\alpha}\in\{\xi_{\alpha-1},\xi_{\alpha-1}+1,\cdots,\xi_{\alpha}\}$, and $A_{\alpha}\in\cA^{(\alpha)}$, define $\gamma^{A_{\alpha}}_{i_\alpha}$ by 
	\begin{align}
	\gamma^{A_{\alpha}}_{i_\alpha}\triangleq \left[\min\{h_{\alpha},H_{A_{\alpha}}\}-\max\{q_{\alpha},Q_{A_{\alpha}}\}\right]^+,
	\end{align}
	where for any $x\in\mathbb{R}$, $[x]^+\triangleq \max\{0,x\}$ as defined after \eqref{def-nonnegative-function}. 
	For notational simplicity, we denote $\gamma^{(\alpha)}_{i_\alpha}$ and $\gamma^{A_{\alpha}}_{i_\alpha}$ by $\gamma^{(\alpha)}_{i}$ and $\gamma^{A_{\alpha}}_{i}$ respectively, where $i\in\{\xi_{\alpha-1},\xi_{\alpha-1}+1,\cdots,\xi_{\alpha}\}$. 
	
	Let $\cA^{(\alpha)}_0$ be the collection of $A_{\alpha}$ such that $\gamma^{A_{\alpha}}_i> 0$. 
	We can verify that for any $i\in\{\xi_{\alpha-1},\xi_{\alpha-1}+1,\cdots, \xi_{\alpha}\}$,
	\begin{equation}
	\sum_{A_{\alpha}\in\cA^{(\alpha)}_0}\gamma_i^{A_{\alpha}}=\gamma_i^{(\alpha)}.   \label{gamma-subset-partition-property}
	\end{equation}
	This means that $\gamma^{A_{\alpha}}_{i},~A_{\alpha}\in\cA^{(\alpha)}$ is a partition of $\gamma^{(\alpha)}_{i}$. 
	This partition is the key idea of the converse proof in \eqref{converse-proof-1}-\eqref{converse-proof-5} that we recursively partition the coefficient of an entropy term into coefficients of entropies in a lower layer. For example, the coefficient of $H(W_1,W_2,W_3|W_4)$ is partitioned into coefficients of $H(W_1,W_2|W_3,W_4)$, $H(W_1,W_3|W_2,W_4)$, and $H(W_2,W_3|W_1,W_4)$. 
	
	For $\alpha=1$, we can see that $\cA^{(1)}_0=\{\emptyset\}$ and for $i\in\{1,2,\cdots,\xi_1\}$,
	\begin{equation}
	\gamma^{A_1}_i=\gamma^{(1)}_i.
	\end{equation}
	If there is an $A_\alpha$ such that $i\in A_\alpha$, then $i=\xi_{\alpha-1}$.
	In particular, for all $A_{\alpha}$ such that $\xi_{\alpha-1}\in A_{\alpha}$, we have $\gamma^{A_{\alpha}}_{\xi_{\alpha-1}}=0$ since
	\begin{eqnarray}
	\gamma_{\xi_{\alpha}}^{(\alpha+1)}&=&\lambda_{\xi_{\alpha}}-\gamma_{\xi_{\alpha}}^{(\alpha)}  \nonumber \\
	&\leq &f_{\eta}(\bm{\lambda})-\gamma_{\xi_{\alpha}}^{(\alpha)}  \nonumber \\
	&=&\sum_{i=\xi_{\alpha-1}}^{\xi_{\alpha}-1}\gamma_i^{(\alpha)},  \label{partition-condiiton}
	\end{eqnarray}
	where the inequality follows from \Cref{lemma-f>lambda1}. 
	It is easy to check that for $i\in\{1,2,\cdots, \xi_{\alpha-1}-1\}$,
	\begin{equation}
	\sum_{k=\xi_{\alpha-1}}^{\xi_{\alpha}}\sum_{A_{\alpha}\in\cA^{(\alpha)}_0:~i\in A_{\alpha}}\gamma_k^{A_{\alpha}}=\lambda_i  \label{gamma-subset-property-2}
	\end{equation}
	and for $i=\xi_{\alpha-1}$,
	\begin{equation}
	\sum_{k=\xi_{\alpha-1}}^{\xi_{\alpha}}\sum_{A_{\alpha}\in\cA^{(\alpha)}_0:~\xi_{\alpha-1}\in A_{\alpha}}\gamma_k^{A_{\alpha}}=\gamma_{\xi_{\alpha-1}}^{(\alpha-1)}.   \label{gamma-subset-property-3}
	\end{equation}
	Thus,
	\begin{align}
	&\sum_{k=\xi_{\alpha-1}}^{\xi_{\alpha}}\sum_{A_{\alpha}\in\cA^{(\alpha)}_0:~\xi_{\alpha-1}\in A_{\alpha}}\gamma_k^{A_{\alpha}}+\sum_{A_{\alpha}\in\cA^{(\alpha)}_0}\gamma_{\xi_{\alpha-1}}^{A_{\alpha}}  \nonumber \\
	&=\gamma_{\xi_{\alpha-1}}^{(\alpha-1)}+\gamma_{\xi_{\alpha-1}}^{(\alpha)}  \nonumber \\
	&=\lambda_{\xi_{\alpha-1}},   \label{gamma-subset-property-5}
	\end{align}
	where the first equality follows from \eqref{gamma-subset-property-3} and \eqref{gamma-subset-partition-property}, and the second equality follows from \eqref{property-gamma-2}.
	
	For any $k\in\{1,2,\cdots,\alpha-1\}$ and $A_{\alpha}\in\cA^{(\alpha)}_0$, 
	let $A_{\alpha}^k=\{i_1,i_2,\cdots,i_k\}$ be the set of the first $k$ smallest elements in $A_{\alpha}$. 
	In particular, $A_{\alpha}^{\alpha-1}=A_{\alpha}$. Then the condition $\gamma_i^{A_{\alpha}}=\sum_{k=\xi_{\alpha}}^{\xi_{\alpha+1}}\gamma_k^{\{i\}\cup A_{\alpha}}$ implies that
	\begin{equation}
	\gamma_i^{A_{\alpha}}=\sum_{A_{\alpha+1}^{\alpha-1}\in\cA^{(\alpha+1)}_0: ~A_{\alpha+1}^{\alpha-1}=A_{\alpha}}\gamma_j^{A_{\alpha+1}}.
	\end{equation}
	For $i\in\cL$ and $\alpha\in\{1,2,\cdots,r-1\}$, we have
	\begin{equation}
	\lambda_i=\lambda_i^{(\alpha)}+\sum_{k=1}^{\alpha}\gamma_i^{(k)}=\lambda_i^{(\alpha)}+\sum_{k=1}^{\alpha}\sum_{A_k\in\cA^{(k)}_0}\gamma_i^{A_k}.
	\end{equation}
	In particular, for $\alpha=r-1$,
	\begin{equation}
	\lambda_i=\lambda_i^{(r-1)}+\sum_{k=1}^{r-1}\sum_{A_k\in\cA^{(k)}_0}\gamma_i^{A_k}.   \label{lambda(i,r-1)}
	\end{equation}
	
	Let $\cA^{(r)}=\Big\{\{i\}\cup A_{r-1}: \gamma_i^{A_{r-1}}> 0 \text{ for }i\in\{\xi_{r-2},\xi_{r-2}+1,\cdots,\xi_{r-1}\}\text{ and }A_{r-1}\in\cA^{(r-1)}_0\Big\}$. 
	Denote the cardinality of $\cA^{(r)}$ by $b_2$. 
	For simplicity, let $\cA^{(r)}=\{B_1,B_2,\cdots,B_{b_2}\}$. 
	For $j\in\{1,2,\cdots,b_2\}$, \eqref{partition-condiiton} implies that
	\begin{equation}
	|B_j|=r-1.
	\end{equation}
	Without loss of generality, let $B_j=\{i\}\cup A_{r-1}$ for some $i\in\{\xi_{r-2},\xi_{r-2}+1,\cdots,\xi_{r-1}\}$ and $A_{r-1}\in\cA^{(r-1)}_0$. 
	For $k\in\{1,2,\cdots,r-1\}$, let 
	\begin{equation}
	B_j^k=
	\begin{cases}
	A_{r-1}^{k},&\text{if }1\leq k\leq r-2  \\
	B_j,&\text{if }k=r-1.
	\end{cases}
	\end{equation}
	Note that $B_j^k$ is the set of the first $k$ smallest elements in $B_j$. Let $\gamma(B_j)=\gamma_i^{A_{r-1}}$ which is the number of $B_j$. 
	Then we have 
	\begin{eqnarray}
	\sum_{j=1}^{b_2} \gamma(B_j)&=&\sum_{i=\xi_{r-2}}^{\xi_{r-1}}~\sum_{A_{r-1}\in\cA^{(r-1)}_0} \gamma_i^{A_{r-1}}  \nonumber  \\
	&=&\sum_{i=\xi_{r-2}}^{\xi_{r-1}} \gamma^{(r-1)}_i  \label{gamma(B)=c(D)-1} \\
	&=&f_{\eta}(\bm{\lambda})  \label{gamma(B)=c(D)-2} \\
	&=&\sum_{k=1}^{b_1}c(D_k),  \label{gamma(B)=c(D)-3}
	\end{eqnarray}
	where \eqref{gamma(B)=c(D)-1} follows from \eqref{gamma-subset-partition-property}, 
	\eqref{gamma(B)=c(D)-2} follows from \eqref{property-gamma-1}, 
	and \eqref{gamma(B)=c(D)-3} follows from \eqref{c(D_k)=f(eta)}. 
	This implies that we have a one-to-one correspondence between $f_{\eta}(\bm{\lambda})$ of $D_{k}$'s and $f_{\eta}(\bm{\lambda})$ of $B_j$'s. 
	The mapping defined by overlap in \reffig{fig-mapping} is a simple one-to-one correspondence.
	\begin{figure}[!t]
		\centering
		\begin{tikzpicture}[font=\scriptsize, scale=0.8]
		\draw[thick] (0,0.92) rectangle (9.3,1.8);
		\draw[thick] (0,0) rectangle (9.3,0.87);
		\draw (1.2,0.92)--(1.2,1.8) (2.2,0.92)--(2.2,1.8) (4.8,0.92)--(4.8,1.8) (5.5,0.92)--(5.5,1.8) (6.5,0.92)--(6.5,1.8) (7.3,0.92)--(7.3,1.8) (8.0,0.92)--(8.0,1.8);
		\draw (0.8,0)--(0.8,0.87) (1.5,0)--(1.5,0.87) (2.6,0)--(2.6,0.87) (3.2,0)--(3.2,0.87) (4.5,0)--(4.5,0.87) (6.8,0)--(6.8,0.87) (8.2,0)--(8.2,0.87);
		
		\node (lambda1) at (0.6,1.3) {$D_1$};
		\node (lambda2) at (1.7,1.3) {$D_2$};
		\node (lambda3) at (3.6,1.3) {$D_3$};
		\node (lambda4) at (5.2,1.3) {$D_4$};
		\node (lambda5) at (6.0,1.3) {$D_5$};
		\node (lambda6) at (6.9,1.3) {$D_6$};
		\node (lambda7) at (7.7,1.3) {$\cdots$};
		\node (lambdaL) at (8.6,1.3) {$D_{b_1}$};
		
		\node (A1) at (0.4,0.4) {$B_1$};
		\node (A2) at (1.2,0.4) {$B_2$};
		\node (A3) at (2.0,0.4) {$B_3$};
		\node (A4) at (2.9,0.4) {$B_4$};
		\node (A5) at (3.9,0.4) {$B_5$};
		\node (A6) at (5.6,0.4) {$B_6$};
		\node (dots) at (7.6,0.4) {$\cdots$};
		\node (Ab) at (8.9,0.4) {$B_{b_2}$};
		\end{tikzpicture}
		\caption{a one-to-one mapping} \label{fig-mapping}
	\end{figure}
	The inequality in \eqref{partition-condiiton} ensures that the number of $D_k$'s that contains $\xi_{r-1}$ is less than or equal to the number of $B_j$'s that don't contain $\xi_{r-1}$. 
	Thus, there exists a correspondence such that $B_j\cap D_k=\emptyset$ if $B_j$ and $D_k$ have overlap in \reffig{fig-mapping}. 
	Without loss of generality, assume the the mapping in \reffig{fig-mapping} is such a correspondence. 
	Let 
	\begin{equation}
	\cO=\big\{(j,k): B_j \text{ and }D_k \text{ have overlap in \reffig{fig-mapping}}\big\}.
	\end{equation}
	Then we have for all $(j,k)\in\cO$ that 
	\begin{equation}
	B_j\cap D_k=\emptyset
	\end{equation}
	and 
	\begin{equation}
	|B_j\cup D_k|=\eta.  \label{cardinality-B&D}
	\end{equation}
	
	For $k\in\{1,2,\cdots,b_1\}$, let $s_k=\sum_{i=1}^kc(D_i)$. 
	For $j\in\{1,2,\cdots,b_2\}$, let $t_j=\sum_{i=1}^j\gamma(B_i)$. 
	For $(j,k)\in\cO$, let $c(B_j,D_k)$ be the length overlap of $B_j$ and $D_k$ in \reffig{fig-mapping}, which is equal to
	\begin{equation}
	c(B_j,D_k)=
	\begin{cases}
	\gamma(B_j),& \text{if } s_{k-1}\leq t_{j-1}\leq t_j\leq s_k\\
	s_k-t_{j-1},& \text{if } s_{k-1}\leq t_{j-1}\leq s_k\leq t_j \\
	c(D_k),& \text{if } t_{j-1}\leq s_{k-1}\leq s_k\leq t_j \\
	t_j-s_{k-1},& \text{if } t_{j-1}\leq s_{k-1}\leq t_j\leq s_k \\
	0,& \text{otherwise}.
	\end{cases}  \label{def-c(B,D)}
	\end{equation}
	
	It is easy to check that for $k\in\{1,2,\cdots,b_1\}$,
	\begin{equation}
	\sum_{j=1}^{b_2}c(B_j,D_k)=c(D_k)   \label{sum-of-c(B,D)-1}
	\end{equation}
	and for $j\in\{1,2,\cdots,b_2\}$,
	\begin{equation}
	\sum_{k=1}^{b_1}c(B_j,D_k)=\gamma(B_j).  \label{sum-of-c(B,D)-2}
	\end{equation}
	Then we have 
	\begin{equation}
	\sum_{(j,k)\in\cO}c(B_j,D_k)=\sum_{k=1}^{b_1}c(D_k)=\sum_{j=1}^{b_2}\gamma(B_j)=f_{\eta}(\bm{\lambda}).   \label{sum-of-c(B,D)-3}
	\end{equation}
	
	The following lemma states the relation between the coefficients $c(B_j,D_k)$ and $\bm{\lambda}$. 
	The detailed proof of the lemma can be found in Appendix \ref{proof_lemma-check-resolution-lambda}.
	\begin{lemma}\label{lemma-check-resolution-lambda}
		For $i\in\cL$, we have 
		\begin{equation}
		\sum_{(j,k)\in\cO:~i\in B_j\cup D_k}c(B_j,D_k)\leq \lambda_i.
		\end{equation}
	\end{lemma}
	For any $\bm{v}\in\Omega_{L}^{\eta}$ and $\bm{v}=(v_1,v_2,\cdots,v_L)$, if $\{i:v_i=1\}=B_j\cup D_k$ for some $(j,k)\in\cO$, let $c_{\eta}(\bm{v})=c(B_j,D_k)$. 
	Otherwise, if there is no $(j,k)\in\cO$ such that $\{i:v_i=1\}=B_j\cup D_k$, let $c_{\eta}(\bm{v})=0$. 
	Then by \eqref{cardinality-B&D}, \eqref{sum-of-c(B,D)-3}, and \Cref{lemma-check-resolution-lambda}, 
	we can see that $\{c_{\eta}(\bm{v}):~\bm{v}\in\Omega_{L}^{\eta}\}$ is an optimal $\eta$-resolution for $\bm{\lambda}$. 
	
	For $i\in\{\xi_{r-1},\xi_{r-1}+1,\cdots,L\}$ and $j\in\{1,2,\cdots,b_2\}$, let 
	\begin{equation}
	c(\{i\}\cup B_j)=\sum_{k\in\{1,2,\cdots,b_1\}:~i\in D_k}c(B_j,D_k).  \label{def-c(i-Bj)}
	\end{equation}
	It is easy to check that 
	\begin{equation}
	\sum_{j=1}^{b_2} c(\{i\}\cup B_j)=\sum_{(j,k)\in\cO: ~i\in D_k}c(B_j,D_k)\leq \lambda^{(r-1)}_i
	\end{equation}
	and
	\begin{align}
	\sum_{i=\xi_{r-1}}^{L} c(\{i\}\cup B_j)&=\sum_{i=\xi_{r-1}}^{L}\sum_{k\in\{1,2,\cdots,b_1\}:~i\in D_k}c(B_j,D_k)  \nonumber \\
	&=\sum_{k=1}^{b_1}c(B_j,D_k)  \\
	&=\gamma(B_j).
	\end{align}
	
	\begin{remark}\label{remark-cof-c}
		In the example at the end of \Cref{section-DS-proof}, the parameter $c(B_j,D_k)$ is the coefficients in \eqref{converse-example-5}-\eqref{converse-example-7}, where for example, the coefficient $\frac{1}{3}$ of $\frac{1}{3}H(W_2W_3|W_1M_1M_2)$ in \eqref{converse-example-5} and $\frac{1}{3}H(W_1W_2W_3|M_1M_2)$ in \eqref{converse-example-6}-\eqref{converse-example-7} is $c(\{1\},\{2,3\})$. 
		The fact that $\{c(B_j,D_k):(j,k)\in\cO\}$ is an optimal $\eta$-resolution ensures us to proceed after the $\eta$-th iteration in the converse proof. 
		The parameter $c(\{i\}\cup B_j)$ is the coefficient in \eqref{converse-example-S1-1}-\eqref{converse-example-S1-3}, where for example,  $\frac{1}{3}$ of $\frac{1}{3}H(W_2|W_1M_{1:2})$ in \eqref{converse-example-S1-3} is $c(\{i\}\cup B_j)$ for $i=2$ and $B_j=\{1\}$; 
	\end{remark}
	
	Before proving the converse, we introduce two important relations that will be repeated used in the proof. 
	For $\alpha=1,2,\cdots,r-1$, and $i,j\in\cL$, $\cB\subseteq\cL$ such that $|\cB|=\alpha-1$ and $i,j\notin\cB$, we have 
	\begin{align}
	H(W_i|W_{\cB}M_{1:\alpha})&\geq H(W_i|W_jW_{\cB}M_{1:\alpha})   \nonumber \\
	&=H(W_i|W_{\{j\}\cup\cB}M_{1:\alpha}).   \label{pre-reverse-1}
	\end{align}
	and
	\begin{align}
	&H(W_i|W_{\cB}M_{1:\alpha-1})  \nonumber \\
	&=H(W_i|W_{\cB}M_{1:\alpha-1}M_{\alpha})+H(M_{\alpha}|W_{\cB}M_{1:\alpha-1})  \nonumber \\
	&\quad -H(M_{\alpha}|W_iW_{\cB}M_{1:\alpha-1})   \nonumber \\
	&=H(W_i|W_{\cB}M_{1:\alpha})+H(M_{\alpha})  \label{pre-reverse-2}
	\end{align}
	For notational simplicity, let $\xi_{-1}=0$. For $\alpha=1,2,\cdots,r-1$, let 
	\begin{align}
	I_\alpha\triangleq& \sum_{i=\xi_{\alpha-2}}^{\xi_{\alpha-1}}\sum_{A_{\alpha-1}\in\cA^{(\alpha-1)}_0}\gamma_i^{A_{\alpha-1}}H(W_iW_{A_{\alpha-1}}|M_{1:\alpha})    \nonumber \\
	&+\sum_{i=\xi_{r-1}}^{L}\sum_{j=1}^{b}c(\{i\}\cup B_j)H(W_i|W_{B_j^{\alpha-1}}M_{1:\alpha})    \nonumber \\
	&+\sum_{i=1}^{L}\left(\sum_{k=\alpha}^{r-1}\sum_{A_k\in\cA^{(k)}_0}\gamma_i^{A_k} H(W_i|W_{A_k^{\alpha-1}}M_{1:\alpha})\right). 
	\end{align}
	We have the following lemma which provides an iteration that is useful in the sequel. 
	The proof of the lemma can be found in Appendix \ref{proof_lemma-converse-iteration}.
	\begin{lemma}\label{lemma-converse-iteration}
		$I_\alpha\geq I_{\alpha+1}+\left[f_1(\bm{\lambda})-\alpha f_{\eta}(\bm{\lambda})\right] H(M_{\alpha+1})$ for $\alpha=1,2,\cdots,r-1$. 
	\end{lemma}
	
	We prove the converse of DS-SMDC (i.e., $\bm{\lambda}\cdot\bm{\sR}\geq g_{\eta}(\bm{\lambda})$ for all $\eta=r+1,r+2,\cdots,L+1$) as follows.
	\begin{align}
	\bm{\lambda}\cdot\bm{R}&=\lambda_1H(W_1)+\lambda_2H(W_2)+\cdots+\lambda_LH(W_L) \nonumber \\
	&=(\sum_{i=1}^L\lambda_i)H(M_1)+\sum_{i=1}^L\lambda_iH(W_i|M_1)  \label{converse-proof-1} \\
	&=f_1(\bm{\lambda})H(M_1)  \nonumber \\
	&\quad +\sum_{i=1}^{L}\left(\lambda_i^{(r-1)}+\sum_{k=1}^{r-1}\sum_{A_k\in\cA^{(k)}_0}\gamma_i^{A_k}\right)H(W_i|M_1)  \label{converse-proof-2} \\
	&\geq \sum_{\alpha=1}^{r}\left[f_1(\bm{\lambda})-(\alpha-1) f_{\eta}(\bm{\lambda})\right]H(M_{\alpha})   \nonumber \\
	&\quad +\sum_{i=\xi_{r-2}}^{\xi_{r-1}} \sum_{A_{r-1}\in\cA^{(r-1)}_0}\gamma_i^{A_{r-1}} H(W_iW_{A_{r-1}}|M_{1:r})  \nonumber \\
	&\quad +\sum_{i=\xi_{r-1}}^{L}\sum_{j=1}^{b_2}c(\{i\}\cup B_j)H(W_i|W_{B_j^{r-1}}M_{1:r})   \label{converse-proof-5}  \\
	&=\sum_{\alpha=1}^{r}\left[f_1(\bm{\lambda})-(\alpha-1) f_{\eta}(\bm{\lambda})\right]H(M_{\alpha})    \nonumber \\
	&\quad +\sum_{j=1}^{b_2}\gamma(B_j) H(W_{B_j}|M_{1:r})  \nonumber \\
	&\quad +\sum_{i=\xi_{r-1}}^{L}\sum_{j=1}^{b_2}c(\{i\}\cup B_j)H(W_i|W_{B_j}M_{1:r})   \label{converse-proof-6}  \\ 
	&\geq \sum_{\alpha=1}^{r}\left[f_1(\bm{\lambda})-(\alpha-1) f_{\eta}(\bm{\lambda})\right]H(M_{\alpha})    \nonumber \\
	&\quad +\sum_{(j,k)\in\cO}c(B_j,D_k) H(W_{B_j}|M_{1:r})  \nonumber \\
	&\quad +\sum_{(j,k)\in\cO}c(B_j,D_k)H(W_{D_k}|W_{B_j}M_{1:r})   \label{converse-proof-7}  \\ 
	&\geq \sum_{\alpha=1}^{r}\left[f_1(\bm{\lambda})-(\alpha-1) f_{\eta}(\bm{\lambda})\right]H(M_{\alpha})    \nonumber \\
	&\quad +\sum_{(j,k)\in\cO}c(B_j,D_k) H(W_{D_k}W_{B_j}|M_{1:r})  \label{converse-proof-8} \\
	&=\sum_{\alpha=1}^{r}\left[f_1(\bm{\lambda})-(\alpha-1) f_{\eta}(\bm{\lambda})\right]H(M_{\alpha})     \nonumber \\
	&\quad +\sum_{(j,k)\in\cO}c(B_j,D_k) H(W_{B_j\cup D_k} M_{r+1}^{\eta}|M_{1:r}) \label{converse-proof-9} \\
	&\geq \sum_{\alpha=1}^{r}\left[f_1(\bm{\lambda})-(\alpha-1) f_{\eta}(\bm{\lambda})\right]H(M_{\alpha})    \nonumber \\
	&\quad +f_{\eta}(\bm{\lambda}) \sum_{\alpha=r+1}^{\eta}H(M_{\alpha})+ \sum_{\bm{v}\in\Omega_{L}^{\eta}}c_{\eta}(\bm{v}) H(W_{\bm{v}}|M_{1:\eta})   \label{converse-proof-10} \\
	&\geq \sum_{\alpha=1}^{r}\left[f_1(\bm{\lambda})-(\alpha-1) f_{\eta}(\bm{\lambda})\right] H(M_{\alpha})    \nonumber \\
	&\quad + f_{\eta}(\bm{\lambda})\sum_{\alpha=r+1}^{\eta} H(M_{\alpha})+ \sum_{\alpha=\eta+1}^{L} f_{\alpha}(\bm{\lambda}) H(M_{\alpha})  \label{converse-proof-11} \\
	&=\sum_{\alpha=1}^{r}(f_1(\bm{\lambda})) H(M_{\alpha})+\sum_{\alpha=\eta+1}^{L}f_{\alpha}(\bm{\lambda})H(M_{\alpha})     \nonumber \\
	&\quad  +f_{\eta}(\bm{\lambda})\left[\sum_{\alpha=r+1}^{\eta}H(M_{\alpha}) -\sum_{\alpha=1}^{r}(\alpha-1) H(M_{\alpha})\right]  \nonumber  \\
	&=\sum_{\alpha=1}^{r}f_1(\bm{\lambda})m_\alpha+ \sum_{\alpha=\eta+1}^{L}f_{\alpha}(\bm{\lambda})m_\alpha     \nonumber \\
	&\quad + f_{\eta}(\bm{\lambda}) \left[\sum_{\alpha=r+1}^{\eta}m_\alpha-\sum_{\alpha=1}^{r}(\alpha-1)m_\alpha\right],  \label{converse-proof-13}
	\end{align} 
	where \eqref{converse-proof-2} follows from \eqref{lambda(i,r-1)}, 
	\eqref{converse-proof-5} follows by applying \Cref{lemma-converse-iteration} for $\alpha=1,2,\cdots,r-1$ successively, 
	\eqref{converse-proof-6} follows from the definition of $B_j$, 
	\eqref{converse-proof-10} follows from \eqref{sum-of-c(B,D)-3}, 
	\eqref{converse-proof-11} follows from the fact that $\{c_{\eta}(\bm{v}):~\bm{v}\in\Omega_{L}^{\eta}\}$ is an optimal $\eta$-resolution for $\bm{\lambda}$ and the iteration in the converse for SMDC in~\cite{yeung99}, 
	and \eqref{converse-proof-7} follows from \eqref{sum-of-c(B,D)-2} and
	\begin{align}
	&\sum_{i=\xi_{r-1}}^{L}\sum_{j=1}^{b_2}c(\{i\}\cup B_j)H(W_i|W_{B_j}M_{1:r})  \nonumber \\
	&=\sum_{j=1}^{b_2}\left[\sum_{i=\xi_{r-1}}^L c(\{i\}\cup B_j) H(W_i|W_{B_j}M_{1:r})\right]  \\
	&=\sum_{j=1}^{b_2}\left[\sum_{i=\xi_{r-1}}^L \sum_{k\in\{1,2,\cdots,b_1\}:~i\in D_k} \!\!\!\!\! c(B_j, D_k) H(W_i|W_{B_j}M_{1:r})\right]  \\
	&=\sum_{j=1}^{b_2}\left[\sum_{k=1}^{b_1}c(B_j, D_k) \left(\sum_{i\in D_k} H(W_i|W_{B_j}M_{1:r})\right)\right]  \\
	&\geq \sum_{j=1}^{b_2}\sum_{k=1}^{b_1}c(B_j, D_k)H(W_{D_k}|W_{B_j}M_{1:r})  \\
	&=\sum_{(j,k)\in\cO}c(B_j, D_k)H(W_{D_k}|W_{B_j}M_{1:r}). 
	\end{align}
	Dividing both sides of \eqref{converse-proof-13} by $a$, we obtain by the definition of $g_{\eta}(\bm{\lambda})$ in \eqref{def-g} that for any $\eta=r+1,r+2,\cdots,L+1$, 
	\begin{equation}
	\sum_{l=1}^L\lambda_l(\sR_l+\epsilon)\geq g_{\eta}(\bm{\lambda}). 
	\end{equation}
	Letting $\epsilon\rightarrow0$, the inequality $\bm{\lambda}\cdot\bm{\sR}\geq g_{\eta}(\bm{\lambda})$ is proved. 
	\begin{remark}\label{remark-connection}
		The step-by-step correspondence between the general proof in \eqref{converse-proof-1}-\eqref{converse-proof-13} and the example in \eqref{converse-example-1}-\eqref{converse-example-8} is as follows: 
		\begin{itemize}
			\item The iteration in \eqref{converse-proof-5} is the generalization of the step in \eqref{converse-example-5}; 
			\item The transform of conditional entropies in \eqref{converse-proof-6}-\eqref{converse-proof-10} play the same role as \eqref{converse-example-6}; 
			\item The application of the $\alpha$-resolution technique in \eqref{converse-proof-11} is the generalization of that in \eqref{converse-example-7}. 
		\end{itemize}
	\end{remark}

	\section{Proof of \Cref{lemma-considerable-lambda}}\label{proof_lemma-considerable-lambda}
	Let $\bm{\lambda}'=(\lambda_1',\lambda_2',\cdots,\lambda_L')$, where
	\begin{equation}
	\lambda_i'=\lambda_i,\text{ for all }i=2,3,\cdots,L  \label{lambda-genaral-2}
	\end{equation}
	and 
	\begin{equation}
	\lambda_1'=\frac{\lambda_2'+\lambda_3'+\cdots+\lambda_L'}{\eta-1}.  \label{lambda-genaral-3}
	\end{equation}
	By Lemma 7 in \cite{guo-yeung-SMDC-IT20}, \eqref{lambda-general-1} implies that
	\begin{equation}
	f_{\eta}(\bm{\lambda})=f_{\eta-1}(\lambda_2,\lambda_3,\cdots,\lambda_L),
	\end{equation}
	and similarly, from \eqref{lambda-genaral-3},
	\begin{equation}
	f_{\eta}(\bm{\lambda}')=f_{\eta-1}(\lambda_2,\lambda_3,\cdots,\lambda_L).
	\end{equation}
	Thus, we have
	\begin{equation}
	f_{\eta}(\bm{\lambda})=f_{\eta}(\bm{\lambda}'),
	\end{equation}
	which by Lemma 5 in \cite{guo-yeung-SMDC-IT20} implies that 
	\begin{equation}
	f_{\alpha}(\bm{\lambda})=f_{\alpha}(\bm{\lambda}'),\text{ for all }\eta\leq \alpha\leq L.
	\end{equation}
	The rate constraint $\bm{\lambda}'\cdot\bm{\sR}\geq g_{\eta}(\bm{\lambda}')$ is the following,
	\begin{align}
	\bm{\lambda}'\cdot\bm{\sR}&\geq \sum_{\alpha=1}^{r}f_1(\bm{\lambda}') \sm_\alpha+\sum_{\alpha=\eta+1}^{L}f_{\alpha}(\bm{\lambda}')\sm_\alpha   \nonumber \\
	&\quad +f_{\eta}(\bm{\lambda}') \left[\sum_{\alpha=r+1}^{\eta}\sm_\alpha- \sum_{\alpha=1}^{r}(\alpha-1)\sm_\alpha\right].   \label{check-constraint-lambda'}
	\end{align}
	This implies
	\begin{align}
	\bm{\lambda}\cdot\bm{\sR}&=\bm{\lambda}'\cdot\bm{\sR}+(\lambda_{1}-\lambda_{1}') \sR_1  \nonumber \\
	&\geq \sum_{\alpha=1}^{r}f_1(\bm{\lambda}') \sm_\alpha+ \sum_{\alpha=\eta+1}^{L}f_{\alpha}(\bm{\lambda}')\sm_\alpha   \nonumber \\
	&\qquad + f_{\eta}(\bm{\lambda}') \left[\sum_{\alpha=r+1}^{\eta}\sm_\alpha- \sum_{\alpha=1}^{r}(\alpha-1)\sm_\alpha\right]  \nonumber \\
	&\qquad +(\lambda_{1}-\lambda_{1}')\left(\sum_{\alpha=1}^{r}\sm_\alpha\right)  \\
	&=\sum_{\alpha=1}^{r}f_1(\bm{\lambda}) \sm_\alpha+ \sum_{\alpha=\eta+1}^{L}f_{\alpha}(\bm{\lambda})\sm_\alpha  \nonumber \\
	&\qquad +f_{\eta}(\bm{\lambda}) \left[\sum_{\alpha=r+1}^{\eta}\sm_\alpha- \sum_{\alpha=1}^{r}(\alpha-1)\sm_\alpha\right],
	\end{align}
	which is exactly the constraint $\bm{\lambda}\cdot\bm{\sR}\geq g_{\eta}(\bm{\lambda})$. This proves the lemma.
	
	\section{Proof of \Cref{lemma-f>lambda1}}\label{proof_lemma-f>lambda1}
	For any $\bm{\lambda}\in\mathbb{R}_{\eta}^L$, we have 
	\begin{equation}
	\lambda_1\leq \frac{1}{\eta-1}\sum_{i=2}^L\lambda_i,  \label{pf-f>lambda1-1}
	\end{equation}
	which by Lemma 1 in \cite{guo-yeung-SMDC-IT20} implies that 
	\begin{equation}
	f_{\eta}(\bm{\lambda})=\frac{1}{\eta}\sum_{i=1}^{L}\lambda_i.
	\end{equation}
	It is easy to check that \eqref{pf-f>lambda1-1} is equivalent to
	\begin{equation}
	\lambda_1\leq \frac{1}{\eta}\sum_{i=1}^L\lambda_i.
	\end{equation}
	Thus, we have $f_{\eta}(\bm{\lambda})\geq \lambda_1$, which proves the lemma.
	
	\section{Proof of \Cref{lemma-last-choice}}\label{proof_lemma-last-choice}
	For $\alpha=1,2,\cdots,r-1$, we have 
	\begin{equation}
	\sum_{i=\xi_{\alpha}}^{L}\lambda_i^{(\alpha)}=\sum_{i=1}^L\lambda_i-\alpha f_{\eta}(\bm{\lambda})\geq (\eta-\alpha) f_{\eta}(\bm{\lambda}).
	\end{equation}
	where the inequality follows from \eqref{f(1)-vs-f(eta)} and the fact that $f_1(\bm{\lambda})=\sum_{i=1}^L\lambda_i$. In particular, for $\alpha=r-1$, 
	\begin{equation}
	\sum_{i=\xi_{r-1}}^{L}\lambda_i^{(r-1)}\geq \left[\eta-(r-1)\right] f_{\eta}(\bm{\lambda}).
	\end{equation}
	Denote the ordered permutation of $\bm{\lambda}^{(r-1)}$ by $\tilde{\bm{\lambda}}^{(r-1)}=\left(\tilde{\lambda}_{\xi_{r-1}}^{(r-1)}, \tilde{\lambda}_{\xi_{r-1}+1}^{(r-1)}, \cdots, \tilde{\lambda}_{L}^{(r-1)}\right)$. 
	Then from \eqref{f(eta)-vs-lambda}, we obtain 
	\begin{equation}
	\frac{1}{\eta-(r-1)}\sum_{i=\xi_{r-1}}^{L}\tilde{\lambda}_i^{(r-1)}\geq f_{\eta}(\bm{\lambda})\geq \tilde{\lambda}_{\xi_{r-1}}^{(r-1)},   \label{f-f-lambda}
	\end{equation}
	which implies 
	\begin{equation}
	\tilde{\lambda}_{\xi_{r-1}}^{(r-1)}\leq \frac{1}{\left[\eta-(r-1)\right]-1}\sum_{i=\xi_{r-1}+1}^{L}\tilde{\lambda}_i^{(r-1)}.
	\end{equation}
	By Lemma 4 and Lemma 7 in \cite{yeung99}, this implies that $\tilde{\bm{\lambda}}^{(r-1)}$ has a perfect $\left[\eta-(r-1)\right]$-resolution (c.f. Appendix~\ref{referenced-lemmas/theorems}) and 
	\begin{align}
	f_{\eta-(r-1)}\left(\tilde{\bm{\lambda}}^{(r-1)}\right)=\frac{1}{\eta-(r-1)}\sum_{i=\xi_{r-1}}^{L}\tilde{\lambda}_i^{(r-1)}.
	\end{align}
	From Lemma 2 in \cite{guo-yeung-SMDC-IT20} and \eqref{f-f-lambda}, this implies that 
	\begin{equation}
	f_{\eta-(r-1)}\left(\bm{\lambda}^{(r-1)}\right)=f_{\eta-(r-1)}\left(\tilde{\bm{\lambda}}^{(r-1)}\right)\geq f_{\eta}(\bm{\lambda}).
	\end{equation}
	This proves the lemma.
	
	\section{Proof of \Cref{lemma-check-resolution-lambda}}\label{proof_lemma-check-resolution-lambda}
	Consider the following five cases where the set $\cL$ is partitioned into five subsets.
	\begin{enumerate}[i.]
		\item For $i\in\{1,2,\cdots,\xi_{r-2}-1\}$, we have 
		\begin{align}
		\sum_{(j,k)\in\cO: ~i\in B_j}c(B_j,D_k)&=\sum_{i\in B_j}\sum_{k=1}^{b_1}c(B_j,D_k)  \nonumber \\
		&=\sum_{i\in B_j}\gamma(B_j)  \label{check-resolution-lambda-1-1} \\
		&=\sum_{k=\xi_{\alpha-1}}^{\xi_{\alpha}}\sum_{A_{\alpha}\in\cA^{(\alpha)}_0:~i\in A_{\alpha}}\gamma_k^{A_{\alpha}}  \nonumber  \\
		&=\lambda_i, \label{check-resolution-lambda-1-3}
		\end{align}
		where \eqref{check-resolution-lambda-1-1} follows from \eqref{sum-of-c(B,D)-2} and \eqref{check-resolution-lambda-1-3} follows from \eqref{gamma-subset-property-2}. 
		
		\item For $i=\xi_{r-2}$, it follows from \eqref{gamma-subset-property-5} that
		\begin{align}
		&\sum_{(j,k)\in\cO: ~i\in B_j}c(B_j,D_k)  \nonumber \\
		&=\sum_{k=\xi_{r-2}}^{\xi_{r-1}}~\sum_{A_{r-1}\in\cA^{(r-1)}_0:~\xi_{r-2}\in A_{r-1}}\gamma_k^{A_{r-1}}  \nonumber \\
		&\quad +\sum_{A_{r-1}\in\cA^{(r-1)}_0} \gamma_{\xi_{r-2}}^{A_{r-1}}   \nonumber \\
		&=\gamma_{\xi_{r-2}}^{(r-2)}+\gamma_{\xi_{r-2}}^{(r-1)}   \nonumber \\
		&=\lambda_{\xi_{r-2}}.  \label{check-resolution-lambda-2}
		\end{align}
		
		\item For $i\in\{\xi_{r-2}+1,\xi_{r-2}+2,\cdots,\xi_{r-1}-1\}$, we have 
		\begin{align}
		\sum_{(j,k)\in\cO: ~i\in B_j}c(B_j,D_k)&=\sum_{A_{r-1}\in\cA^{(r-1)}_0}\gamma_i^{A_{r-1}}   \nonumber  \\
		&=\gamma_i^{(r-1)}  \label{check-resolution-lambda-3-1} \\
		&=\lambda_i, \label{check-resolution-lambda-3-2}
		\end{align}
		where \eqref{check-resolution-lambda-3-1} follows from \eqref{gamma-subset-partition-property} 
		and \eqref{check-resolution-lambda-3-2} follows from \eqref{def-gamma(i,alpha)}.
		
		\item For $i=\xi_{r-1}$,
		\begin{align}
		&\sum_{(j,k)\in\cO: ~i\in B_j}c(B_j,D_k)  \nonumber \\
		&=\sum_{A_{r-1}\in\cA^{(r-1)}_0} \gamma_{\xi_{r-1}}^{A_{r-1}}+\sum_{\xi_{r-1}\in D_k}\sum_{j=1}^{b_2}c(B_j,D_k)   \nonumber \\
		&=\sum_{A_{r-1}\in\cA^{(r-1)}_0} \gamma_{\xi_{r-1}}^{A_{r-1}}+\sum_{\xi_{r-1}\in D_k}c(D_k)   \label{check-resolution-lambda-5-1} \\
		&\leq \gamma_{\xi_{r-1}}^{(r-1)}+\gamma_{\xi_{r-1}}^{(r)}   \label{check-resolution-lambda-5-2} \\
		&=\lambda_{\xi_{r-1}},  \label{check-resolution-lambda-5-3}
		\end{align}
		where \eqref{check-resolution-lambda-5-1} follows from \eqref{sum-of-c(B,D)-1}, 
		\eqref{check-resolution-lambda-5-2} follows from \eqref{gamma-subset-partition-property} and the fact that $\{c(D_k):k=1,2,\cdots,b_1\}$ is an $[\eta-(r-1)]$-resolution for $\bm{\lambda}^{(r-1)}$, 
		and \eqref{check-resolution-lambda-5-3} follows from \eqref{def-gamma(i,alpha)}. 
		
		\item For $i\in\{\xi_{r-1}+1,\xi_{r-1}+2,\cdots,L\}$, we have 
		\begin{align}
		\sum_{(j,k)\in\cO: ~i\in D_k}c(B_j,D_k)&=\sum_{i\in D_k}\sum_{j=1}^{b_2}c(B_j,D_k)  \nonumber \\
		&=\sum_{i\in D_k}c(D_k)  \label{check-resolution-lambda-4-1} \\
		&\leq \lambda_i,  \label{check-resolution-lambda-4-2}
		\end{align}
		where \eqref{check-resolution-lambda-4-1} follows from \eqref{sum-of-c(B,D)-1} 
		and \eqref{check-resolution-lambda-4-2} follows from the fact that $\{c(D_k):k=1,2,\cdots,b_1\}$ is an $[\eta-(r-1)]$-resolution for $\bm{\lambda}^{(r-1)}$.
	\end{enumerate}
	
	\section{Proof of \Cref{lemma-converse-iteration}}\label{proof_lemma-converse-iteration}
	For $\alpha=1,2,\cdots,r-1$, we have the following iteration,
	\begin{align}
	I_\alpha&=\sum_{i=\xi_{\alpha-2}}^{\xi_{\alpha-1}}\sum_{A_{\alpha-1}\in\cA^{(\alpha-1)}_0}\gamma_i^{A_{\alpha-1}}H(W_iW_{A_{\alpha-1}}|M_{1:\alpha})    \nonumber \\
	&\quad +\sum_{i=\xi_{r-1}}^{L}\sum_{j=1}^{b}c(\{i\}\cup B_j)H(W_i|W_{B_j^{\alpha-1}}M_{1:\alpha})  \nonumber \\ 
	&\quad +\sum_{i=1}^{L}\left(\sum_{k=\alpha}^{r-1}\sum_{A_k\in\cA^{(k)}_0}\gamma_i^{A_k} H(W_i|W_{A_k^{\alpha-1}}M_{1:\alpha})\right)  \label{iteration-1}  \\
	&=\sum_{i=\xi_{\alpha-2}}^{\xi_{\alpha-1}}\sum_{A_{\alpha-1}\in\cA^{(\alpha-1)}_0}\gamma_i^{A_{\alpha-1}}H(W_iW_{A_{\alpha-1}}|M_{1:\alpha})  \nonumber \\
	&\quad +\sum_{i=\xi_{\alpha-1}}^{\xi_{\alpha}} \sum_{A_{\alpha}\in\cA^{(\alpha)}_0} \gamma_i^{A_{\alpha}} H(W_i|W_{A_{\alpha}}M_{1:\alpha})    \nonumber \\
	&\quad +\sum_{i=\xi_{r-1}}^{L}\sum_{j=1}^{b}c(\{i\}\cup B_j)H(W_i|W_{B_j^{\alpha-1}}M_{1:\alpha})    \nonumber \\
	&\quad +\sum_{i=1}^{L}\left(\sum_{k=\alpha+1}^{r-1}\sum_{A_k\in\cA^{(k)}_0}\gamma_i^{A_k} H(W_i|W_{A_k^{\alpha-1}}M_{1:\alpha})\right) \label{iteration-2} \\
	&=\sum_{i=\xi_{\alpha-1}}^{\xi_{\alpha}}\sum_{A_{\alpha}\in\cA^{(\alpha)}_0}\gamma_i^{A_{\alpha}}H(W_{A_{\alpha}}|M_{1:\alpha})    \nonumber \\
	&\quad +\sum_{i=\xi_{\alpha-1}}^{\xi_{\alpha}} \sum_{A_{\alpha}\in\cA^{(\alpha)}_0} \gamma_i^{A_{\alpha}} H(W_i|W_{A_{\alpha}}M_{1:\alpha})    \nonumber \\
	&\quad +\sum_{i=\xi_{r-1}}^{L}\sum_{j=1}^{b}c(\{i\}\cup B_j)H(W_i|W_{B_j^{\alpha-1}}M_{1:\alpha})    \nonumber \\
	&\quad +\sum_{i=1}^{L}\left(\sum_{k=\alpha+1}^{r-1}\sum_{A_k\in\cA^{(k)}_0}\gamma_i^{A_k} H(W_i|W_{A_k^{\alpha-1}}M_{1:\alpha})\right)  \label{iteration-3} \\
	&=\sum_{i=\xi_{\alpha-1}}^{\xi_{\alpha}}\sum_{A_{\alpha}\in\cA^{(\alpha)}_0}\gamma_i^{A_{\alpha}}H(W_iW_{A_{\alpha}}|M_{1:\alpha})    \nonumber \\
	&\quad +\sum_{i=\xi_{r-1}}^{L}\sum_{j=1}^{b}c(\{i\}\cup B_j)H(W_i|W_{B_j^{\alpha-1}}M_{1:\alpha})    \nonumber \\
	&\quad +\sum_{i=1}^{L}\left(\sum_{k=\alpha+1}^{r-1}\sum_{A_k\in\cA^{(k)}_0}\gamma_i^{A_k} H(W_i|W_{A_k^{\alpha-1}}M_{1:\alpha})\right)  \label{iteration-4} \\
	&\geq \sum_{i=\xi_{\alpha-1}}^{\xi_{\alpha}}\sum_{A_{\alpha}\in\cA^{(\alpha)}_0}\gamma_i^{A_{\alpha}}H(W_iW_{A_{\alpha}}|M_{1:\alpha+1})    \nonumber \\
	&\quad +\sum_{i=\xi_{r-1}}^{L}\sum_{j=1}^{b}c(\{i\}\cup B_j)H(W_i|W_{B_j^{\alpha}}M_{1:\alpha})    \nonumber \\
	&\quad +\sum_{i=1}^{L}\left(\sum_{k=\alpha+1}^{r-1}\sum_{A_k\in\cA^{(k)}_0}\gamma_i^{A_k} H(W_i|W_{A_k^{\alpha}}M_{1:\alpha})\right)  \label{iteration-5} \\
	&=\sum_{i=\xi_{\alpha-1}}^{\xi_{\alpha}}\sum_{A_{\alpha}\in\cA^{(\alpha)}_0}\gamma_i^{A_{\alpha}}H(W_iW_{A_{\alpha}}|M_{1:\alpha+1})     \nonumber \\
	&\quad +\left[f_1(\bm{\lambda})-\alpha f_{\eta}(\bm{\lambda})\right] H(M_{\alpha+1})   \nonumber \\
	&\quad +\sum_{i=\xi_{r-1}}^{L}\sum_{j=1}^{b}c(\{i\}\cup B_j)H(W_i|W_{B_j^{\alpha}}M_{1:\alpha+1})    \nonumber \\
	&\quad +\sum_{i=1}^{L}\left(\sum_{k=\alpha+1}^{r-1}\sum_{A_k\in\cA^{(k)}_0}\gamma_i^{A_k} H(W_i|W_{A_k^{\alpha}}M_{1:\alpha+1})\right) \label{iteration-6} \\
	&=I_{\alpha+1}+\left[f_1(\bm{\lambda})-\alpha f_{\eta}(\bm{\lambda})\right] H(M_{\alpha+1}),
	\end{align}
	where \eqref{iteration-3} follows from 
	\begin{align}
	&\sum_{i=\xi_{\alpha-2}}^{\xi_{\alpha-1}}\sum_{A_{\alpha-1}\in\cA^{(\alpha-1)}_0}\gamma_i^{A_{\alpha-1}}H(W_iW_{A_{\alpha-1}}|M_{1:\alpha})  \\ 
	&=\sum_{i=\xi_{\alpha-2}}^{\xi_{\alpha-1}}\sum_{A_{\alpha-1}\in\cA^{(\alpha-1)}_0}\sum_{j=\xi_{\alpha-1}}^{\xi_{\alpha}}\gamma_j^{\{i\}\cup A_{\alpha-1}}H(W_{\{i\}\cup A_{\alpha-1}}|M_{1:\alpha})   \\ 
	&=\sum_{j=\xi_{\alpha-1}}^{\xi_{\alpha}}\sum_{i=\xi_{\alpha-2}}^{\xi_{\alpha-1}}\sum_{A_{\alpha-1}\in\cA^{(\alpha-1)}_0}\gamma_j^{\{i\}\cup A_{\alpha-1}}H(W_{\{i\}\cup A_{\alpha-1}}|M_{1:\alpha})   \\ 
	&=\sum_{j=\xi_{\alpha-1}}^{\xi_{\alpha}}\sum_{A_{\alpha}\in\cA^{(\alpha)}_0}\gamma_j^{A_{\alpha}}H(W_{A_{\alpha}}|M_{1:\alpha}),
	\end{align}
	\eqref{iteration-5} follows from the fact that conditioning does not increase entropy, and \eqref{iteration-6} follows from \eqref{pre-reverse-2}.

\end{appendices}
\bibliographystyle{ieeetr}
\bibliography{guotao-sMDC_ref}

\begin{IEEEbiographynophoto}
	{Tao Guo} (S'16--M'19) received his B.E. degree in Telecommunications Engineering from Xidian University in 2013, and the Ph.D. degree from the Department of Information Engineering, The Chinese University of Hong Kong in 2018. He was a Postdoctoral Research Associate at the Department of Electrical and Computer Engineering, Texas A\&M University from 2018 to 2020. He is currently a Postdoctoral Scholar in the Department of Electrical and Computer Engineering at the University of California, Los Angeles. His research interests include information theory and it applications to security and privacy, multi-user source coding, and coding for distributed storage systems.
\end{IEEEbiographynophoto}

\begin{IEEEbiographynophoto}
	{Chao Tian} (S'00--M'05--SM'12) received the B.E. degree in Electronic Engineering from Tsinghua University, Beijing, China, in 2000 and the M.S. and Ph. D. degrees in Electrical and Computer Engineering from Cornell University, Ithaca, NY in 2003 and 2005, respectively. Dr. Tian was a postdoctoral researcher at Ecole Polytechnique Federale de Lausanne (EPFL) from 2005 to 2007, a member of technical staff--research at AT\&T Labs--Research in New Jersey from 2007 to 2014, and an Associate Professor in the Department of Electrical Engineering and Computer Science at the University of Tennessee Knoxville from 2014 to 2017. He joined the Department of Electrical and Computer Engineering at Texas A\&M University in 2017. His research interests include data storage systems, multi-user information theory, joint source-channel coding, signal processing, and compute algorithms.
	
	Dr. Tian received the Liu Memorial Award at Cornell University in 2004, AT\&T Key Contributor Award in 2010, 2011 and 2013. His authored and co-authored papers received the 2014 IEEE ComSoc DSTC Data Storage Best Paper Award and the 2017 IEEE Jack Keil Wolf ISIT Student Paper Award. He was an Associate Editor for {\sc the IEEE Signal Processing Letters} from 2012 to 2014, and is currently an Editor for {\sc the IEEE Transactions on Communications} and an Associate Editor for {\sc the IEEE Transactions on Information Theory}.
\end{IEEEbiographynophoto}

\begin{IEEEbiographynophoto}
	{Tie Liu} received his B.S. (1998) and M.S. (2000) degrees, both in Electrical Engineering, from Tsinghua University, Beijing, China and a second M.S. degree in Mathematics (2004) and a Ph.D. degree in Electrical and Computer Engineering (2006) from the University of Illinois at Urbana-Champaign. Since August 2006 he has been with Texas A\&M University, where he is currently a Professor in the Department of Electrical and Computer Engineering. His primary research interest is in the area of information and statistical learning theory.
	
	Dr. Liu received an M. E. Van Valkenburg Graduate Research Award (2006) from the University of Illinois at Urbana-Champaign, a CAREER Award (2009) from the National Science Foundation, and an Outstanding Professor Award from Texas A\&M University (2018). He was a Technical Program Committee Co-Chair for the 2008 IEEE GLOBECOM, a General Co-Chair for the 2011 IEEE North American School of Information Theory, and an Associate Editor for Shannon Theory for the IEEE Transactions on Information Theory during 2014-2016.
\end{IEEEbiographynophoto}

\begin{IEEEbiographynophoto}
	{\bf Raymond W. Yeung} (S'85-M'88-SM'92-F'03) was born in Hong Kong on June 3, 1962.  He received the B.S., M.Eng., and Ph.D.\ degrees in electrical engineering from Cornell University, Ithaca, NY, in 1984, 1985, and 1988, respectively.
	
	He was on leave at Ecole Nationale Sup\'{e}rieure des T\'{e}l\'{e}communications, Paris, France, during fall 1986.  He  was a Member of Technical Staff of AT\&T Bell Laboratories from 1988 to 1991. Since 1991, he has been with The Chinese University of Hong Kong, where he is now Choh-Ming Li Professor of Information Engineering and Co-Director of Institute of Network Coding.  
	He has held visiting positions at Cornell University, Nankai University, the University of Bielefeld, the University of Copenhagen, Tokyo Institute of Technology, Munich University of Technology, and Columbia University.  He was a consultant in a project of Jet Propulsion Laboratory, Pasadena, CA, for salvaging
	the malfunctioning Galileo Spacecraft and a consultant for NEC, USA. His 25-bit synchronization marker was used onboard the Galileo Spacecraft for 
	image synchronization.
	
	His research interests include information theory and network coding. He is the author of the textbooks {\em A First Course in Information Theory} (Kluwer Academic/Plenum 2002) and its revision {\em Information Theory and Network Coding} (Springer 2008), which have been adopted by over 100 institutions around the world.  This book has also been published in Chinese (Higher Education Press 2011, translation by Ning Cai~{\em et~al.}). He also co-authored with Shenghao Yang the monograph {\em BATS Codes: Theory and Applications} (Morgan \& Claypool Publishers, 2017). In spring 2014, he gave the first MOOC on information theory that reached over 25,000 students.
	
	Dr.\ Yeung was a member of the Board of Governors of the IEEE Information Theory Society from 1999 to 2001.  He has served on the committees of a number of information theory symposiums and workshops.  He was General Chair of the First and the Fourth Workshops on Network, Coding, and Applications (NetCod 2005
	and 2008), a Technical Co-Chair for the 2006 IEEE International Symposium on Information Theory, a Technical Co-Chair
	for the 2006 IEEE Information Theory Workshop (Chengdu, China),
	and a General Co-Chair of the 2015 IEEE
	International Symposium on Information Theory.
	He currently serves as an Editor-at-Large
	of {\em Communications in Information and Systems},
	an Editor of {\em Foundation and Trends in Communications and Information
		Theory} and of {\em Foundation and Trends in Networking}, 
	and was an Associate Editor for Shannon Theory of the {\em IEEE Transactions
		on Information Theory} from 2003 to 2005.
	In 2011-12, he serves as a Distinguished Lecturer of the IEEE Information Theory Society.
	
	He was a recipient of the Croucher Foundation Senior 
	Research Fellowship for 2000/2001,
	the Best Paper Award (Communication Theory) of the 
	2004 International Conference on Communications, Circuits
	and System, 
	the 2005 IEEE Information Theory Society Paper Award, the Friedrich Wilhelm Bessel Research Award of the Alexander von Humboldt Foundation in 2007, 
	the 2016 IEEE Eric E. Sumner Award
	(``for pioneering contributions to the field of network coding"), and the 2018 ACM SIGMOBILE Test-of-Time Paper Award.
	In 2015, he was named (together with Zhen Zhang) an Outstanding Overseas Chinese Information Theorist by the China 
	Information Theory Society.
	In 2019, his team won a Gold Medal with Congratulations of the Jury at the 47th International Exhibition of Inventions of Geneva for their invention ``BATS: Enabling the Nervous System of Smart Cities." 
	He is a Fellow of the IEEE, Hong Kong Academy of Engineering Sciences, and Hong Kong Institution of Engineers.
\end{IEEEbiographynophoto}

\end{document}